\pdfoutput=1
\documentclass[a4paper,USenglish]{lipics-v2021}

\hideLIPIcs
\nolinenumbers

\usepackage{amsmath,amssymb,amsfonts}
\usepackage{algpseudocode}
\usepackage{graphicx}
\usepackage{pseudo}
\usepackage{textcomp}
\usepackage{xcolor}
\usepackage{theoremref}
\usepackage{todonotes}
\usepackage{comment}
\usepackage{mathrsfs}
\usepackage{tikz}
\usepackage{comment}
\usepackage{amsmath, amsthm, amssymb}
\usepackage{graphicx}
\usepackage{enumerate}
\usepackage{thmtools}
\usepackage{thm-restate}
\usepackage{enumitem}

\usetikzlibrary{arrows,shapes,decorations.pathmorphing,automata,backgrounds,petri}
\usetikzlibrary{positioning,calc}
\usetikzlibrary{patterns}
\tikzstyle{legendborder}=[rectangle, draw, black, rounded corners, thin, top color=white, text=black, minimum width=2.5cm, text width=4.5cm]
\tikzstyle{legendnoborder}=[rectangle, draw, white, rounded corners, thin, top color=white, text=black, minimum width=2.5cm]
\tikzstyle{selected edge} = [draw,line width=1pt,black]
\usetikzlibrary{shapes,backgrounds,plothandlers,plotmarks,calc,arrows,fadings,decorations.pathreplacing,decorations.pathmorphing,}

\tikzstyle{v}=[circle,fill=black,draw=black!75,inner sep=0pt,minimum size=0.3em]
\tikzstyle{I}=[circle,draw=black!75,inner sep=0pt,minimum size=0.8em]
\tikzstyle{J}=[rectangle,draw=black!75,inner sep=0pt,minimum size=0.7em]
\tikzstyle{vertex}=[circle,inner sep=2,minimum size =2mm,semithick,fill=white!80!blue, draw=black]

\newcommand{\mc}{\mathcal}

\newcommand{\WO}{\textsf{W[1]}}

\newcommand{\PPSPACE}{\textsf{PSPACE}}
\newcommand{\NPP}{\textsf{NP}}
\newcommand{\FPT}{\textsf{FPT}}

\newcommand{\A}{\mathcal{A}}
\newcommand{\B}{\mathcal{B}}
\newcommand{\C}{\mathcal{C}}
\newcommand{\D}{\mathcal{D}}
\newcommand{\E}{\mathcal{E}}
\newcommand{\F}{\mathcal{F}}

\newcommand{\R}{\mathcal{R}}
\newcommand{\T}{\mathcal{T}}
\renewcommand{\S}{\mathcal{S}}
\newcommand{\pwtw}{w}
\newcommand{\fvs}{f}
\newcommand{\degen}{d}
\newcommand{\numb}{r}
\newcommand{\struct}{s}

\newcommand{\Q}{\textsc{Multi-Tape-Rec}}
\newcommand{\AutRec}{\textsc{Tape-Rec}}
\newcommand{\SyncAutRec}{\textsc{Sync-Tape-Rec}}
\newcommand{\SyncPathAutRec}{\textsc{Sync-Path-Tape-Rec}}
\newcommand{\SyncQ}{\textsc{Sync-Multi-Tape-Rec}}
\newcommand{\PathAutRec}{\textsc{Path-Tape-Rec}}
\newcommand{\TSDSR}{\textsc{TS-DSR}}
\newcommand{\TJDSR}{\textsc{TJ-DSR}}

\newcommand{\ie}{i.e., }

\newcommand{\dsr}{\textsc{TS-DSR}}
\newcommand{\dscr}{\textsc{TS-DCR}}

\usepackage{comment}
\title{The tape reconfiguration problem and its consequences for dominating set reconfiguration}

\titlerunning{The tape reconfiguration problem} 



\author{Nicolas {Bousquet}}{CNRS, INSA Lyon, UCBL, LIRIS, UMR5205, F-69622 Villeurbanne, France  \and CNRS - Université de Montréal CRM - CNRS \and \url{https://perso.liris.cnrs.fr/nbousquet/} }{nicolas.bousquet@cnrs.fr}{https://orcid.org/0000-0003-0170-0503}{The author has been partly supported by ANR project ENEDISC (ANR-24-CE48-7768-01) and by Campus-France PHC Cedre (ProLo).}

\author{Quentin {Deschamps}}{CNRS, INSA Lyon, UCBL, LIRIS, UMR5205, F-69622 Villeurbanne, France}{quentin.deschamps@liris.cnrs.fr}{}{}

\author{Arnaud {Mary}}{Université de Lyon, Université Lyon 1, CNRS, Laboratoire de Biométrie et Biologie Evolutive UMR
5558, 69622 Villeurbanne, France \and
ERABLE team, Inria Lyon, Villeurbanne, France}{arnaud.mary@univ-lyon1.fr}{https://orcid.org/0000-0002-8201-227X}{}

\author{Amer E.~Mouawad}{Department of Computer Science, American University of Beirut, Beirut, Lebanon \and \url{https://www.aub.edu.lb/pages/profile.aspx?MemberId=aa368}}{aa368@aub.edu.lb}{}{}

\author{Théo {Pierron}}{CNRS, INSA Lyon, UCBL, LIRIS, UMR5205, F-69622 Villeurbanne, France \and \url{https://perso.liris.cnrs.fr/tpierron/}}{theo.pierron@liris.cnrs.fr}{}{The author has been partly supported by ANR projects ENEDISC (ANR-24-CE48-7768-01) and P-GASE (ANR-21-CE48-0001-01).}

\authorrunning{N. Bousquet, Q. Deschamps, A. Mary, A.~E.~Mouawad, T. Pierron}

\Copyright{Nicolas Bousquet, Quentin Deschamps, Arnaud Mary, Amer E.~Mouawad, Théo Pierron}

\date{}

\begin{document}

\maketitle

\begin{abstract}
A dominating set of a graph $G=(V,E)$ is a set of vertices $D \subseteq V$ whose closed neighborhood is $V$, i.e., $N[D]=V$.
We view a dominating set as a collection of tokens placed on the vertices of $D$. 
In the token sliding variant of the \textsc{Dominating Set Reconfiguration} problem (TS-DSR), we seek to transform a source dominating set into a target dominating set in $G$ by sliding tokens along edges, and while maintaining a dominating set all along the transformation.

TS-DSR is known to be \PPSPACE-complete even restricted to graphs of pathwidth $\pwtw$, for some non-explicit constant $\pwtw$ and to be \textsf{XL}-complete parameterized by the size $k$ of the solution. The first contribution of this article consists in using a novel approach to provide the first explicit constant for which the TS-DSR problem is \PPSPACE-complete, a question that was left open in the literature. 

From a parameterized complexity perspective, the token jumping variant of DSR, i.e., where tokens can jump to arbitrary vertices, is known to be \textsf{FPT} when parameterized by the size of the dominating sets on nowhere dense classes of graphs. But, in contrast, no non-trivial result was known about TS-DSR. We prove that DSR is actually much harder in the sliding model since it is \textsf{XL}-complete when  restricted to bounded pathwidth graphs and even when parameterized by $k$ plus the feedback vertex set number of the graph.
This gives, for the first time, a difference of behavior between the complexity under token sliding and token jumping for some problem on graphs of bounded treewidth. 
All our results are obtained using a brand new method, based on the hardness of the so-called \textsc{Tape Reconfiguration} problem, a problem we believe to be of independent interest. 

We complement these hardness results with a positive result showing that DSR (parameterized by $k$) in the sliding model is \textsf{FPT} on planar graphs, also answering an open problem from the literature.
\end{abstract}

\section{Introduction}

The \emph{combinatorial reconfiguration} framework aims to investigate algorithmic and structural aspects of the solution space of an underlying \emph{base problem}. For instance, given an instance of some base problem along with two feasible solutions, i.e., the source and target feasible solutions, the goal is to determine if (and in how many steps) one can transform the source to the target via a sequence of adjacent feasible solutions. Such a sequence is called a \emph{reconfiguration sequence} and every step in the sequence (going from one solution to an adjacent one) is called a \emph{reconfiguration step}. 
Reconfiguration problems arise in various fields such as combinatorial games, motion of robots, random sampling, and enumeration. It has been extensively studied for various rules and types of (base) problems such as satisfiability~\cite{gopalan2009connectivity,DBLP:journals/siamdm/MouawadNPR17}, graph colorings \cite{bonamy2019conjecture,CerecedaHJ11}, vertex covers and independent sets  \cite{HearnD05,ito2011complexity,lokshtanov2018complexity} and matchings~\cite{BonamyBHIKMMW19}. The reader is referred to the surveys~\cite{Nishimura17,van2013complexity,DBLP:journals/csr/BousquetMNS24} for a more complete overview of the field.

An alternative view of reconfiguration problems is via the notion of (re)configuration graphs. Let $\Pi$ be a base problem and let $\mathcal{I}$ be an instance of $\Pi$. The \emph{configuration graph} $\mathcal{R}_\mathcal{I}$ is a graph whose nodes correspond to feasible solutions of $\mathcal{I}$ and in which two solutions are adjacent if and only if we can transform the first into the second in a single reconfiguration step. 
In this paper, we focus on the reachability question that asks, given two solutions $S,S'$ of $\mathcal{I}$, whether there exists a reconfiguration sequence from $S$ to $S'$. In other words, does there exist a path between $S$ and $S'$ in the configuration graph. Other works have focused on different problems such as the connectivity of the configuration graph or the diameter of its connected components; see, e.g.~\cite{bousquet2019polynomial,gopalan2009connectivity}.

A \emph{dominating set} of a graph $G=(V,E)$ is a set of vertices $D \subseteq V$ such that $N[D]=V$. That is, the union of the closed neighborhoods of vertices of $D$ \emph{spans} $V$. The \textsc{Dominating Set} (DS) problem, i.e., deciding whether a graph contains a dominating set of size at most $k$, is one of the fundamental \textsf{NP}-complete problems~\cite{DBLP:conf/coco/Karp72}. By now, both the classical and parameterized complexity of the problem are very well understood especially when restricted to sparse graph classes and even when considering various possible parameterizations~\cite{DBLP:conf/stacs/DrangeDFKLPPRVS16,DBLP:conf/fsttcs/DawarK09,DBLP:conf/swat/AlberBFN00,DBLP:journals/algorithmica/AlonG09}. In particular, and relevant to this work, \textsc{Dominating Set} is fixed-parameter tractable (\textsf{FPT}) parameterized by the size of the solution in nowhere-dense classes of graphs~\cite{Dawar09}.

In the rest of this section,  we informally state our results and put them into context. In Section~\ref{subsec:tr}, we introduce the main problem from which we will derive our results and state the hardness result we obtain for it. In Section~\ref{subseq:consDSR}, we explain the consequences of this result for \textsc{Dominating Set Reconfiguration} (DSR) and in Section~\ref{subseq:pos}, we state our main positive results.

\subparagraph*{Token jumping vs. token sliding.}
The main focus of this paper is to study the parameterized complexity of reconfiguration problems restricted to sparse classes of graphs and in particular the case of \textsc{Dominating Set Reconfiguration}, or DSR for short.
There are three different natural adjacency relations one can consider between dominating sets, which gave rise to three different models in the literature; namely the token jumping (TJ) model, the token sliding (TS) model, and the token addition/removal (TAR) model. Since the TAR model is known to be polynomially equivalent to the TJ model~\cite{DBLP:journals/dam/BonamyDO21}, we only discuss the other two.

In the \emph{token jumping} model, we say that two solutions $S,S'$ are \emph{adjacent} if both $S \setminus S'$ and $S'\setminus S$ have size at most one. In other words, we can remove a vertex of $S$ and add a vertex of $S'$ to transform $S$ into $S'$. In the token view of the problem, we say a token \emph{jumps} from some vertex $u$ to another (possibly already occupied) vertex $v$. In the \emph{token sliding} model, we additionally require that $uv \in E$, i.e., the vertices $u$ and $v$ must be adjacent, and the token is said to \emph{slide} on the edge $uv$ from $u$ to $v$. We shall use TS-$\Pi$ and TJ-$\Pi$ to denote the reconfiguration variant of problem $\Pi$ under the token sliding and token jumping model, respectively. For instance, TS-DSR corresponds to DSR in the token sliding model and TJ-DSR corresponds to DSR in the jumping model. 

TJ-DSR is known to be \PPSPACE-complete on split graphs, bipartite graphs~\cite{DBLP:journals/dam/BonamyDO21,haddadan2016complexity,kaminski2012complexity}, and planar graphs of maximum degree $3$~\cite{HearnD05} and bounded-bandwidth graphs~\cite{DBLP:journals/jcss/Wrochna18}\footnote{Bandwidth is a very restricted subclass of pathwidth which, in turn, is a restriction of treewidth. That is bandwidth$(G) \ge $ pathwidth$(G) \ge $ treewidth$(G)$.}. On the positive side, linear-time algorithms are known for trees, interval graphs, and cographs. As for TS-DSR, the problem is known to be \PPSPACE-complete on circle graphs~\cite{DBLP:conf/fct/BousquetJ21}, split graphs~\cite{DBLP:journals/dam/BonamyDO21}, bipartite graphs~\cite{DBLP:journals/dam/BonamyDO21}, and planar bounded-bandwidth graphs of maximum degree three~\cite{DBLP:journals/jcss/Wrochna18}. Polynomial-time algorithms for TS-DSR are known for circular-arc graphs, dually chordal graphs, and cographs~\cite{DBLP:journals/tcs/HaddadanIMNOST16,DBLP:journals/dam/BonamyDO21,DBLP:conf/fct/BousquetJ21,DBLP:journals/corr/abs-2310-00241}. 

Even though both TS-DSR and TJ-DSR are known to be \PPSPACE-complete  restricted to instances of constant bandwidth/pathwidth/treewidth, an exact constant above which the problems become hard is not explicit in the hardness proofs of Wrochna~\cite{DBLP:journals/jcss/Wrochna18}. Determining an explicit upper bound for which the problems are hard has been left open for almost a decade now, see e.g.~\cite{DBLP:journals/jcss/BartierBM23,DBLP:conf/fct/BousquetJ21,bartier2021}.\footnote{The question is open for a wealth of reconfiguration problems including \textsc{Independent Set Reconfiguration}, \textsc{Vertex Cover Reconfiguration}, \textsc{Shortest Path Reconfiguration}, and \textsc{Dominating Set Reconfiguration}.}
Our first main result consists in giving an explicit upper bound on the treewidth and even the pathwidth above which the TS-DSR problem becomes \PPSPACE-complete. Our proof uses a completely different approach from the one of Wrochna~\cite{DBLP:journals/jcss/Wrochna18} and is based on a new problem, namely \textsc{Tape Reconfiguration} (\AutRec{}), we introduce and prove to also be hard 
(formal definitions in Section~\ref{subsec:tr}).

\begin{theorem}\label{thm:twintro}
\TSDSR{} is \PPSPACE-complete even when restricted to graphs of treewidth (resp. pathwidth) at most $12$ (resp. $18$).
\end{theorem}

The existence of explicit constants $b$ and $\pwtw$ for which (i) TS-DSR is \PPSPACE{}-complete in graphs of bandwidth $b$, and (ii) TJ-DSR is \PPSPACE{}-complete in graphs of pathwidth or treewidth at most $\pwtw$ is also open. Unfortunately, our proof technique does not yield such explicit constants. Moreover, it is very unlikely that a small variation of our proof technique can provide such constants since this would almost automatically imply that the problems (parameterized by dominating set size $k$) are also \textsf{XL}-complete when restricted to the aforementioned graph classes, which is in contradiction with known results~\cite{DBLP:journals/csr/BousquetMNS24}. This follows, in part, from the fact that our reductions construct instances in which the dominating set size is negligible with respect to $n$.

We note that the constants $12$ and $18$ are probably not optimal. In particular, the complexity status of both TJ-DSR and TS-DSR in outerplanar graphs (which are graphs of treewidth at most $2$) is still open.
We will discuss in more detail the proof technique in Section~\ref{sec:roadmap}, but note that the proof technique of Theorem~\ref{thm:twintro} is \textit{drastically} different from the one of Wrochna~\cite{DBLP:journals/jcss/Wrochna18} since the size of the dominating set is linear in the size of the graph in~\cite{DBLP:journals/jcss/Wrochna18} while it is independent from $n = |V(G)|$ in our proof. This, in turn, allows us to derive the non-existence of parameterized algorithms which was completely out of reach with the techniques of~\cite{DBLP:journals/jcss/Wrochna18}.

\subparagraph*{Parameterized complexity and the mysterious case of token sliding.}
A systematic study of the parameterized complexity of
reconfiguration problems was initiated by Mouawad et al.~\cite{DBLP:journals/algorithmica/MouawadN0SS17}. This was followed by a long sequence of results trying to push the tractability boundary of many reconfiguration problems (see~\cite{DBLP:journals/csr/BousquetMNS24} for a survey which mainly focuses on dominating set and independent set reconfiguration results). 

On general graphs,
TJ-DSR and TS-DSR are \textsf{W[2]}-complete parameterized by the dominating sets size $k$ plus the length of a reconfiguration sequence~$\ell$~\cite{DBLP:journals/algorithmica/MouawadN0SS17,DBLP:conf/iwpec/BodlaenderGS21}. For the parameter $k$ alone, it was shown by Bodlaender et al.~\cite{DBLP:conf/iwpec/BodlaenderGS21} that both problems are \textsf{XL}-complete\footnote{\textsf{XL} is a complexity class that contains the \textsf{W}-hierarchy and which naturally contains most of the reconfiguration problems. In particular, when a problem is \textsf{XL}-complete it is very unlikely that it is \textsf{FPT}. See Section~\ref{sec:definition} for a formal definition.} if $\ell$ is not given, \textsf{XNL}-complete if $\ell$ is given in binary, and \textsf{XNLP}-complete if $\ell$ is given in unary as part of the input. The constructions of Bodlaender et al.~\cite{DBLP:conf/iwpec/BodlaenderGS21} result in heavily dense graphs that contain all possible subgraphs $H$ as subgraphs. In particular, it was left as an open problem whether the token sliding and token jumping variants of the problems behave the same on $H$-free graphs\footnote{A graph $G$ is \emph{$H$-free} if it does not contain $H$ as an \emph{induced subgraph}. That is, there is no subset of vertices in $G$ whose induced subgraph is isomorphic to $H$.
A graph $G$ is \emph{$H$-minor-free} if it does not contain $H$ as a \emph{minor}. That is, $H$ cannot be formed from $G$ by deleting vertices and edges and by contracting edges.}. 

When parameterized by $\ell$ alone, both TJ-DSR and TS-DSR are \textsf{FPT} on any class of graphs where first-order model-checking is \textsf{FPT}, e.g., nowhere dense classes~\cite{DBLP:journals/jacm/GroheKS17} and classes of bounded twinwidth~\cite{DBLP:journals/jacm/BonnetKTW22} (assuming a contraction sequence is given as part of the input). Indeed, for the parameter $\ell$, we can find a fixed sentence expressing the existence of a reconfiguration sequence of length~$\ell$.

Unfortunately, the approach via first-order model-checking does not help us deal with the parameterization by $k$ since reconfiguration sequences might be arbitrarily long compared to $k$. It is not possible to state the existence of such a sequence in first-order logic with a bounded length sentence. 
Instead, the key tool to tackle TJ-DSR parameterized by $k$ alone is based on the notion of domination cores~\cite{DBLP:conf/fsttcs/DawarK09}. 
In fact, the complexity of TJ-DSR parameterized by $k$ is quite well understood for sparse classes of graphs. The problem is known to be \textsf{FPT} parameterized by $k$ for biclique-free graphs and semi-ladder-free graphs~\cite{DBLP:journals/jcss/LokshtanovMPRS18,DBLP:conf/stacs/FabianskiPST19} (classes encompassing nowhere dense and degenerate graphs). In particular, it restricts quite a lot the type of graph on which TJ-DSR might be \textsf{XL}-complete.

On the other hand, the parameterized complexity of TS-DSR remains open even when restricted to very simple graph classes such as bounded pathwidth or treewidth graphs\footnote{Note that the parameterized complexity status of TS-DSR on bounded bandwidth graphs is trivially \textsf{FPT} since the number of vertices is upper bounded by a function of the bandwidth and the domination number.} or planar graphs. This difference of knowledge between the sliding and jumping variants is also observed for other problems such as \textsc{Independent Set Reconfiguration} (ISR) where the token jumping version is known to be FPT (parameterized by $k$) on $K_{\ell,\ell}$-free graphs while its sliding counterpart is only known to be FPT on planar graphs and is open beyond~\cite{DBLP:journals/csr/BousquetMNS24}. Bodlaender et al.~\cite{DBLP:conf/iwpec/BodlaenderGS21} mentioned that ``\emph{it would be interesting to investigate for which graph classes switching between token jumping and token sliding does affect the parameterized complexities}''. The second goal of this paper is to study this question via the lens of \textsc{Dominating Set Reconfiguration}.

In the token sliding case, apart from the few polynomial results on very restricted classes, no \textsf{FPT} algorithm parameterized by $k$ is known. The existence of such an algorithm is conjectured in several papers including~\cite{DBLP:journals/csr/BousquetMNS24,DBLP:journals/algorithmica/LokshtanovMPS22,DBLP:journals/jcss/BartierBM23} or the recent survey~\cite{DBLP:journals/csr/BousquetMNS24}.
The second main contribution of this article makes a step towards closing this gap by proving the following theorem:

\begin{theorem}\label{th:fvs}
TS-DSR parameterized by $k$ is \textsf{XL}-complete when restricted to graphs of treewidth at most $12$ and pathwidth at most $18$. Moreover, it remains \textsf{XL}-hard when parameterized by $k$ plus the feedback vertex set number. 
\end{theorem}

A \emph{feedback vertex} is a subset of vertices whose deletion leaves an acyclic graph. The \emph{feedback vertex set number} is the minimum size of a feedback vertex set. 
The first hardness result of Theorem~\ref{th:fvs} (for pathwidth and treewidth) is again a consequence of the hardness of the \AutRec{} problem.  
The second hardness result is also based on the  hardness of \AutRec{} but the proof is more technical since we need to maintain a small feedback vertex set number. One can wonder if we can strengthen the result of Theorem~\ref{th:fvs} by replacing feedback vertex set number by vertex cover number. The answer is negative.
A \emph{vertex cover} is a subset of vertices whose deletion leaves an independent set. Recall that given a simple undirected graph $G$, a set of vertices $I \subseteq V$ is an \emph{independent set} if the vertices of $I$ are pairwise non-adjacent. The \emph{vertex cover number} is the minimum size of a vertex cover. We note that TS-DSR parameterized by $k$ plus the vertex cover number is trivially \textsf{FPT}. This follows from the fact that given a vertex cover $C$, we can partition $R = V \setminus C$ into classes such that two vertices belong to the same class if and only if they have the same neighbors in $C$. All classes consist of independent sets (since $C$ is a vertex cover) and any class containing more than $k$ vertices can be reduced since no more than $k$ vertices can be used from a class (and all vertices of a class are ``equivalent''). Hence, we obtain an instance with at most $|C| + 2^{|C|}k$ vertices (assuming no isolated vertices). 

Theorem~\ref{th:fvs} answers, as mentioned before, a problem that has been left open in several previous articles~\cite{DBLP:journals/csr/BousquetMNS24,DBLP:journals/algorithmica/LokshtanovMPS22,DBLP:journals/jcss/BartierBM23,DBLP:journals/csr/BousquetMNS24}. 
The result ensures that, even in very restricted sparse graph classes, TS-DSR and TJ-DSR behave completely differently. To our knowledge, it is the first time the token sliding and token jumping models behave differently for a reconfiguration problem on a class of bounded treewidth (and even nowhere dense)\footnote{For sparse graphs, such a difference has already been observed for bounded degeneracy graphs for instance~\cite{DBLP:journals/csr/BousquetMNS24}}. 
Actually, it is, as far as we know, the first reconfiguration problem that is not \textsf{FPT} on bounded treewidth graphs.
It also ensures that TS-DSR and TS-ISR behave very differently on nowhere dense classes of graphs (since an easy applications of the lemmas of~\cite{DBLP:journals/jcss/BartierBM23} ensures that TS-DSR is \textsf{FPT} parameterized by $k$ plus the feedback vertex set number). 
 Moreover, our results provide the first (and infinite) collection of graphs (any supergraph of a large enough biclique) for which TS-DSR is \textsf{XL}-complete parameterized by $k$ while TJ-DSR is \textsf{FPT}, which partially answers the  question of~\cite{DBLP:conf/iwpec/BodlaenderGS21}. 
 
To sum up, our results make progress in three directions by showing that:
\begin{enumerate}
    \item TS and TJ do not necessarily behave the same on graphs of bounded treewidth (and below);
    \item ISR and DSR do not behave the same on nowhere dense graph classes; and
    \item TJ-DSR and TS-DSR behave differently on infinitely many $H$-free graphs.
\end{enumerate}

We complement our negative results by positive results proving that  \TSDSR{} is \textsf{FPT} on planar graphs (and actually beyond), answering a problem left open in~\cite{DBLP:journals/csr/BousquetMNS24}.

\begin{theorem}
    \TSDSR{} parameterized by $k$ is \textsf{FPT} on planar graphs.
\end{theorem}

To obtain the positive result, we adapt the strategy used in the token jumping model via domination cores to handle the case where the size of a forbidden minor is not too large (namely $K_{4,\ell}$-minor-free graphs). However, while the proof is almost direct in the token jumping variant whenever domination cores exist, the proof gets more technical here as connectivity matters. Note that this method cannot be generalized much further since we prove that TS-DSR is \textsf{XL}-complete for graphs of treewidth at most $12$, and when parameterized by $k$ plus feedback vertex set number. However, it would be interesting to understand the limit between tractability and intractability on $H$-minor free graphs.

\subparagraph*{Reconfiguration of connected dominating sets.} 
Our proof techniques are actually strong enough to be generalized further to the \textsc{Connected Dominating Set Reconfiguration} problem (CDSR). Recall that a dominating set $D$ is a {\em connected dominating set} if the graph induced by $D$ is connected. Deciding whether a given graph contains a connected dominating set of size at most $k$ is known to be \textsf{NP}-complete~\cite{DBLP:conf/coco/Karp72} and \textsf{W[2]}-complete parameterized by~$k$~\cite{DBLP:journals/siamcomp/DowneyF95}. 

TS-CDSR and TJ-CDSR, the corresponding reconfiguration variants, are known to be \textsf{PSPACE}-complete (even on graphs of bounded pathwidth) and \textsf{W[1]}-hard parameterized by $k+\ell$~\cite{DBLP:journals/algorithmica/LokshtanovMPS22,DBLP:conf/iwpec/BodlaenderGS21}. Moreover, TJ-CDSR parameterized by $k+\ell$ remains \textsf{W[1]}-hard even when restricted to $5$-degenerate graphs. On the positive side, TJ-CDSR parameterized by $k$ is known to be \textsf{FPT} on planar graphs~\cite{DBLP:journals/algorithmica/LokshtanovMPS22}.

TJ-CDSR parameterized by $k$ has been conjectured to be \textsf{FPT} on every nowhere dense class of graphs in several papers, see e.g.~\cite{DBLP:journals/csr/BousquetMNS24,DBLP:journals/algorithmica/LokshtanovMPS22}. 
In~\cite{DBLP:journals/algorithmica/LokshtanovMPS22}, the authors ask whether TJ-CDSR and TS-CDSR parameterized by $k$ are \textsf{FPT} on graphs of bounded pathwidth. We provide a negative answer to all of the aforementioned questions. 
Our proof techniques can also be adapted for other dominating set variants such as distance $d$ dominating sets and total dominating sets which have been studied in the literature. 

\subparagraph*{Related work.} 
Deciding whether a given graph contains an independent set of size at least $k$ is known to be \NPP-complete~\cite{DBLP:conf/coco/Karp72} and  \WO-complete parameterized by (solution size)~$k$~\cite{DBLP:journals/siamcomp/DowneyF95}. 

The \textsc{Independent Set Reconfiguration} (ISR) problem is 
another central problem that has been very widely studied under the combinatorial reconfiguration framework. 
Both TJ-ISR and TS-ISR are known to be \textsf{PSPACE}-complete restricted to planar graphs of bounded bandwidth~\cite{DBLP:journals/jcss/Wrochna18,DBLP:conf/iwpec/Zanden15} and \textsf{XL}-complete when parameterized by the size of independent sets~$k$~\cite{DBLP:conf/iwpec/BodlaenderGS21}.

Similarly to TJ-DSR, the complexity of the token jumping variant, i.e., TJ-ISR, is rather well understood. Indeed, Ito et al.~\cite{DBLP:conf/isaac/ItoKO14} proved that TJ-ISR parameterized by $k$ is \textsf{FPT} on planar graphs. This result has been generalized to nowhere dense and degenerate classes in~\cite{DBLP:journals/jcss/LokshtanovMPRS18,DBLP:conf/isaac/AgrawalHM22} and even to $K_{d,d}$-free graphs in~\cite{BousquetMP17}. The idea of these proofs is to follow a two-step strategy. First, an ``important'' subset $X$ of vertices of bounded size is identified. Second, the remaining vertices are classified according to their neighborhood in $X$, and it is shown that if some class is too large then it can be reduced.

Again, the situation becomes a lot less clear when we consider the token sliding model, i.e., TS-ISR. For instance, while TJ-ISR is trivial on chordal graphs~\cite{kaminski2012complexity}, TS-ISR is \PPSPACE-complete on split graphs~\cite{DBLP:journals/mst/BelmonteKLMOS21}. Lokshtanov
and Mouawad~\cite{lokshtanov2018complexity}  showed that, in bipartite graphs, TJ-ISR is \textsf{NP}-complete while
TS-ISR remains \PPSPACE-complete. 

Just like DSR, even after a decade of research, very little is known about the parameterized complexity of ISR under the token sliding model. It was proved in~\cite{BartierBDLM21} that TS-ISR is \textsf{FPT} on bipartite $C_4$-free graphs. The result was later generalized by Bartier et al.~\cite{DBLP:journals/jcss/BartierBM23} on planar graphs and graphs of bounded maximum degree and graphs of girth at least $5$~\cite{BartierBHMS24}.  One can easily use the reduction rules introduced in~\cite{DBLP:journals/jcss/BartierBM23} in order to prove that TS-ISR is \textsf{FPT} parameterized by $k$ plus the feedback vertex set number of the graph, which shows that the behavior of TS-DSR and TS-ISR is very different on sparse graph classes.

\section{Technical overview of our results}

\subsection{Tape Reconfiguration}\label{subsec:tr}

Let $\Sigma$ be an alphabet. The elements of $\Sigma$ are called \emph{letters}.  A \emph{$\Sigma$-tape} (or tape for short when $\Sigma$ is clear from context) is a graph where each vertex, called a \emph{cell}, is labeled with a subset of $\Sigma$, called its \emph{content}. Moreover, each tape comes with two distinguished vertices, called the \emph{start} and \emph{end} cell, respectively. We will move a token, called \emph{(read) head}, on each tape.

Let $\mathcal{T} = \{T_1,\ldots,T_p\}$ be a set of $\Sigma$-tapes.
For $i\in[1,p]$,\footnote{$[1,p]$ denotes the set of integers between $1$ and $p$.} let $c_i$ be a cell of $T_i$, and \emph{start}$_i$ (resp. \emph{end}$_i$) be the start (resp. end) cell of $T_i$. We say that $(c_1,\ldots,c_p)$ forms a \emph{valid configuration} if the union of their contents is $\Sigma$.
Let $C=(c_1,\ldots,c_p)$ and $C'=(c_1',\ldots,c_p')$ be two valid configurations. We say that there is an \emph{elementary transformation} (or a \emph{(reconfiguration) step}) between $C$ and $C'$ if they differ on exactly one element, say the $j$-th, and $c_jc'_j$ is an edge in~$T_j$.
\medskip

\noindent
\textsc{Tape Reconfiguration (\AutRec)} \\
\textbf{Input:} An alphabet $\Sigma$ and a set $\mathcal{T}$ of $\Sigma$-tapes, an initial and a final configuration $C_s,C_t$. \\
\textbf{Parameter:} Size of the alphabet $\Sigma$ plus the number of $\Sigma$-tapes.\\
\textbf{Output:} Yes, if it is possible to transform the start configuration (\ie $(start_1,\ldots,start_p)$) into the end configuration (\ie $(end_1,\ldots,end_p)$) via a sequence of elementary transformations while keeping a valid configuration all along.
\medskip

We argue in Appendix~\ref{sec-count-tapes-reduce} that one can always assume the number of $\Sigma$-tapes is bounded by the size of the alphabet $|\Sigma|$. Consequently, one can assume $|\Sigma|$ to be the parameter.
The problem can be equivalently defined as follows. We are given a set of $p$ graphs, each with a read head positioned on one vertex. 
In one step, we are permitted to slide a read head along an edge within any one of these graphs.
Note that since we only allow the movement of tokens along edges we can assume that these graphs are connected.
We will mostly show hardness results for this problem, and a natural question is about how the structure of the input graphs affects its complexity.
As it turns out, even the simplest case -- captured by the variant \PathAutRec{}, which requires all tapes to be paths whose endpoints are the start and end cells -- is already hard.
In particular, we show the following:
    
\begin{theorem}
\AutRec{} is \textsf{XL}-complete.
\end{theorem}

\subsection{Consequences for \TSDSR{} (and \textsc{TS-CDSR})}\label{subseq:consDSR}

We prove that, for DSR, the results for TJ have not been generalized to TS for a good reason; \TSDSR{} is hard even on very restricted graph classes. Namely, we show the following:

\begin{restatable}{theorem}{thmTSDSR}
\label{thm:TSDSR}
\TSDSR{} and \textsc{TS-CDSR} are  \PPSPACE-complete and \textsf{XL}-complete parameterized by $k$ plus the feedback vertex set number, even restricted to $7$-degenerate graphs of treewidth at most $12$ and pathwidth at most $18$.
\end{restatable}

Note that, as far as we know, this is the first ``natural'' reconfiguration problem which becomes hard parameterized by $k$ plus the feedback vertex set number of the graph. Very few problems are hard for that parameter since allowing the feedback vertex set number to be part of the parameter enables the use of very powerful algorithmic techniques (given the simple structure of graphs having a feedback vertex set of bounded size).

With minor modifications to the reduction, we obtain the following slightly weaker result for \textsc{TJ-CDSR}:

\begin{restatable}{theorem}{thmCDSR}
\label{thm:CDSR}
    \textsc{TJ-CDSR} is \PPSPACE-complete and \textsf{XL}-complete when parameterized by $k$ plus the feedback vertex set number, even restricted to $5$-degenerate graphs of treewidth at most $13$ and pathwidth at most $19$.
\end{restatable}

Theorems~\ref{thm:TSDSR} and~\ref{thm:CDSR} imply that both TS-DSR and TJ-CDSR are very hard problems, even though the connectivity constraints are quite different in both; the latter is a global condition on the dominating set while the former only wants to ensure  local connectivity when moving along edges. In addition, both problems are much harder than TJ-DSR which is known to be \textsf{FPT} parameterized by $k$ on nowhere dense, degenerate, biclique-free, and semi-ladder-free graphs.

\subsection{Positive results for \TSDSR{}}\label{subseq:pos}

We complement our hardness results with positive ones. In particular, we show that TS-DSR is \textsf{FPT} on planar graphs. We first prove that the following holds:

\begin{restatable}{theorem}{thmFPTplanar}
\label{thm:FPT-planar}
\dsr{} parameterized by $k$ is \textsf{FPT} on $K_{3,d}$-free graphs. 
\end{restatable}

As a by-product of Theorem~\ref{thm:FPT-planar}, \dsr{} (parameterized by $k$) is \FPT{} on planar graphs. The proof technique is inspired from~\cite{DBLP:conf/isaac/ItoKO14,DBLP:journals/jcss/LokshtanovMPRS18} and consists in looking at neighborhood classes in a domination core and proving that all these classes can be reduced to have bounded size. It was proven in~\cite{eiben2019lossy} that a domination core of size $\mc{O}(k^{d+1})$ exists (and can be efficiently computed) even in $K_{d,d}$-free graphs (having a dominating set of size $k$).

\begin{restatable}{theorem}{thmFPTminor}
\label{thm:FPT-minor}
\dsr{} parameterized by $k$ is \textsf{FPT} on $K_{4,d}$-minor-free graphs.
\end{restatable}

To prove Theorem~\ref{thm:FPT-minor}, we start as in the proof of Theorem~\ref{thm:FPT-planar}. Although it can be seen as rather incremental compared to the previous result, reducing neighborhood classes is actually much trickier and requires more general arguments. In some steps of the proof we have to use the fact that dominating the domination core ensures that the whole graph is dominated to show that some subsets of vertices have to intersect with the dominating set (which in turn permits reducing some large enough parts of the graph). 


\section{Hardness results roadmap~\label{sec:roadmap}}

\subsection{Synchronized version and width of tape reconfiguration problems}

The goal of this section is to state our main hardness results and give outlines of their proofs. In order to do so, we will need to define some auxiliary problems along the way, which we also believe to be of independent interest.  

The first auxiliary problems we will need are synchronized versions of \AutRec{} and \PathAutRec{} where we assume that all the read heads move at the same speed. 
We say that a tape is $\numb$-\emph{numbered} if each cell is labeled with some integer of $[1,\numb]$ (the same integer possibly occurring on several cells), where the number of adjacent cells differs by at most $1$ modulo $\numb$. A configuration is numbered with $j$ if all the read heads are on a cell numbered with $j$. The read heads are \emph{synchronized} if, for every pair of tapes, the numbers of the cells under their read heads differ by at most $1$ modulo $\numb$. In other words, if the read head in the tape $T$ is at a cell numbered $q$, then the read head in $T'$ is on a cell numbered with $q$ or $q \pm 1\mod \numb$. A transformation is synchronized if the read heads are synchronized at each step. A configuration of a synchronized instance is \emph{valid} if read heads on the synchronized tapes are synchronized and the union of the contents of the reading heads is the whole alphabet.

The synchronized version of \AutRec{} is defined as follows: \medskip

\noindent
\textsc{\SyncAutRec} \\
\textbf{Input:} An alphabet $\Sigma$, a set $\T$ of $\numb$-numbered tapes and two numbered configurations $C_s,C_t$. \\
\textbf{Parameter:} $|\Sigma|+|\T|$. \\
\textbf{Output:} Yes if and only if there exists a synchronized transformation from $C_s$ to $C_t$.
\medskip

We similarly define the restriction \PathAutRec{} when tapes are paths whose endpoints are the start and end cells, and its synchronized version \SyncPathAutRec{} where, for technical reasons, we assume that the start cells are all numbered with $1$ and the numbering is non-decreasing along each tape.

Observe that the synchronization makes the problem much simpler to solve in some configurations. For instance, consider instances of \PathAutRec{} where cells of paths are labeled with their position in the path. Then, one can easily upper bound the number of valid configurations by $2^{|\T|} \cdot n$ valid configurations (where $n$ is the size of the longest tape). In particular, the following holds:

\begin{remark}
    In the particular setting where cells of paths are labeled with their position in the path, \SyncPathAutRec{} is \textsf{FPT}. 
\end{remark}

Note that we do not know if \SyncPathAutRec{} lies in \textsf{P} in that setting.

\subparagraph*{A more general problem.}
We will consider a harder problem whose flavor is close to \PathAutRec, denoted by \Q{}, and prove that this problem is \textsf{W[$*$]}-hard. The idea consists in allowing to make choices to construct a valid instance of \PathAutRec{} by selecting tapes among some tuples. More formally, we define the problem \Q{} as follows: \medskip
\\
\Q{} \\
\textbf{Input:} An alphabet $\Sigma$, $k$ tuples $\T^1,\ldots,\T^k$ of path $\Sigma$-tapes.\\
\textbf{Parameter:} $k+|\Sigma|$\\
\textbf{Output:} Yes if and only if there exists $i_1,\ldots,i_k$ such that $\{\T_{i_j}^j\mid j\in[1,k]\}$ is a positive instance of \PathAutRec. 
\medskip

There is a trivial reduction from \PathAutRec{} to \Q; by considering an instance of \Q{} where each tuple has size $1$. More formally, $\{T_1,\ldots,T_k\}$ is a positive instance of \PathAutRec{} if and only if the $k$ tuples $(T_1),\ldots,(T_k)$ of size $1$ form a positive instance of \Q. In other words, we have no choice on the tape to choose in each tuple so we simply have an instance of \PathAutRec. We will later provide a reduction in the converse direction, proving that \PathAutRec{} and \Q{} have the same complexity. 
We define \SyncQ{} as the synchronized version of \Q{} as we defined \SyncAutRec{} from \AutRec{} earlier.

Before stating our main results, let us first prove that \SyncQ{} is hard. We obtain much stronger hardness results later but this proof is interesting since it illustrates the utility of having synchronized tapes to design hardness reductions with the following very simple statement.

\begin{theorem}\label{thm:w2}
 \SyncQ{} is \textsf{W[2]}-hard even when $|\Sigma|=1$.
\end{theorem}

\begin{proof}
We provide a reduction from \textsc{Dominating Set} parameterized by solution size $k$ to \SyncQ. 
It is well-known that \textsc{Dominating Set} is a \textsf{W[2]}-hard problem with respect to parameter $k$. 
Let $G=(V,E)$ be an $n$-vertex graph. Let $V=\{v_1,\ldots,v_n\}$ be the vertices of $G$.

We create $k$ tuples $\T^1,\ldots,\T^k$ of $n$ path tapes, each containing $n$ cells numbered by $1,\ldots,n$ from one endpoint to the other. The alphabet $\Sigma$ is a $1$-letter alphabet $\{\checkmark\}$, so every content is either $\checkmark$ or $\emptyset$. Now, for every $q \le k$, the $j$-th cell of the $i$-th tape of $\T^q$ contains $\checkmark$ if and only if the vertex $v_i$ is adjacent to $v_j$ in $G$ (or $i=j$). It essentially means that the $i$-th tape of each tuple encodes the neighborhood of $v_i$ in $G$. Indeed, the $j$-th cell of that tape contains $\checkmark$ if and only if $v_iv_j$ is an edge. Finally, in all the tapes, the initial reading head position is the leftmost cell of the path (numbered with $1$) and the target reading head position is the rightmost cell (numbered with $n$). 

We claim that $(\{\checkmark \},\T^1,\ldots,\T^k)$ is a positive instance of \SyncQ{} if and only if $G$ has a dominating set of size $k$. 

Assume first that $(\{\checkmark \},\T^1,\ldots,\T^k)$ is a positive instance. For every $j\in[1,k]$, let $T^j_{i_j}$ be the selected tape in the tuple $\T^j$. By definition of \SyncQ, there exists a synchronized transformation $\mathcal{S}$ of the instance $(\{\checkmark \},T^1_{i_1},\ldots,T^k_{i_k})$ of \SyncAutRec{} where, in the initial (resp. target) configuration, all the reading heads are on the leftmost (resp. rightmost) cell.
    
We claim that the following holds: for every $j \le n$, there exists a configuration $C_j$ in the transformation $\mathcal{S}$ where all the reading heads are on the $j$-th cell of each tape. For $j=1$ and $j=n$, the conclusion holds since they correspond to initial and target configurations. Now, for every $2 \le j \le n-1$, the configuration just before the first appearance of a cell numbered $j+1$ should only contain cells numbered with $j$.
Since for every $j$, the configuration $C_j$ is valid (since $\mathcal{S}$ is valid), one of the vertices $v_{i_1},\ldots,v_{i_k}$ is in the closed neighborhood of $v_j$. Therefore $\{v_{i_1},\ldots,v_{i_k}\}$ is a dominating set of $G$.
\smallskip

Conversely, assume that $D=\{ v_{i_1},\ldots,v_{i_k} \}$ is a dominating set. We select the tapes $T^j_{i_j}$ in $\T^j$ for every $j$. We define $C_p$ as the configuration where all the reading heads are at the $p$-th position in all the tapes $T^j_{i_j}$.

The initial and target positions, namely $C_1$ and $C_n$, indeed see $\checkmark$ since $D$ dominates $v_1$ and $v_n$. Now, for $j<n$, we conclude by defining a transformation from $C_j$ to $C_{j+1}$. Since $D$ dominates $v_{j+1}$, there is a vertex in $D$, say $v_{i_1}$, adjacent to $v_{j+1}$. We move the reading head on $T^1_{i_1}$ one cell to the right, so that this reading head has $\checkmark$ in its content. Then we can safely move one after another all the other reading heads at position $j+1$, which completes the proof.
\end{proof}

\begin{figure}
\center
\begin{tikzpicture}[thick]
    \node[circle, draw] (1) at ($(6,2)+(54:1)$) {1};
    \node[circle, draw] (2) at ($(6,2)+(126:1)$)  {2};
    \node[circle, draw] (3) at ($(6,2)+(198:1)$) {3};
    \node[circle, draw] (4) at ($(6,2)+(270:1)$) {4};
    \node[circle, draw] (5) at ($(6,2)+(-18:1)$) {5};
    
    \draw (1) -- (2);
    \draw (2) -- (3);
    \draw (3) -- (4);
    \draw (4) -- (5);
    \draw (5) -- (1);
    \tikzset{xscale=.45,yscale=.8}
    \node at (-1, 0) {$\biggl($};
    
    \node[circle, draw, minimum size=0.1cm] (a00) at (0, 0) {};
    \node at (0, 0.5) {\checkmark};
    \node[circle, draw, minimum size=0.1cm] (a01) at (1, 0) {};
    \node at (1, 0.5) {\checkmark};
    \node[circle, draw, minimum size=0.1cm] (a02) at (2, 0) {};
    \node at (2, 0.5) {$\emptyset$};
    \node[circle, draw, minimum size=0.1cm] (a03) at (3, 0) {};
    \node at (3, 0.5) {$\emptyset$};
    \node[circle, draw, minimum size=0.1cm] (a04) at (4, 0) {};
    \node at (4, 0.5) {\checkmark};
    \draw (a00) -- (a01);
    \draw (a01) -- (a02);
    \draw (a02) -- (a03);
    \draw (a03) -- (a04);
    
    \node at (5, 0) {,};
    
    \node[circle, draw, minimum size=0.1cm] (a10) at (6, 0) {};
    \node at (6, 0.5) {\checkmark};
    \node[circle, draw, minimum size=0.1cm] (a11) at (7, 0) {};
    \node at (7, 0.5) {\checkmark};
    \node[circle, draw, minimum size=0.1cm] (a12) at (8, 0) {};
    \node at (8, 0.5) {\checkmark};
    \node[circle, draw, minimum size=0.1cm] (a13) at (9, 0) {};
    \node at (9, 0.5) {$\emptyset$};
    \node[circle, draw, minimum size=0.1cm] (a14) at (10, 0) {};
    \node at (10, 0.5) {$\emptyset$};
    \draw (a10) -- (a11);
    \draw (a11) -- (a12);
    \draw (a12) -- (a13);
    \draw (a13) -- (a14);
    
    \node at (11, 0) {,};
    
    \node[circle, draw, minimum size=0.1cm] (a20) at (12, 0) {};
    \node at (12, 0.5) {$\emptyset$};
    \node[circle, draw, minimum size=0.1cm] (a21) at (13, 0) {};
    \node at (13, 0.5) {\checkmark};
    \node[circle, draw, minimum size=0.1cm] (a22) at (14, 0) {};
    \node at (14, 0.5) {\checkmark};
    \node[circle, draw, minimum size=0.1cm] (a23) at (15, 0) {};
    \node at (15, 0.5) {\checkmark};
    \node[circle, draw, minimum size=0.1cm] (a24) at (16, 0) {};
    \node at (16, 0.5) {$\emptyset$};
    \draw (a20) -- (a21);
    \draw (a21) -- (a22);
    \draw (a22) -- (a23);
    \draw (a23) -- (a24);
    
    \node at (17, 0) {,};
    
    \node[circle, draw, minimum size=0.1cm] (a30) at (18, 0) {};
    \node at (18, 0.5) {$\emptyset$};
    \node[circle, draw, minimum size=0.1cm] (a31) at (19, 0) {};
    \node at (19, 0.5) {$\emptyset$};
    \node[circle, draw, minimum size=0.1cm] (a32) at (20, 0) {};
    \node at (20, 0.5) {\checkmark};
    \node[circle, draw, minimum size=0.1cm] (a33) at (21, 0) {};
    \node at (21, 0.5) {\checkmark};
    \node[circle, draw, minimum size=0.1cm] (a34) at (22, 0) {};
    \node at (22, 0.5) {\checkmark};
    \draw (a30) -- (a31);
    \draw (a31) -- (a32);
    \draw (a32) -- (a33);
    \draw (a33) -- (a34);
    
    \node at (23, 0) {,};
    
    \node[circle, draw, minimum size=0.1cm] (a40) at (24, 0) {};
    \node at (24, 0.5) {\checkmark};
    \node[circle, draw, minimum size=0.1cm] (a41) at (25, 0) {};
    \node at (25, 0.5) {$\emptyset$};
    \node[circle, draw, minimum size=0.1cm] (a42) at (26, 0) {};
    \node at (26, 0.5) {$\emptyset$};
    \node[circle, draw, minimum size=0.1cm] (a43) at (27, 0) {};
    \node at (27, 0.5) {\checkmark};
    \node[circle, draw, minimum size=0.1cm] (a44) at (28, 0) {};
    \node at (28, 0.5) {\checkmark};
    \draw (a40) -- (a41);
    \draw (a41) -- (a42);
    \draw (a42) -- (a43);
    \draw (a43) -- (a44);
    
    \node at (29, 0) {$\biggl)$};

    \node at (-1, -1.5) {$\biggl($};
    
    \node[circle, draw, minimum size=0.1cm] (b00) at (0, -1.5) {};
    \node at (0, -1) {\checkmark};
    \node[circle, draw, minimum size=0.1cm] (b01) at (1, -1.5) {};
    \node at (1, -1) {\checkmark};
    \node[circle, draw, minimum size=0.1cm] (b02) at (2, -1.5) {};
    \node at (2, -1) {$\emptyset$};
    \node[circle, draw, minimum size=0.1cm] (b03) at (3, -1.5) {};
    \node at (3, -1) {$\emptyset$};
    \node[circle, draw, minimum size=0.1cm] (b04) at (4, -1.5) {};
    \node at (4, -1) {\checkmark};
    \draw (b00) -- (b01);
    \draw (b01) -- (b02);
    \draw (b02) -- (b03);
    \draw (b03) -- (b04);
    
    \node at (5, -1.5) {,};
    
    \node[circle, draw, minimum size=0.1cm] (b10) at (6, -1.5) {};
    \node at (6, -1) {\checkmark};
    \node[circle, draw, minimum size=0.1cm] (b11) at (7, -1.5) {};
    \node at (7, -1) {\checkmark};
    \node[circle, draw, minimum size=0.1cm] (b12) at (8, -1.5) {};
    \node at (8, -1) {\checkmark};
    \node[circle, draw, minimum size=0.1cm] (b13) at (9, -1.5) {};
    \node at (9, -1) {$\emptyset$};
    \node[circle, draw, minimum size=0.1cm] (b14) at (10, -1.5) {};
    \node at (10, -1) {$\emptyset$};
    \draw (b10) -- (b11);
    \draw (b11) -- (b12);
    \draw (b12) -- (b13);
    \draw (b13) -- (b14);
    
    \node at (11, -1.5) {,};
    
    \node[circle, draw, minimum size=0.1cm] (b20) at (12, -1.5) {};
    \node at (12, -1) {$\emptyset$};
    \node[circle, draw, minimum size=0.1cm] (b21) at (13, -1.5) {};
    \node at (13, -1) {\checkmark};
    \node[circle, draw, minimum size=0.1cm] (b22) at (14, -1.5) {};
    \node at (14, -1) {\checkmark};
    \node[circle, draw, minimum size=0.1cm] (b23) at (15, -1.5) {};
    \node at (15, -1) {\checkmark};
    \node[circle, draw, minimum size=0.1cm] (b24) at (16, -1.5) {};
    \node at (16, -1) {$\emptyset$};
    \draw (b20) -- (b21);
    \draw (b21) -- (b22);
    \draw (b22) -- (b23);
    \draw (b23) -- (b24);
    
    \node at (17, -1.5) {,};
    
    \node[circle, draw, minimum size=0.1cm] (b30) at (18, -1.5) {};
    \node at (18, -1) {$\emptyset$};
    \node[circle, draw, minimum size=0.1cm] (b31) at (19, -1.5) {};
    \node at (19, -1) {$\emptyset$};
    \node[circle, draw, minimum size=0.1cm] (b32) at (20, -1.5) {};
    \node at (20, -1) {\checkmark};
    \node[circle, draw, minimum size=0.1cm] (b33) at (21, -1.5) {};
    \node at (21, -1) {\checkmark};
    \node[circle, draw, minimum size=0.1cm] (b34) at (22, -1.5) {};
    \node at (22, -1) {\checkmark};
    \draw (b30) -- (b31);
    \draw (b31) -- (b32);
    \draw (b32) -- (b33);
    \draw (b33) -- (b34);
    
    \node at (23, -1.5) {,};
    
    \node[circle, draw, minimum size=0.1cm] (b40) at (24, -1.5) {};
    \node at (24, -1) {\checkmark};
    \node[circle, draw, minimum size=0.1cm] (b41) at (25, -1.5) {};
    \node at (25, -1) {$\emptyset$};
    \node[circle, draw, minimum size=0.1cm] (b42) at (26, -1.5) {};
    \node at (26, -1) {$\emptyset$};
    \node[circle, draw, minimum size=0.1cm] (b43) at (27, -1.5) {};
    \node at (27, -1) {\checkmark};
    \node[circle, draw, minimum size=0.1cm] (b44) at (28, -1.5) {};
    \node at (28, -1) {\checkmark};
    \draw (b40) -- (b41);
    \draw (b41) -- (b42);
    \draw (b42) -- (b43);
    \draw (b43) -- (b44);
    
    \node at (29, -1.5) {$\biggl)$};

\end{tikzpicture}
    \caption{The instance of \SyncQ{} equivalent to the instance of \textsc{Dominating Set} for the $5$-cycle and where $k = 2$.}
\end{figure}
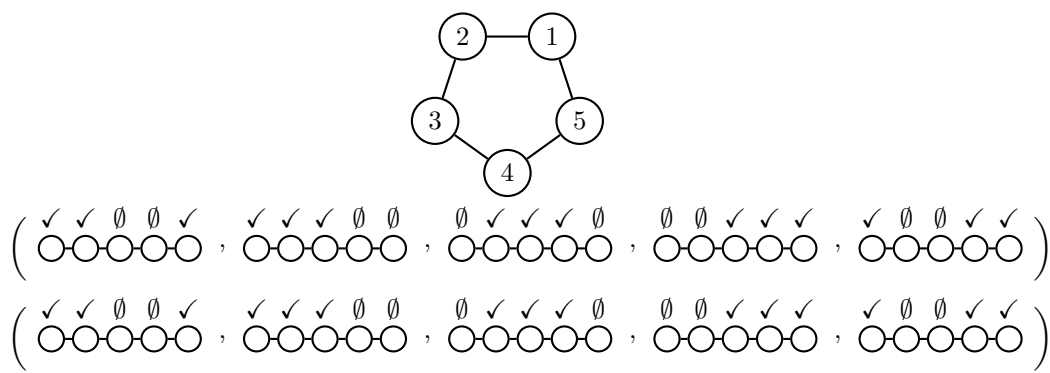

\subsection{Hardness results for tape reconfiguration}\label{sec:hardautrec}

\subparagraph*{Width of tape reconfiguration instances.}
In the rest of this section we will need some notions of width of a tape reconfiguration instance. Let $I=(\Sigma,(T_i)_i,C_s,C_t)$ be an instance of \AutRec{} (or any other tape problem defined above such as \Q{} or \SyncAutRec{}), where $\Sigma$ is the alphabet and $(T_i)_i$ is a collection of tapes. The \emph{extended graph} of $I$ is the graph with vertex set $\bigcup_i V(T_i)\cup\Sigma$, and where $xy$ is an edge if either it is an edge of some $T_i$, or $x$ is a cell of some $T_i$ which contains the letter $y$. 

We say that $I$ is \emph{$\degen$-degenerate} whenever its extended graph is $\degen$-degenerate. We similarly define the notion of pathwidth, treewidth, and minors of an instance $I$; the \emph{pathwidth (resp. treewidth)} of $I$ is the pathwidth (resp. treewidth) of its extended graph. 

We say that a path (resp. tree) decomposition of the extended graph of $I$ is \emph{$\struct$-structured} if each of its bags contains vertices of at most $\struct$ tapes. The utility of structured decompositions will appear naturally in future reductions. The minimum width of a $\struct$-structured path (resp. tree) decomposition is called the \emph{$\struct$-structured pathwidth (resp. treewidth)} of $I$.

\subparagraph*{Hardness results.}
The goal of this section is to give the main hardness results for \AutRec{} and its variants. The first and one of our main results is the following:

\begin{restatable}{theorem}{thmAutRecPSPACE}
\label{thm:AutRecPSPACE}
\AutRec{} is \PPSPACE-complete and \textsf{XL}-complete even restricted to $5$-degenerate instances of $2$-structured treewidth at most $9$ and $3$-structured pathwidth at most~$14$.
\end{restatable}


The proof of Theorem~\ref{thm:AutRecPSPACE} is divided into two steps. The first one consists in proving that the synchronized version of \AutRec{}, \SyncAutRec{}, is \PPSPACE-complete and \textsf{XL}-complete. The reduction is from TJ-DSR.
The proof generalizes the proof of Theorem~\ref{thm:w2}. Indeed, Theorem~\ref{thm:w2} shows first how to encode a choice of vertices by a choice of tapes, and then translates checking that these vertices dominate the graph into moving read heads on these tapes. We reuse these ``checking paths'' to prove Theorem~\ref{thm:AutRecPSPACE}, but now we have to encode a way to reconfigure dominating sets. This can be done by arranging the checking paths between some cells that correspond to vertices of $G$. By doing so, our tapes are not paths anymore; they become subdivided stars. Moreover, in order to guarantee that we reconfigure dominating sets one vertex at a time, we also need to add another tape. This construction leads to the following statement which we prove in Section~\ref{sec:PSPACE-Sync}:

\begin{restatable}{theorem}{thmSyncAutTw}
\label{thm:auto_tw4}
    \SyncAutRec{} is \PPSPACE-complete and \textsf{XL}-complete parameterized by $k$ plus the feedback vertex set number, even restricted to $3$-degenerate instances of $1$-structured treewidth at most $3$, $2$-structured pathwidth at most $5$, and whose tapes are subdivided stars.
\end{restatable}

The second step in the proof of Theorem~\ref{thm:AutRecPSPACE} consists in providing a very generic reduction from synchronized problems to their unsynchronized versions. This reduction allows us to prove that synchronized and unsynchronized versions of each problem are actually equivalent (up to small changes in the parameters and the structured width). 

To prove this, we add a new tape and new characters to $\Sigma$ that force the read heads to move at the same ``speed'' in all the tapes. The reduction also ensures that all the cells are useful to cover the whole alphabet $\Sigma$. Namely, we say that an instance $(\Sigma,\T,C_s,C_t)$ of \AutRec{} is \emph{irreducible} if $\Sigma$ cannot be written as the union of the contents of less than $|\T|$~cells.

Formally, we show that the following holds, which directly implies Theorem~\ref{thm:AutRecPSPACE}, when combined with Theorem~\ref{thm:auto_tw4}.

\begin{restatable}{theorem}{thmSyncToAsync}
    \label{thm:sync_to_async}
    There is an \textsf{FPT}-reduction (and \textsf{PL}-reduction) $\Phi$ from instances of \SyncAutRec{} to irreducible instances of \AutRec{}. Moreover, if $I$ is an instance of \SyncAutRec{} of $\struct$-structured treewidth (resp. pathwidth) $\pwtw$, then $\Phi(I)$ has $(\struct+1)$-structured treewidth (resp. pathwidth) at most $\pwtw+3\struct+3$. Furthermore, $\Phi$ increases the degeneracy by at most $2$ and the feedback vertex set number by at most $3|\Sigma|+3$.
\end{restatable}

Theorem~\ref{thm:AutRecPSPACE} is a direct corollary of Theorems~\ref{thm:auto_tw4} and~\ref{thm:sync_to_async}. Theorem~\ref{thm:sync_to_async} will be proved in Section~\ref{sec:sync_to_async}.
We give two proofs of Theorem~\ref{thm:sync_to_async}; one where one additional tape is a long path and one where it is a triangle. The spirit of the two proofs is similar but each construction allows us to reach a specific conclusion. For long paths, we guarantee that all tapes are paths and thus obtain hardness results for \PathAutRec{}. However this path reduction comes at a cost; the size of clique-minors and treewidth increase by a function of the size of the alphabet, hence they become unbounded. We thus provide another reduction where the additional tape is a triangle that permits to control widths and degeneracy. In the end, the first reduction yields hardness results for \TSDSR{} parameterized by feedback vertex set number, while the second handles the bounded widths and degeneracy cases.


\subsection{Consequences for (C)DSR}
Let us finally explain the consequences of the results from the previous subsections on \TSDSR{}, TJ-CDSR, and TS-CDSR.

\thmTSDSR*

The proof of Theorem~\ref{thm:TSDSR} follows from the next lemma:
    
    \begin{lemma}\label{thm:reduc_tape_dsr}
        There is an \FPT{}-reduction (and \textsf{PL}-reduction) from \AutRec{} on irreducible $\degen$-degenerate instances of $\struct$-structured treewidth (resp. pathwidth) at most $\pwtw$ with $k$ tapes and of feedback vertex set number $\fvs$ to \TSDSR{} and \textsc{TS-CDSR} with $k+1$ tokens on graphs:
        \begin{itemize}
            \item of treewidth (resp. pathwidth) at most $\struct+\pwtw+1$ and degeneracy at most $\degen+2$, or
            \item of feedback vertex set number $\fvs+k+1$.
        \end{itemize}
    \end{lemma}
    \begin{proof}
    Let $I=(\Sigma,\{T_1,\ldots,T_k\},C_s,C_t)$  be an instance of \AutRec.
    We construct an instance of \TSDSR{} $(G,D_s,D_t)$ as follows. The graph $G$ is built from the extended graph of $I$ by adding some vertices. Recall that the extended graph of $I$ contains a copy of $T_1,\ldots,T_k$ and a vertex for each letter of $\Sigma$. We call the former \emph{tape vertices} and the latter \emph{alphabet vertices}. We add:
    \begin{itemize}
        \item for $i\in [1,k]$, a vertex $x_i$ adjacent to all the tape vertices in $T_i$.
        \item a vertex $y$ adjacent to all the tape vertices.
        \item a vertex $z$ of degree one attached to $y$.
    \end{itemize}
Let us denote by $G$ the resulting graph. We define $D_s$ (resp. $D_t$) as the set of tape vertices under the read heads in $C_s$ (resp. $C_t$), plus $y$. 

We prove that the transformation is a parameterized logspace (PL) reduction. The graph $G$ is built from the extended graph of $I$ by adding $k+2$ additional vertices ($x_1, \ldots, x_k$, $y$, and $z$), and connecting them to existing tape vertices in a fixed, locally determined way. Since each of these additions depends only on local information and bounded iteration over $k$ tapes and the alphabet $\Sigma$, the entire construction can be carried out using logarithmic space. The dominating sets $D_s$ and $D_t$ are defined using the positions of the read heads in $C_s$ and $C_t$, which can be extracted from the input configuration in logspace. Thus, the transformation is computable in logspace and preserves the parameter, and therefore constitutes a valid PL-reduction.

The soundness of the reduction is based on the following structural property of minimum dominating sets in $G$.

    \begin{claim}\label{clm:correspondance}
        Minimum dominating sets of $G$ have size $k+1$, contain one vertex in $\{y,z\}$, exactly one tape vertex in each $T_i$ and no alphabet vertex.
    \end{claim}
    
    \begin{claimproof}
        Since $C_s$ is a valid configuration, the tape vertices in $D_s$ dominate $\Sigma$ plus all the $x_i$'s and $y$. Since $D_s$ also contains $y$, it dominates $z$ hence $D_s$ is a dominating set of size $k+1$.

        Note that the neighborhoods of all the $x_i$'s and of $z$ are pairwise disjoint and do not contain any alphabet vertex, hence minimum dominating sets have size exactly $k+1$ and cannot contain any alphabet vertex. Moreover, in order to dominate $z$, every dominating set must contain a vertex in $\{y,z\}$. Finally, since minimum dominating sets do not contain alphabet vertices, these vertices can only be dominated by tape vertices. Since $I$ is irreducible, we get that minimum dominating sets must contain one tape vertex in each $T_i$. 
    \end{claimproof} 

    Since the closed neighborhood of $z$ is included in the one of $y$, in a DSR transformation, we can replace every move to $z$ by a move to $y$, so that the resulting reconfiguration sequence never puts a token on $z$. Since $D_s$ and $D_t$ do not contain $z$, we can assume in the rest of the proof that all the dominating sets we consider contain $y$ and not $z$. In particular, they are all connected and the reduction is also well-defined for \textsc{TS-CDSR}. Therefore, given a minimum dominating set $D$, by construction of $G$, placing the reading head of $T_i$ on the tape vertex of $T_i$ in $D$ yields a valid configuration of $I$. Moreover any token slide in $D$ corresponds to moving a reading head on an adjacent cell. This shows that the reductions are sound.

    It remains to study the widths and feedback vertex set numbers of this instance. One can easily remark that removing all the vertices $x_i$ and $y$ yields the extension graph of $I$ plus the isolated vertex $z$. Thus,  the feedback vertex set number increased by at most $k+1$ from $I$ to $G$. Now, given an $\struct$-structured decomposition $A$ of $I$, we add $x_i$ to all the bags containing some tape vertex of $T_i$. Since tapes are connected, the bags containing $x_i$ also induce a connected subpart of the decomposition. Moreover, the size of the bags increased by at most $\struct$ since $A$ was $\struct$-structured. This yields a decomposition of $G-\{y,z\}$. We now get a decomposition of $G$ by adding $y$ in all bags, and creating a leaf bag containing $y$ and $z$. 
\end{proof}

By Theorem~\ref{thm:AutRecPSPACE}, \AutRec{} is \PPSPACE-complete and \textsf{XL}-complete on $5$-degenerate instances of $2$-structures treewidth at most $9$ and $3$-structured pathwidth at most $14$. Moreover, by Theorem~\ref{thm:sync_to_async}, this remains true when the instances are additionally irreducible. Therefore, applying Lemma~\ref{thm:reduc_tape_dsr}, we can conclude the proof of Theorem~\ref{thm:TSDSR}. 

Note that considering irreducible instances in the above proof allowed us to increase the widths by at most $\pwtw+1$. Without this assumption, a similar reduction works, but we need to introduce a twin of each $x_i$. This yields a reduction to possibly non-irreducible instances of \AutRec{} with slightly worse bounds (adding $2\pwtw+1$ to the widths instead).

The reduction can also be modified to get hardness results for the parameterized (and classical) complexity of \textsc{TJ-CDSR}, which solves open problems from the literature. Namely, we adapt our reduction to show that \textsc{TJ-CDSR} is \PPSPACE-complete and \textsf{XL}-complete.

\thmCDSR*

\begin{proof} 
The basis of the reduction of this proof is the one of Theorem~\ref{thm:TSDSR} that we will slightly modify.
As shown before, in the token sliding model, we can only consider minimum dominating sets containing $y$. Therefore, all these sets are connected. For the token jumping model, we adjust  the reduction as described next. We subdivide once every edge between two tape vertices. The new vertices are called the \emph{subdivided vertices} and the others the \emph{original vertices}. We link $y,x_1,\ldots,x_k$ with the tape vertices as follows:

\begin{itemize}
    \item $y$ is adjacent to all original vertices in each tape (but not to subdivided vertices); 
    \item $x_i$ is adjacent to all subdivided vertices of tape $T_i$ (but not to original vertices).
\end{itemize}

One can easily adapt the proof of Claim~\ref{clm:correspondance} to show that connected dominating sets of size $3k+1$ contain exactly one of $\{y,z\}$, all the $x_i$'s (up to adding dummy cells to ensure that tapes have size at least $2k$) and exactly one original vertex and one subdivided vertex in each tape, which are adjacent. Indeed, each dominating set of size $3k+1$ must contain every $x_i$ (since no other vertex can dominate more than $2$ subdivided vertices in $T_i$), and by connectivity to $y$, it must also contain an original and subdivided (adjacent) vertices in each tape. 
%
%
Using this fact, we get: 

\begin{enumerate}
    \item If $D$ is a connected dominating set of size $3k+1$ then the set of original vertices is a valid reading head position in the instance $I$ of \AutRec{}.
    \item If a token jumps from an original vertex $v$, it ends up on another original vertex corresponding to a cell adjacent to the cell of $v$.
    \item If a token jumps from a subdivided vertex, it ends up on another subdivided vertex still adjacent to the original vertex in the same tape. 
\end{enumerate}

This concludes the reduction. Moreover, using a proof similar to the one of Lemma~\ref{thm:reduc_tape_dsr}, it follows that the widths of the instance of TJ-CDSR increase by at most $\struct + 2$ compared to the $\struct$-structured widths of the instance of \AutRec{} (subdividing edges adds at most $1$ to the widths, and then we add at most $\struct+1$ vertices in each bag). For the degeneracy, observe that subdivided vertices have degree $3$, so one can find a $5$-degeneracy ordering by first taking the subdivided vertices, then following the order in $I$, then $y,z$ and the $x_i$'s.
\end{proof}

\subsection{Hardness of tape reconfiguration on paths}

\AutRec{} is \textsf{XL}-complete even restricted to instances where tapes are trees (and even subdivided stars). One can then naturally wonder if the same holds when tapes are simpler graphs, for example paths. We did not succeed to obtain exactly the same hardness result but nevertheless we prove that \PathAutRec{} is \textsf{W[$*$]}-hard. The \textsf{W[2]}-hardness is a consequence of the synchronization result used to prove Theorem~\ref{thm:sync_to_async}. This result is actually quite generic and can be applied to all the problems defined before. In particular, we will get the following for free from Theorem~\ref{thm:w2}:

\begin{remark}
    \Q{} is \textsf{W[2]}-hard.
\end{remark}

Recall that there is a trivial reduction from \PathAutRec{} to \Q. Using another idea, called the \emph{selector gadgets}, we can provide a reduction in the converse direction, proving that \PathAutRec{} and \Q{} have the same complexity. More formally, we will prove that the following holds:

\begin{theorem}\label{thm:Q-AUTREC}
    \Q{} and \PathAutRec{} are equivalent under \textsf{FPT}-reductions.
\end{theorem}

As a byproduct, Theorems~\ref{thm:w2} and~\ref{thm:Q-AUTREC} ensure that \PathAutRec{} is \textsf{W[2]}-hard. We will actually prove the following more general result (as sketched above):

\begin{restatable}{theorem}{thmWstarHard}
\label{thm:wstarhard}
    \Q{} (and consequently \PathAutRec{}) is \textsf{W[$*$]}-hard.
\end{restatable}

To prove Theorem~\ref{thm:wstarhard}, we will generalize the ideas of Theorem~\ref{thm:w2}. In the proof of Theorem~\ref{thm:w2}, we gave a reduction from \textsc{Dominating Set}, which is a problem in \textsf{W[2]}. In other words, it can be expressed with a logical sentence, whose variables correspond to the vertices chosen in the dominating set, with the following shape: 
\[\bigwedge_{u\in V}\bigvee_{v\in N[u]} x_v,\]
that is a conjunction of disjunctions of literals. 
The reduction of Theorem~\ref{thm:w2} allowed us to make a selection (by the definition of \Q) and then check the conjunction of conditions (one condition per cell), each of them corresponding to a disjunction (at least one cell has to contain~$\checkmark$).

The \textsf{W}-hierarchy can be generalized to higher levels;  \textsf{W[$h$]} problems can be expressed in the same way except that up to $h$ nested boolean operators are allowed. Therefore, in order to prove that \Q{} is \textsf{W[$h$]}-hard, we build on this \textsf{W[2]}-hardness result and design OR gadgets and AND gadgets allowing us to lift \textsf{W[$h+1$]}-hardness from \textsf{W[$h$]}-hardness. Note that this has to be done without adding too many characters nor additional tapes. The details of the proof are postponed to Section~\ref{sec:proofw*}.

We were not able to determine if the problem is \textsf{XL}-complete or not and leave as an open problem the exact complexity of \PathAutRec{}.

\section{FPT results roadmap}

The goal of this section is to give a high-level idea of the proof of Theorem~\ref{thm:FPT-planar}. The proof of Theorem~\ref{thm:FPT-minor} has a similar flavor but is heavily more technical. We proceed along the lines of Lokshtanov et al.~\cite{DBLP:journals/jcss/LokshtanovMPRS18} who proved that \TJDSR{} is \textsf{FPT} on nowhere dense graphs. 
A \emph{$k$-domination core} of a graph $G$ (or simply a domination core when $k$ and $G$ are clear from context) is a subset of vertices $X$ such that a subset $D$ of size at most $k$ is a dominating set of $G$ if and only if $D$ dominates $X$. In other words, in order to dominate the graph, one only has to guarantee that $X$ is dominated. 
From~\cite{eiben2019lossy}, all $K_{3,d}$-free graphs having a dominating set of size $k$ admit a $k$-domination core $X$ of size at most $(2d+1)k^{d+1}$ that can be found in polynomial time; starting with $X = V(G)$, and as long as $|X| > (2d+1)k^{d+1}$ one can find (in polynomial time) a vertex $x \in X$ such that $X \setminus \{x\}$ is still a $k$-domination core. We therefore assume in the rest of the section that each graph comes with a domination core. 

Lokshtanov et al.~\cite{DBLP:journals/jcss/LokshtanovMPRS18} partitioned the vertices of $G$ into classes according to their neighborhoods in the domination core $X$ and proved that it is enough to keep one vertex per class. In the case of token sliding, it is not clear whether we are allowed to reduce to one vertex per class. We successfully adapt the domination core technique to the token sliding model, but only on smaller classes of graphs (a necessary restriction, given our hardness results), namely $K_{3,d}$ and $K_{4,d}$-minor free graphs. 

We sketch the proof for $K_{3,d}$-free graphs. The proof for $K_{4,d}$-minor free graphs follows the same basic framework but involves more technical arguments.
Let $G$ be a $K_{3,d}$-free graph and $I=(G,k,D_s,D_t)$ be an instance of \TSDSR.  Since super-sets of domination cores are still domination cores, we can compute a domination core $X$ of $G$ containing $D_s\cup D_t$ (we need this condition for convenience to ensure that vertices outside of $X$ do not belong to the source and target dominating sets). By definition, $I$ is a positive instance of \TSDSR{} if and only if there exists a TS-reconfiguration sequence from $D_s$ to $D_t$ that dominates $X$ all along.


For every $Y \subseteq X$, the $Y$-class is the set of vertices $v\in V\setminus X$ with $N(v) \cap X = Y$ \footnote{All along the paper, the \emph{neighborhood} $N(v)$ of $v$ will denote the set of neighbors of $v$ and $N[v]$ will be $N(v) \cup \{v\}$.}, and classes of type $p$ (or $p$-classes) are those with $|Y|=p$. Lokshtanov et al.~\cite{DBLP:journals/jcss/LokshtanovMPRS18} remarked that we can remove all but at most one of the vertices of each class if we consider the token jumping model. In the same spirit, our goal is to reduce classes. First, observe that $3^+$-classes have size less than $d$ (since $G$ is $K_{3,d}$-free). The remaining classes are less straightforward to handle. An easy first step consists in deleting \emph{twins} (vertices with the same neighborhood), and contracting each edge between two vertices of the same class. Note that contraction may not preserve $K_{3,d}$-freeness, but this hypothesis is only used to bound the size of classes of type at least $3$, and we never use it again afterwards.

At that point $G$ is a twin-free graph whose classes are independent sets. We prove (in the appendix) that anytime a token slides on an edge $uv$ between two classes of type at most $2$ then, unless $G-uv$ is disconnected, we can slide the token along a $uv$-path avoiding the edge $uv$ while still dominating $X$. We can then show that adding a vertex $y$ (in the $0$-class) adjacent to all the vertices of $V \setminus X$ does not affect the validity of the instance. This allows us to remove all the edges between $\le 2$-classes (except the ones incident on $y$). Consequently, all the vertices of the same $\le 2$-class are twins up to their neighborhood in $\ge 3$-classes. Since $\ge 3$-classes have bounded size and their number is at most $2^{|X|}$, the twin removal argument guarantees that $\le 2$-classes can also be bounded in size, as needed to complete the argument.






\bibliographystyle{plain}
\bibliography{bliblio}

\newpage
\appendix

\section{Hardness results for tape reconfiguration (Proofs of Section~\ref{sec:hardautrec})}

The goal of this section is to provide a reduction showing that \AutRec{} is \PPSPACE-complete and \textsf{XL}-complete, even for restricted instances. We start by recalling the definitions of the complexity classes that will be used all along this section.  

\subsection{Preliminaries}\label{sec:definition}

\subparagraph{The \textsf{W}-hierarchy and the class \textsf{W[$*$]}.} 
A parameterized problem belongs to the class \textsf{W[$h$]}, $h \geq 1$, if it can be reduced in \textsf{FPT}-time to the so-called \textsc{Weighted}-$h$\textsc{-Normalized-Sat} problem. An \emph{$h$-normalized} boolean formulas is defined inductively; a single literal is $0$-normalized, and $h$-normalized formulas are conjunctions of disjunctions of $(h-1)$-normalized formulas. Given an $h$-normalized boolean formula and an integer $k$, the \textsc{Weighted}-$h$\textsc{-normalized-Sat} problem consists in deciding whether the formula has an truth assignment with at most $k$ variables set to true. A parameterized problem is \textsf{W[$*$]}-hard if it is \textsf{W[$h$]}-hard for all $h \geq 1$.

\medskip
\noindent
\textsc{Weighted}-$h$\textsc{-normalized-Sat} \\
\textbf{Input:} An $h$-normalized boolean formula $\varphi$.\\
\textbf{Parameter:} $k\in\mathbb{N}$.\\
\textbf{Output:} Yes if there is a truth assignment satisfying $\varphi$ of weight at most $k$ (i.e. where at most $k$ variables are set to true).

\subparagraph{The class \textsf{XL}.} 
\textsf{XL} is a wider class of parameterized problems that contains the \textsf{W}-hierarchy. It  consists of parameterized problems that can be solved by a deterministic Turing machine using $f(k) \cdot \log n$ space, where $k$ is the parameter, $n$ is the input size, and $f$ is any computable function\footnote{The class \textsf{XL} can be seen as the parameterized counterpart of the class \textsf{LogSpace} in the classical complexity setting.}. 

A \emph{symmetric Turing machine} is a nondeterministic Turing machine, where the transitions are symmetric, that is,  every transition can be taken in reverse. Informally speaking, it means that if we perform a computation step in the Turing machine, we can immediately cancel that step and come back to the previous configuration. 



Also note that a symmetric Turing machine is not deterministic but one can actually easily prove that \textsf{XL} $=$ \textsf{XSL}, that is, what can be done with a deterministic Turing machine is actually equivalent to what can be done with a symmetric non-deterministic Turing machine (see e.g.~\cite{DBLP:conf/iwpec/BodlaenderGS21} for more detail in the context of reconfiguration problems). This statement can be seen as the parameterized counterpart of the fact that reachability in an undirected graph is in \textsf{LogSpace} (while in directed graphs it is \textsf{NLogSpace}-complete). This is helpful in our DSR setting, since the transformations are indeed symmetric, and moreover, they can be represented using $f(k) \log n$ bits since there are at most $n^k$ dominating sets. In particular, we get that both TJ-DSR and TS-DSR belong to \textsf{XSL} hence to \textsf{XL}.

\subparagraph{PL-reductions.}
We justify that all reductions used to establish \textsf{XL}-completeness in this paper are computable in \emph{parameterized logspace}, i.e., they are PL-reductions. This ensures the correctness of our claims within space-bounded parameterized complexity classes such as \textsf{XL}. 
We recall~\cite{DBLP:conf/iwpec/BodlaenderGS21} that a function $f$ is a \emph{PL-reduction} from a parameterized (by $k$) problem $A$ to a parameterized problem $B$ if:
\begin{enumerate}
    \item $x \in A \iff f(x) \in B$,
    \item $f$ is computable in $\mathcal{O}(\log |x| + g(k))$ space, for some computable function $g$,
    \item The output parameter is bounded by a computable function of the input parameter $k$.
\end{enumerate}
This restriction ensures that both the transformation and the parameter blowup are space-efficient, which is essential when working with classes such as \textsf{XL}.

\subsection{\PPSPACE{} and \textsf{XL}-hardness of synchronized tape reconfiguration (Proof of Theorem~\ref{thm:auto_tw4})}\label{sec:PSPACE-Sync}

This part is devoted to showing the following:

\thmSyncAutTw*

    We provide a reduction from \textsc{TJ-Partitioned-DSR} to \SyncAutRec{} with the desired structured width. In the \PPSPACE-complete and \textsf{XL}-complete  \textsc{TJ-Partitioned-DSR} problem~\cite{DBLP:conf/iwpec/BodlaenderGS21}, we are given a graph $G$, $k$ vertex disjoint subsets of vertices $V_1,\ldots,V_k$, and two dominating sets $D_s,D_t$, each containing one vertex in each $V_i$. We are looking for a transformation via token jumps where every intermediate dominating set also contains exactly one vertex in each $V_i$. 
    
    Let $G$ be an $n$-vertex graph, $V_1,\ldots,V_k$ be sets of vertices of $G$ and $D_s,D_t$ be two dominating sets of $G$ of size $k$. We create an instance of \SyncAutRec{} as follows. For $i\in[1,k]$, let $n_i=|V_i|$. We define the tape $T_i$ as the star with $n_i$ branches where every edge is replaced by a path of length $2n+1$. We number the cells of $T_i$ so that each branch is numbered $1,2,\ldots,n+1,1,n+1,\ldots,2$. Among the cells numbered with $1$, there is the \emph{center} of the tape and other cells that are called the \emph{middle} cells. We can associate each middle cell with a vertex in $V_i$. Finally, we create another tape $T^*$ constructed similarly but with $k$ branches. This tape will represent which token is currently allowed to move in the transformation.

    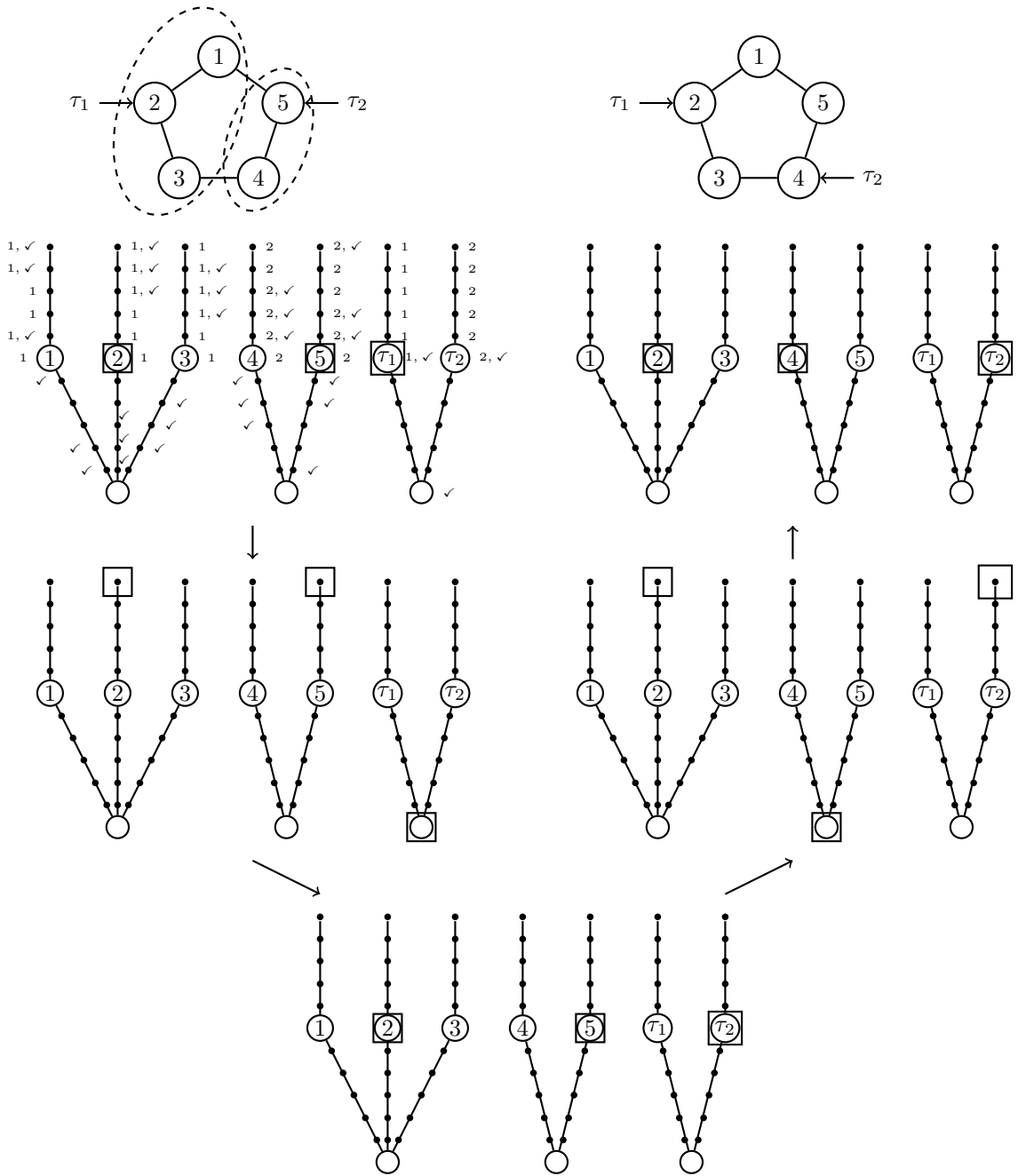
\begin{figure}[!ht]
    \centering
    \begin{tikzpicture}[thick]
    \tikzset{xshift=1.5cm}
        \node[draw, circle] (A) at ($(3,0)+(90:1)$) {$1$};
        \node[draw, circle] (B) at ($(3,0)+(162:1)$) {$2$};
        \node[draw, circle] (C) at ($(3,0)+(234:1)$) {$3$};
        \node[draw, circle] (D) at ($(3,0)+(306:1)$) {$4$};
        \node[draw, circle] (E) at ($(3,0)+(18:1)$) {$5$};
        \draw (A) -- (B) -- (C) -- (D) -- (E) -- (A);
        \draw[dashed,rotate around={72:($(3,0)+(-18:0.77)$)}] ($(3,0)+(-18:0.77)$) ellipse (1.1cm and 0.6cm);
        \draw[dashed,rotate around={72:($(3,0)+(162:0.6)$)}] ($(3,0)+(162:0.6)$) ellipse (1.6cm and 0.9cm);
        \node[left=.5cm of B] (t1) {$\tau_1$};
        \node[right=.5cm of E] (t2) {$\tau_2$};
        \draw[->] (t1) to (B);
        \draw[->] (t2) to (E);

        \tikzset{xshift=8cm}
        \node[draw, circle] (A) at ($(3,0)+(90:1)$) {$1$};
        \node[draw, circle] (B) at ($(3,0)+(162:1)$) {$2$};
        \node[draw, circle] (C) at ($(3,0)+(234:1)$) {$3$};
        \node[draw, circle] (D) at ($(3,0)+(306:1)$) {$4$};
        \node[draw, circle] (E) at ($(3,0)+(18:1)$) {$5$};
        \draw (A) -- (B) -- (C) -- (D) -- (E) -- (A);
        \node[left=.5cm of B] (t1) {$\tau_1$};
        \node[right=.5cm of D] (t2) {$\tau_2$};
        \draw[->] (t1) to (B);
        \draw[->] (t2) to (D);
        \tikzset{yshift=-5.5cm,xshift=-6.5cm}

        \node[draw,circle] (c1) at (0,0) {};
        \node[draw,circle,inner sep=1pt,label=left:{\tiny$1$}] (m11) at (-1,2) {1};
        \node[draw,rectangle,inner sep=6pt] at (0,2) {};
        \node[draw,circle,inner sep=1pt,label=right:{\tiny$1$}] (m12) at (0,2) {2};
        \node[draw,circle,inner sep=1pt,label=right:{\tiny$1$}] (m13) at (1,2) {3};
        \node[fill,circle,inner sep=1pt,label=left:{\tiny$\checkmark$}] (f111) at (-.16,.33) {};
        \node[fill,circle,inner sep=1pt,label=left:{\tiny$\checkmark$}] (f112) at (-.33,.66) {};
        \node[fill,circle,inner sep=1pt] (f113) at (-.5,1) {};
        \node[fill,circle,inner sep=1pt] (f114) at (-.66,1.33) {};
        \node[fill,circle,inner sep=1pt,label=left:{\tiny$\checkmark$}] (f115) at (-.83,1.66) {};
        
        \node[fill,circle,inner sep=1pt,label=left:{\tiny$1,\checkmark$}] (f111) at (-1,2.33) {};
        \node[fill,circle,inner sep=1pt,label=left:{\tiny$1$}] (f112) at (-1,2.66) {};
        \node[fill,circle,inner sep=1pt,label=left:{\tiny$1$}] (f113) at (-1,3) {};
        \node[fill,circle,inner sep=1pt,label=left:{\tiny$1,\checkmark$}] (f114) at (-1,3.33) {};
        \node[fill,circle,inner sep=1pt,label=left:{\tiny$1,\checkmark$}] (f115) at (-1,3.66) {};
        \draw (c1) -- (m11) -- (f115);

        \node[fill,circle,inner sep=1pt] (f111) at (.16,.33) {};
        \node[fill,circle,inner sep=1pt,label=right:{\tiny$\checkmark$}] (f112) at (.33,.66) {};
        \node[fill,circle,inner sep=1pt,label=right:{\tiny$\checkmark$}] (f113) at (.5,1) {};
        \node[fill,circle,inner sep=1pt,label=right:{\tiny$\checkmark$}] (f114) at (.66,1.33) {};
        \node[fill,circle,inner sep=1pt] (f115) at (.83,1.66) {};
        
        \node[fill,circle,inner sep=1pt,label=right:{\tiny$1$}] (f111) at (1,2.33) {};
        \node[fill,circle,inner sep=1pt,label=right:{\tiny$1,\checkmark$}] (f112) at (1,2.66) {};
        \node[fill,circle,inner sep=1pt,label=right:{\tiny$1,\checkmark$}] (f113) at (1,3) {};
        \node[fill,circle,inner sep=1pt,label=right:{\tiny$1,\checkmark$}] (f114) at (1,3.33) {};
        \node[fill,circle,inner sep=1pt,label=right:{\tiny$1$}] (f115) at (1,3.66) {};
        \draw (c1) -- (m13) -- (f115);

        \node[fill,circle,inner sep=1pt] (f111) at (0,.33) {};
        \node at (0.1,.48) {\tiny $\checkmark$};
        \node[fill,circle,inner sep=1pt] (f112) at (0,.66) {};
        \node at (0.1,.81) {\tiny $\checkmark$};
        \node[fill,circle,inner sep=1pt] (f113) at (0,1) {};
        \node at (0.1,1.15) {\tiny $\checkmark$};
        \node[fill,circle,inner sep=1pt] (f114) at (0,1.33) {};
        \node[fill,circle,inner sep=1pt] (f115) at (0,1.66) {};
        
        \node[fill,circle,inner sep=1pt,label=right:{\tiny$1$}] (f111) at (0,2.33) {};
        \node[fill,circle,inner sep=1pt,label=right:{\tiny$1$}] (f112) at (0,2.66) {};
        \node[fill,circle,inner sep=1pt,label=right:{\tiny$1,\checkmark$}] (f113) at (0,3) {};
        \node[fill,circle,inner sep=1pt,label=right:{\tiny$1,\checkmark$}] (f114) at (0,3.33) {};
        \node[fill,circle,inner sep=1pt,label=right:{\tiny$1,\checkmark$}] (f115) at (0,3.66) {};
        \draw (c1) -- (m12) -- (f115);

        \tikzset{xshift=.5cm}
        \node[draw,circle] (c2) at (2,0) {};
        \node[draw,circle,inner sep=1pt,label=right:{\tiny$2$}] (m21) at (1.5,2) {4};
        \node[draw,rectangle,inner sep=6pt] at (2.5,2) {};
        \node[draw,circle,inner sep=1pt,label=right:{\tiny$2$}] (m22) at (2.5,2) {5};
        \node[fill,circle,inner sep=1pt] (f211) at (1.90,.33) {};
        \node[fill,circle,inner sep=1pt] (f212) at (1.82,.66) {};
        \node[fill,circle,inner sep=1pt,label=left:{\tiny$\checkmark$}] (f213) at (1.74,1) {};
        \node[fill,circle,inner sep=1pt,label=left:{\tiny$\checkmark$}] (f214) at (1.66,1.33) {};
        \node[fill,circle,inner sep=1pt,label=left:{\tiny$\checkmark$}] (f215) at (1.58,1.66) {};
        
        \node[fill,circle,inner sep=1pt,label=right:{\tiny$2,\checkmark$}] (f111) at (1.5,2.33) {};
        \node[fill,circle,inner sep=1pt,label=right:{\tiny$2,\checkmark$}] (f112) at (1.5,2.66) {};
        \node[fill,circle,inner sep=1pt,label=right:{\tiny$2,\checkmark$}] (f113) at (1.5,3) {};
        \node[fill,circle,inner sep=1pt,label=right:{\tiny$2$}] (f114) at (1.5,3.33) {};
        \node[fill,circle,inner sep=1pt,label=right:{\tiny$2$}] (f115) at (1.5,3.66) {};
        \draw (c2) -- (m21) -- (f115);

        \node[fill,circle,inner sep=1pt,label=right:{\tiny$\checkmark$}] (f211) at (2.10,.33) {};
        \node[fill,circle,inner sep=1pt] (f212) at (2.18,.66) {};
        \node[fill,circle,inner sep=1pt] (f213) at (2.26,1) {};
        \node[fill,circle,inner sep=1pt,label=right:{\tiny$\checkmark$}] (f214) at (2.34,1.33) {};
        \node[fill,circle,inner sep=1pt,label=right:{\tiny$\checkmark$}] (f215) at (2.42,1.66) {};
        
        \node[fill,circle,inner sep=1pt,label=right:{\tiny$2,\checkmark$}] (f111) at (2.5,2.33) {};
        \node[fill,circle,inner sep=1pt,label=right:{\tiny$2,\checkmark$}] (f112) at (2.5,2.66) {};
        \node[fill,circle,inner sep=1pt,label=right:{\tiny$2$}] (f113) at (2.5,3) {};
        \node[fill,circle,inner sep=1pt,label=right:{\tiny$2$}] (f114) at (2.5,3.33) {};
        \node[fill,circle,inner sep=1pt,label=right:{\tiny$2,\checkmark$}] (f115) at (2.5,3.66) {};
        \draw (c2) -- (m22) -- (f115);

        \tikzset{xshift=.5cm}
        \node[draw,circle,label=right:{\tiny$\checkmark$}] (c3) at (3.5,0) {};
        \node[draw,rectangle,inner sep=7pt] at (3,2) {};
        \node[draw,circle,inner sep=.5pt,label=right:{\hspace{-.1cm}\tiny$1,\checkmark$}] (m31) at (3,2) {$\tau_1$};
        \node[draw,circle,inner sep=.5pt,label=right:{\tiny$2,\checkmark$}] (m32) at (4,2) {$\tau_2$};
        \node[fill,circle,inner sep=1pt] (f211) at (3.4,.33) {};
        \node[fill,circle,inner sep=1pt] (f212) at (3.32,.66) {};
        \node[fill,circle,inner sep=1pt] (f213) at (3.24,1) {};
        \node[fill,circle,inner sep=1pt] (f214) at (3.16,1.33) {};
        \node[fill,circle,inner sep=1pt] (f215) at (3.08,1.66) {};
        
        \node[fill,circle,inner sep=1pt,label=right:{\tiny $1$}] (f111) at (3,2.33) {};
        \node[fill,circle,inner sep=1pt,label=right:{\tiny $1$}] (f112) at (3,2.66) {};
        \node[fill,circle,inner sep=1pt,label=right:{\tiny $1$}] (f113) at (3,3) {};
        \node[fill,circle,inner sep=1pt,label=right:{\tiny $1$}] (f114) at (3,3.33) {};
        \node[fill,circle,inner sep=1pt,label=right:{\tiny $1$}] (f115) at (3,3.66) {};
        \draw (c3) -- (m31) -- (f115);

        \node[fill,circle,inner sep=1pt] (f211) at (3.60,.33) {};
        \node[fill,circle,inner sep=1pt] (f212) at (3.68,.66) {};
        \node[fill,circle,inner sep=1pt] (f213) at (3.76,1) {};
        \node[fill,circle,inner sep=1pt] (f214) at (3.84,1.33) {};
        \node[fill,circle,inner sep=1pt] (f215) at (3.92,1.66) {};
        
        \node[fill,circle,inner sep=1pt,label=right:{\tiny $2$}] (f111) at (4,2.33) {};
        \node[fill,circle,inner sep=1pt,label=right:{\tiny $2$}] (f112) at (4,2.66) {};
        \node[fill,circle,inner sep=1pt,label=right:{\tiny $2$}] (f113) at (4,3) {};
        \node[fill,circle,inner sep=1pt,label=right:{\tiny $2$}] (f114) at (4,3.33) {};
        \node[fill,circle,inner sep=1pt,label=right:{\tiny $2$}] (f115) at (4,3.66) {};
        \draw (c3) -- (m32) -- (f115);
        \draw[<-] (1,-1) to (1,-.5);
        \tikzset{xshift=-1cm,yshift=-5cm}

        \node[draw,circle] (c1) at (0,0) {};
        \node[draw,circle,inner sep=1pt] (m11) at (-1,2) {1};
        \node[draw,circle,inner sep=1pt] (m12) at (0,2) {2};
        \node[draw,circle,inner sep=1pt] (m13) at (1,2) {3};
        \node[fill,circle,inner sep=1pt] (f111) at (-.16,.33) {};
        \node[fill,circle,inner sep=1pt] (f112) at (-.33,.66) {};
        \node[fill,circle,inner sep=1pt] (f113) at (-.5,1) {};
        \node[fill,circle,inner sep=1pt] (f114) at (-.66,1.33) {};
        \node[fill,circle,inner sep=1pt] (f115) at (-.83,1.66) {};
        
        \node[fill,circle,inner sep=1pt] (f111) at (-1,2.33) {};
        \node[fill,circle,inner sep=1pt] (f112) at (-1,2.66) {};
        \node[fill,circle,inner sep=1pt] (f113) at (-1,3) {};
        \node[fill,circle,inner sep=1pt] (f114) at (-1,3.33) {};
        \node[fill,circle,inner sep=1pt] (f115) at (-1,3.66) {};
        \draw (c1) -- (m11) -- (f115);

        \node[fill,circle,inner sep=1pt] (f111) at (.16,.33) {};
        \node[fill,circle,inner sep=1pt] (f112) at (.33,.66) {};
        \node[fill,circle,inner sep=1pt] (f113) at (.5,1) {};
        \node[fill,circle,inner sep=1pt] (f114) at (.66,1.33) {};
        \node[fill,circle,inner sep=1pt] (f115) at (.83,1.66) {};
        
        \node[fill,circle,inner sep=1pt] (f111) at (1,2.33) {};
        \node[fill,circle,inner sep=1pt] (f112) at (1,2.66) {};
        \node[fill,circle,inner sep=1pt] (f113) at (1,3) {};
        \node[fill,circle,inner sep=1pt] (f114) at (1,3.33) {};
        \node[fill,circle,inner sep=1pt] (f115) at (1,3.66) {};
        \draw (c1) -- (m13) -- (f115);

        \node[fill,circle,inner sep=1pt] (f111) at (0,.33) {};
        \node[fill,circle,inner sep=1pt] (f112) at (0,.66) {};
        \node[fill,circle,inner sep=1pt] (f113) at (0,1) {};
        \node[fill,circle,inner sep=1pt] (f114) at (0,1.33) {};
        \node[fill,circle,inner sep=1pt] (f115) at (0,1.66) {};
        
        \node[fill,circle,inner sep=1pt] (f111) at (0,2.33) {};
        \node[fill,circle,inner sep=1pt] (f112) at (0,2.66) {};
        \node[fill,circle,inner sep=1pt] (f113) at (0,3) {};
        \node[fill,circle,inner sep=1pt] (f114) at (0,3.33) {};
        \node[draw,rectangle,inner sep=6pt] at (0,3.66) {};
        \node[fill,circle,inner sep=1pt] (f115) at (0,3.66) {};
        \draw (c1) -- (m12) -- (f115);

        \tikzset{xshift=.5cm}
        \node[draw,circle] (c2) at (2,0) {};
        \node[draw,circle,inner sep=1pt] (m21) at (1.5,2) {4};
        \node[draw,circle,inner sep=1pt] (m22) at (2.5,2) {5};
        \node[fill,circle,inner sep=1pt] (f211) at (1.90,.33) {};
        \node[fill,circle,inner sep=1pt] (f212) at (1.82,.66) {};
        \node[fill,circle,inner sep=1pt] (f213) at (1.74,1) {};
        \node[fill,circle,inner sep=1pt] (f214) at (1.66,1.33) {};
        \node[fill,circle,inner sep=1pt] (f215) at (1.58,1.66) {};
        
        \node[fill,circle,inner sep=1pt] (f111) at (1.5,2.33) {};
        \node[fill,circle,inner sep=1pt] (f112) at (1.5,2.66) {};
        \node[fill,circle,inner sep=1pt] (f113) at (1.5,3) {};
        \node[fill,circle,inner sep=1pt] (f114) at (1.5,3.33) {};
        \node[draw,rectangle,inner sep=6pt] at (2.5,3.66) {};
        \node[fill,circle,inner sep=1pt] (f115) at (1.5,3.66) {};
        \draw (c2) -- (m21) -- (f115);

        \node[fill,circle,inner sep=1pt] (f211) at (2.10,.33) {};
        \node[fill,circle,inner sep=1pt] (f212) at (2.18,.66) {};
        \node[fill,circle,inner sep=1pt] (f213) at (2.26,1) {};
        \node[fill,circle,inner sep=1pt] (f214) at (2.34,1.33) {};
        \node[fill,circle,inner sep=1pt] (f215) at (2.42,1.66) {};
        
        \node[fill,circle,inner sep=1pt] (f111) at (2.5,2.33) {};
        \node[fill,circle,inner sep=1pt] (f112) at (2.5,2.66) {};
        \node[fill,circle,inner sep=1pt] (f113) at (2.5,3) {};
        \node[fill,circle,inner sep=1pt] (f114) at (2.5,3.33) {};
        \node[fill,circle,inner sep=1pt] (f115) at (2.5,3.66) {};
        \draw (c2) -- (m22) -- (f115);

        \tikzset{xshift=.5cm}
        \node[draw,rectangle,inner sep=6pt] at (3.5,0) {};
        \node[draw,circle] (c3) at (3.5,0) {};
        \node[draw,circle,inner sep=.5pt] (m31) at (3,2) {$\tau_1$};
        \node[draw,circle,inner sep=.5pt] (m32) at (4,2) {$\tau_2$};
        \node[fill,circle,inner sep=1pt] (f211) at (3.4,.33) {};
        \node[fill,circle,inner sep=1pt] (f212) at (3.32,.66) {};
        \node[fill,circle,inner sep=1pt] (f213) at (3.24,1) {};
        \node[fill,circle,inner sep=1pt] (f214) at (3.16,1.33) {};
        \node[fill,circle,inner sep=1pt] (f215) at (3.08,1.66) {};
        
        \node[fill,circle,inner sep=1pt] (f111) at (3,2.33) {};
        \node[fill,circle,inner sep=1pt] (f112) at (3,2.66) {};
        \node[fill,circle,inner sep=1pt] (f113) at (3,3) {};
        \node[fill,circle,inner sep=1pt] (f114) at (3,3.33) {};
        \node[fill,circle,inner sep=1pt] (f115) at (3,3.66) {};
        \draw (c3) -- (m31) -- (f115);

        \node[fill,circle,inner sep=1pt] (f211) at (3.60,.33) {};
        \node[fill,circle,inner sep=1pt] (f212) at (3.68,.66) {};
        \node[fill,circle,inner sep=1pt] (f213) at (3.76,1) {};
        \node[fill,circle,inner sep=1pt] (f214) at (3.84,1.33) {};
        \node[fill,circle,inner sep=1pt] (f215) at (3.92,1.66) {};
        
        \node[fill,circle,inner sep=1pt] (f111) at (4,2.33) {};
        \node[fill,circle,inner sep=1pt] (f112) at (4,2.66) {};
        \node[fill,circle,inner sep=1pt] (f113) at (4,3) {};
        \node[fill,circle,inner sep=1pt] (f114) at (4,3.33) {};
        \node[fill,circle,inner sep=1pt] (f115) at (4,3.66) {};
        \draw (c3) -- (m32) -- (f115);
        \draw[<-] (2,-1) to (1,-.5);
        \tikzset{xshift=3cm,yshift=-5cm}
        \node[draw,circle] (c1) at (0,0) {};
        \node[draw,circle,inner sep=1pt] (m11) at (-1,2) {1};
        \node[draw,rectangle,inner sep=6pt] at (0,2) {};
        \node[draw,circle,inner sep=1pt] (m12) at (0,2) {2};
        \node[draw,circle,inner sep=1pt] (m13) at (1,2) {3};
        \node[fill,circle,inner sep=1pt] (f111) at (-.16,.33) {};
        \node[fill,circle,inner sep=1pt] (f112) at (-.33,.66) {};
        \node[fill,circle,inner sep=1pt] (f113) at (-.5,1) {};
        \node[fill,circle,inner sep=1pt] (f114) at (-.66,1.33) {};
        \node[fill,circle,inner sep=1pt] (f115) at (-.83,1.66) {};
        
        \node[fill,circle,inner sep=1pt] (f111) at (-1,2.33) {};
        \node[fill,circle,inner sep=1pt] (f112) at (-1,2.66) {};
        \node[fill,circle,inner sep=1pt] (f113) at (-1,3) {};
        \node[fill,circle,inner sep=1pt] (f114) at (-1,3.33) {};
        \node[fill,circle,inner sep=1pt] (f115) at (-1,3.66) {};
        \draw (c1) -- (m11) -- (f115);

        \node[fill,circle,inner sep=1pt] (f111) at (.16,.33) {};
        \node[fill,circle,inner sep=1pt] (f112) at (.33,.66) {};
        \node[fill,circle,inner sep=1pt] (f113) at (.5,1) {};
        \node[fill,circle,inner sep=1pt] (f114) at (.66,1.33) {};
        \node[fill,circle,inner sep=1pt] (f115) at (.83,1.66) {};
        
        \node[fill,circle,inner sep=1pt] (f111) at (1,2.33) {};
        \node[fill,circle,inner sep=1pt] (f112) at (1,2.66) {};
        \node[fill,circle,inner sep=1pt] (f113) at (1,3) {};
        \node[fill,circle,inner sep=1pt] (f114) at (1,3.33) {};
        \node[fill,circle,inner sep=1pt] (f115) at (1,3.66) {};
        \draw (c1) -- (m13) -- (f115);

        \node[fill,circle,inner sep=1pt] (f111) at (0,.33) {};
        \node[fill,circle,inner sep=1pt] (f112) at (0,.66) {};
        \node[fill,circle,inner sep=1pt] (f113) at (0,1) {};
        \node[fill,circle,inner sep=1pt] (f114) at (0,1.33) {};
        \node[fill,circle,inner sep=1pt] (f115) at (0,1.66) {};
        
        \node[fill,circle,inner sep=1pt] (f111) at (0,2.33) {};
        \node[fill,circle,inner sep=1pt] (f112) at (0,2.66) {};
        \node[fill,circle,inner sep=1pt] (f113) at (0,3) {};
        \node[fill,circle,inner sep=1pt] (f114) at (0,3.33) {};
        \node[fill,circle,inner sep=1pt] (f115) at (0,3.66) {};
        \draw (c1) -- (m12) -- (f115);

        \tikzset{xshift=.5cm}
        \node[draw,circle] (c2) at (2,0) {};
        \node[draw,circle,inner sep=1pt] (m21) at (1.5,2) {4};
        \node[draw,rectangle,inner sep=6pt] at (2.5,2) {};
        \node[draw,circle,inner sep=1pt] (m22) at (2.5,2) {5};
        \node[fill,circle,inner sep=1pt] (f211) at (1.90,.33) {};
        \node[fill,circle,inner sep=1pt] (f212) at (1.82,.66) {};
        \node[fill,circle,inner sep=1pt] (f213) at (1.74,1) {};
        \node[fill,circle,inner sep=1pt] (f214) at (1.66,1.33) {};
        \node[fill,circle,inner sep=1pt] (f215) at (1.58,1.66) {};
        
        \node[fill,circle,inner sep=1pt] (f111) at (1.5,2.33) {};
        \node[fill,circle,inner sep=1pt] (f112) at (1.5,2.66) {};
        \node[fill,circle,inner sep=1pt] (f113) at (1.5,3) {};
        \node[fill,circle,inner sep=1pt] (f114) at (1.5,3.33) {};
        \node[fill,circle,inner sep=1pt] (f115) at (1.5,3.66) {};
        \draw (c2) -- (m21) -- (f115);

        \node[fill,circle,inner sep=1pt] (f211) at (2.10,.33) {};
        \node[fill,circle,inner sep=1pt] (f212) at (2.18,.66) {};
        \node[fill,circle,inner sep=1pt] (f213) at (2.26,1) {};
        \node[fill,circle,inner sep=1pt] (f214) at (2.34,1.33) {};
        \node[fill,circle,inner sep=1pt] (f215) at (2.42,1.66) {};
        
        \node[fill,circle,inner sep=1pt] (f111) at (2.5,2.33) {};
        \node[fill,circle,inner sep=1pt] (f112) at (2.5,2.66) {};
        \node[fill,circle,inner sep=1pt] (f113) at (2.5,3) {};
        \node[fill,circle,inner sep=1pt] (f114) at (2.5,3.33) {};
        \node[fill,circle,inner sep=1pt] (f115) at (2.5,3.66) {};
        \draw (c2) -- (m22) -- (f115);

        \tikzset{xshift=.5cm}
        \node[draw,circle] (c3) at (3.5,0) {};
        \node[draw,rectangle,inner sep=7pt] at (4,2) {};
        \node[draw,circle,inner sep=.5pt] (m31) at (3,2) {$\tau_1$};
        \node[draw,circle,inner sep=.5pt] (m32) at (4,2) {$\tau_2$};
        \node[fill,circle,inner sep=1pt] (f211) at (3.4,.33) {};
        \node[fill,circle,inner sep=1pt] (f212) at (3.32,.66) {};
        \node[fill,circle,inner sep=1pt] (f213) at (3.24,1) {};
        \node[fill,circle,inner sep=1pt] (f214) at (3.16,1.33) {};
        \node[fill,circle,inner sep=1pt] (f215) at (3.08,1.66) {};
        
        \node[fill,circle,inner sep=1pt] (f111) at (3,2.33) {};
        \node[fill,circle,inner sep=1pt] (f112) at (3,2.66) {};
        \node[fill,circle,inner sep=1pt] (f113) at (3,3) {};
        \node[fill,circle,inner sep=1pt] (f114) at (3,3.33) {};
        \node[fill,circle,inner sep=1pt] (f115) at (3,3.66) {};
        \draw (c3) -- (m31) -- (f115);

        \node[fill,circle,inner sep=1pt] (f211) at (3.60,.33) {};
        \node[fill,circle,inner sep=1pt] (f212) at (3.68,.66) {};
        \node[fill,circle,inner sep=1pt] (f213) at (3.76,1) {};
        \node[fill,circle,inner sep=1pt] (f214) at (3.84,1.33) {};
        \node[fill,circle,inner sep=1pt] (f215) at (3.92,1.66) {};
        
        \node[fill,circle,inner sep=1pt] (f111) at (4,2.33) {};
        \node[fill,circle,inner sep=1pt] (f112) at (4,2.66) {};
        \node[fill,circle,inner sep=1pt] (f113) at (4,3) {};
        \node[fill,circle,inner sep=1pt] (f114) at (4,3.33) {};
        \node[fill,circle,inner sep=1pt] (f115) at (4,3.66) {};
        \draw (c3) -- (m32) -- (f115);

        \draw[<-] (5,4.5) to (4,4);
\tikzset{xshift=3cm,yshift=5cm}

        \node[draw,circle] (c1) at (0,0) {};
        \node[draw,circle,inner sep=1pt] (m11) at (-1,2) {1};
        \node[draw,circle,inner sep=1pt] (m12) at (0,2) {2};
        \node[draw,circle,inner sep=1pt] (m13) at (1,2) {3};
        \node[fill,circle,inner sep=1pt] (f111) at (-.16,.33) {};
        \node[fill,circle,inner sep=1pt] (f112) at (-.33,.66) {};
        \node[fill,circle,inner sep=1pt] (f113) at (-.5,1) {};
        \node[fill,circle,inner sep=1pt] (f114) at (-.66,1.33) {};
        \node[fill,circle,inner sep=1pt] (f115) at (-.83,1.66) {};
        
        \node[fill,circle,inner sep=1pt] (f111) at (-1,2.33) {};
        \node[fill,circle,inner sep=1pt] (f112) at (-1,2.66) {};
        \node[fill,circle,inner sep=1pt] (f113) at (-1,3) {};
        \node[fill,circle,inner sep=1pt] (f114) at (-1,3.33) {};
        \node[fill,circle,inner sep=1pt] (f115) at (-1,3.66) {};
        \draw (c1) -- (m11) -- (f115);

        \node[fill,circle,inner sep=1pt] (f111) at (.16,.33) {};
        \node[fill,circle,inner sep=1pt] (f112) at (.33,.66) {};
        \node[fill,circle,inner sep=1pt] (f113) at (.5,1) {};
        \node[fill,circle,inner sep=1pt] (f114) at (.66,1.33) {};
        \node[fill,circle,inner sep=1pt] (f115) at (.83,1.66) {};
        
        \node[fill,circle,inner sep=1pt] (f111) at (1,2.33) {};
        \node[fill,circle,inner sep=1pt] (f112) at (1,2.66) {};
        \node[fill,circle,inner sep=1pt] (f113) at (1,3) {};
        \node[fill,circle,inner sep=1pt] (f114) at (1,3.33) {};
        \node[fill,circle,inner sep=1pt] (f115) at (1,3.66) {};
        \draw (c1) -- (m13) -- (f115);

        \node[fill,circle,inner sep=1pt] (f111) at (0,.33) {};
        \node[fill,circle,inner sep=1pt] (f112) at (0,.66) {};
        \node[fill,circle,inner sep=1pt] (f113) at (0,1) {};
        \node[fill,circle,inner sep=1pt] (f114) at (0,1.33) {};
        \node[fill,circle,inner sep=1pt] (f115) at (0,1.66) {};
        
        \node[fill,circle,inner sep=1pt] (f111) at (0,2.33) {};
        \node[fill,circle,inner sep=1pt] (f112) at (0,2.66) {};
        \node[fill,circle,inner sep=1pt] (f113) at (0,3) {};
        \node[fill,circle,inner sep=1pt] (f114) at (0,3.33) {};
        \node[fill,circle,inner sep=1pt] (f115) at (0,3.66) {};
        \node[draw,rectangle,inner sep=6pt] at (0,3.66) {};

        \draw (c1) -- (m12) -- (f115);

        \tikzset{xshift=.5cm}
        \node[draw,circle] (c2) at (2,0) {};
        \node[draw,rectangle,inner sep=6pt] at (2,0) {};
        \node[draw,circle,inner sep=1pt] (m21) at (1.5,2) {4};
        \node[draw,circle,inner sep=1pt] (m22) at (2.5,2) {5};
        \node[fill,circle,inner sep=1pt] (f211) at (1.90,.33) {};
        \node[fill,circle,inner sep=1pt] (f212) at (1.82,.66) {};
        \node[fill,circle,inner sep=1pt] (f213) at (1.74,1) {};
        \node[fill,circle,inner sep=1pt] (f214) at (1.66,1.33) {};
        \node[fill,circle,inner sep=1pt] (f215) at (1.58,1.66) {};
        
        \node[fill,circle,inner sep=1pt] (f111) at (1.5,2.33) {};
        \node[fill,circle,inner sep=1pt] (f112) at (1.5,2.66) {};
        \node[fill,circle,inner sep=1pt] (f113) at (1.5,3) {};
        \node[fill,circle,inner sep=1pt] (f114) at (1.5,3.33) {};
        \node[fill,circle,inner sep=1pt] (f115) at (1.5,3.66) {};
        \draw (c2) -- (m21) -- (f115);

        \node[fill,circle,inner sep=1pt] (f211) at (2.10,.33) {};
        \node[fill,circle,inner sep=1pt] (f212) at (2.18,.66) {};
        \node[fill,circle,inner sep=1pt] (f213) at (2.26,1) {};
        \node[fill,circle,inner sep=1pt] (f214) at (2.34,1.33) {};
        \node[fill,circle,inner sep=1pt] (f215) at (2.42,1.66) {};
        
        \node[fill,circle,inner sep=1pt] (f111) at (2.5,2.33) {};
        \node[fill,circle,inner sep=1pt] (f112) at (2.5,2.66) {};
        \node[fill,circle,inner sep=1pt] (f113) at (2.5,3) {};
        \node[fill,circle,inner sep=1pt] (f114) at (2.5,3.33) {};
        \node[fill,circle,inner sep=1pt] (f115) at (2.5,3.66) {};
        \draw (c2) -- (m22) -- (f115);

        \tikzset{xshift=.5cm}
        \node[draw,circle] (c3) at (3.5,0) {};
        \node[draw,circle,inner sep=.5pt] (m31) at (3,2) {$\tau_1$};
        \node[draw,circle,inner sep=.5pt] (m32) at (4,2) {$\tau_2$};
        \node[fill,circle,inner sep=1pt] (f211) at (3.4,.33) {};
        \node[fill,circle,inner sep=1pt] (f212) at (3.32,.66) {};
        \node[fill,circle,inner sep=1pt] (f213) at (3.24,1) {};
        \node[fill,circle,inner sep=1pt] (f214) at (3.16,1.33) {};
        \node[fill,circle,inner sep=1pt] (f215) at (3.08,1.66) {};
        
        \node[fill,circle,inner sep=1pt] (f111) at (3,2.33) {};
        \node[fill,circle,inner sep=1pt] (f112) at (3,2.66) {};
        \node[fill,circle,inner sep=1pt] (f113) at (3,3) {};
        \node[fill,circle,inner sep=1pt] (f114) at (3,3.33) {};
        \node[fill,circle,inner sep=1pt] (f115) at (3,3.66) {};
        \draw (c3) -- (m31) -- (f115);

        \node[fill,circle,inner sep=1pt] (f211) at (3.60,.33) {};
        \node[fill,circle,inner sep=1pt] (f212) at (3.68,.66) {};
        \node[fill,circle,inner sep=1pt] (f213) at (3.76,1) {};
        \node[fill,circle,inner sep=1pt] (f214) at (3.84,1.33) {};
        \node[fill,circle,inner sep=1pt] (f215) at (3.92,1.66) {};
        
        \node[fill,circle,inner sep=1pt] (f111) at (4,2.33) {};
        \node[fill,circle,inner sep=1pt] (f112) at (4,2.66) {};
        \node[fill,circle,inner sep=1pt] (f113) at (4,3) {};
        \node[fill,circle,inner sep=1pt] (f114) at (4,3.33) {};
        \node[fill,circle,inner sep=1pt] (f115) at (4,3.66) {};
        \node[draw,rectangle,inner sep=7pt] at (4,3.66) {};
        \draw (c3) -- (m32) -- (f115);
        \draw[<-] (1,4.5) to (1,4);
\tikzset{xshift=-1cm,yshift=5cm}

        \node[draw,circle] (c1) at (0,0) {};
        \node[draw,circle,inner sep=1pt] (m11) at (-1,2) {1};
        \node[draw,rectangle,inner sep=6pt] at (0,2) {};
        \node[draw,circle,inner sep=1pt] (m12) at (0,2) {2};
        \node[draw,circle,inner sep=1pt] (m13) at (1,2) {3};
        \node[fill,circle,inner sep=1pt] (f111) at (-.16,.33) {};
        \node[fill,circle,inner sep=1pt] (f112) at (-.33,.66) {};
        \node[fill,circle,inner sep=1pt] (f113) at (-.5,1) {};
        \node[fill,circle,inner sep=1pt] (f114) at (-.66,1.33) {};
        \node[fill,circle,inner sep=1pt] (f115) at (-.83,1.66) {};
        
        \node[fill,circle,inner sep=1pt] (f111) at (-1,2.33) {};
        \node[fill,circle,inner sep=1pt] (f112) at (-1,2.66) {};
        \node[fill,circle,inner sep=1pt] (f113) at (-1,3) {};
        \node[fill,circle,inner sep=1pt] (f114) at (-1,3.33) {};
        \node[fill,circle,inner sep=1pt] (f115) at (-1,3.66) {};
        \draw (c1) -- (m11) -- (f115);

        \node[fill,circle,inner sep=1pt] (f111) at (.16,.33) {};
        \node[fill,circle,inner sep=1pt] (f112) at (.33,.66) {};
        \node[fill,circle,inner sep=1pt] (f113) at (.5,1) {};
        \node[fill,circle,inner sep=1pt] (f114) at (.66,1.33) {};
        \node[fill,circle,inner sep=1pt] (f115) at (.83,1.66) {};
        
        \node[fill,circle,inner sep=1pt] (f111) at (1,2.33) {};
        \node[fill,circle,inner sep=1pt] (f112) at (1,2.66) {};
        \node[fill,circle,inner sep=1pt] (f113) at (1,3) {};
        \node[fill,circle,inner sep=1pt] (f114) at (1,3.33) {};
        \node[fill,circle,inner sep=1pt] (f115) at (1,3.66) {};
        \draw (c1) -- (m13) -- (f115);

        \node[fill,circle,inner sep=1pt] (f111) at (0,.33) {};
        \node[fill,circle,inner sep=1pt] (f112) at (0,.66) {};
        \node[fill,circle,inner sep=1pt] (f113) at (0,1) {};
        \node[fill,circle,inner sep=1pt] (f114) at (0,1.33) {};
        \node[fill,circle,inner sep=1pt] (f115) at (0,1.66) {};
        
        \node[fill,circle,inner sep=1pt] (f111) at (0,2.33) {};
        \node[fill,circle,inner sep=1pt] (f112) at (0,2.66) {};
        \node[fill,circle,inner sep=1pt] (f113) at (0,3) {};
        \node[fill,circle,inner sep=1pt] (f114) at (0,3.33) {};
        \node[fill,circle,inner sep=1pt] (f115) at (0,3.66) {};
        \draw (c1) -- (m12) -- (f115);

        \tikzset{xshift=.5cm}
        \node[draw,circle] (c2) at (2,0) {};
        \node[draw,rectangle,inner sep=6pt] at (1.5,2) {};
        \node[draw,circle,inner sep=1pt] (m21) at (1.5,2) {4};
        \node[draw,circle,inner sep=1pt] (m22) at (2.5,2) {5};
        \node[fill,circle,inner sep=1pt] (f211) at (1.90,.33) {};
        \node[fill,circle,inner sep=1pt] (f212) at (1.82,.66) {};
        \node[fill,circle,inner sep=1pt] (f213) at (1.74,1) {};
        \node[fill,circle,inner sep=1pt] (f214) at (1.66,1.33) {};
        \node[fill,circle,inner sep=1pt] (f215) at (1.58,1.66) {};
        
        \node[fill,circle,inner sep=1pt] (f111) at (1.5,2.33) {};
        \node[fill,circle,inner sep=1pt] (f112) at (1.5,2.66) {};
        \node[fill,circle,inner sep=1pt] (f113) at (1.5,3) {};
        \node[fill,circle,inner sep=1pt] (f114) at (1.5,3.33) {};
        \node[fill,circle,inner sep=1pt] (f115) at (1.5,3.66) {};
        \draw (c2) -- (m21) -- (f115);

        \node[fill,circle,inner sep=1pt] (f211) at (2.10,.33) {};
        \node[fill,circle,inner sep=1pt] (f212) at (2.18,.66) {};
        \node[fill,circle,inner sep=1pt] (f213) at (2.26,1) {};
        \node[fill,circle,inner sep=1pt] (f214) at (2.34,1.33) {};
        \node[fill,circle,inner sep=1pt] (f215) at (2.42,1.66) {};
        
        \node[fill,circle,inner sep=1pt] (f111) at (2.5,2.33) {};
        \node[fill,circle,inner sep=1pt] (f112) at (2.5,2.66) {};
        \node[fill,circle,inner sep=1pt] (f113) at (2.5,3) {};
        \node[fill,circle,inner sep=1pt] (f114) at (2.5,3.33) {};
        \node[fill,circle,inner sep=1pt] (f115) at (2.5,3.66) {};
        \draw (c2) -- (m22) -- (f115);

        \tikzset{xshift=.5cm}
        \node[draw,circle] (c3) at (3.5,0) {};
        \node[draw,rectangle,inner sep=7pt] at (4,2) {};
        \node[draw,circle,inner sep=.5pt] (m31) at (3,2) {$\tau_1$};
        \node[draw,circle,inner sep=.5pt] (m32) at (4,2) {$\tau_2$};
        \node[fill,circle,inner sep=1pt] (f211) at (3.4,.33) {};
        \node[fill,circle,inner sep=1pt] (f212) at (3.32,.66) {};
        \node[fill,circle,inner sep=1pt] (f213) at (3.24,1) {};
        \node[fill,circle,inner sep=1pt] (f214) at (3.16,1.33) {};
        \node[fill,circle,inner sep=1pt] (f215) at (3.08,1.66) {};
        
        \node[fill,circle,inner sep=1pt] (f111) at (3,2.33) {};
        \node[fill,circle,inner sep=1pt] (f112) at (3,2.66) {};
        \node[fill,circle,inner sep=1pt] (f113) at (3,3) {};
        \node[fill,circle,inner sep=1pt] (f114) at (3,3.33) {};
        \node[fill,circle,inner sep=1pt] (f115) at (3,3.66) {};
        \draw (c3) -- (m31) -- (f115);

        \node[fill,circle,inner sep=1pt] (f211) at (3.60,.33) {};
        \node[fill,circle,inner sep=1pt] (f212) at (3.68,.66) {};
        \node[fill,circle,inner sep=1pt] (f213) at (3.76,1) {};
        \node[fill,circle,inner sep=1pt] (f214) at (3.84,1.33) {};
        \node[fill,circle,inner sep=1pt] (f215) at (3.92,1.66) {};
        
        \node[fill,circle,inner sep=1pt] (f111) at (4,2.33) {};
        \node[fill,circle,inner sep=1pt] (f112) at (4,2.66) {};
        \node[fill,circle,inner sep=1pt] (f113) at (4,3) {};
        \node[fill,circle,inner sep=1pt] (f114) at (4,3.33) {};
        \node[fill,circle,inner sep=1pt] (f115) at (4,3.66) {};
        \draw (c3) -- (m32) -- (f115);

    \end{tikzpicture}
    \caption{The construction from Theorem~\ref{thm:auto_tw4}, and the corresponding moves of the reading heads encoding the move of the token $\tau_2$ from $4$ to $5$. Here $V_1=\{1,2,3\}$ and $V_2=\{4,5\}$.}
    \label{fig:stars}
    \end{figure}
    
    Between any middle cell and the center cell of the same tape, there is a path, called \emph{center} path. The cells not in a center path induce paths \emph{pending} from a middle cell. The path of $T_i$ (resp. $T^*$) pending from a middle cell associated with $v\in V_i$ (resp. a token $\tau$) is also associated with $v$ (resp.~$\tau$). 

    The alphabet $\Sigma$ contains $k$ characters $\{ 1,\ldots,k \}$ plus one additional character $\checkmark$. The content of every cell is defined as follows:
    \begin{itemize}
      \item For every $i \le k$, we add the letter $i$ to the content of all middle cells and all cells in a pending path in $T_i$.
      \item In each tape $T_1,\ldots, T_k$, for $i,j\le n$, we add $\checkmark$ on the cells numbered with $i+1$ in a path associated with $v_j$ if and only if $v_iv_j\in E(G)$.
      \item For every $i \le k$, the letter $i$ is added in $T^*$ to the content of the middle cell associated with token $i$, and to its pending path.
      \item Every cell numbered with $1$ in $T^*$ receives $\checkmark$.
  \end{itemize}

Note that the second item is reminiscent from the construction in Theorem~\ref{thm:w2} (that allowed us to check that a set of vertices is dominating). But here instead of selecting one path amongst a tuple, we connect these paths in a star-like way in order to allow two moves (i) go from one vertex to another (via the center of the star) or, (ii) stay at the same vertex (going back and forth in a pending path). The additional tape $T^*$ will ensure that if we decide to go to the center of the star in some tape $T_i$ then, in all the other tapes, we are staying on pending paths.

To complete the construction we simply have to choose the initial and target configurations. 
Let $v_{i_1},\ldots,v_{i_k}$ be the vertices of $D_s$ (resp. $D_t$). 
We denote the configuration $C_s$ (resp. $C_t$) as the middle cell associated with $v_{i_j}$ in each tape $T_j$ plus the middle cell of the tape $T^*$ associated with token $1$. The instance of \SyncAutRec{} is then $(\Sigma,\{T_1,\ldots,T_k,T^*\},C_s,C_t)$.

We show that the reduction from \textsc{TJ-Partitioned-DSR} to \SyncAutRec{} described above is a parameterized logspace PL-reduction. First, observe that the parameter in the \textsc{TJ-Partitioned-DSR} problem is $k$, the number of vertex partitions $V_1, \ldots, V_k$ (and hence the size of each dominating set $D_s, D_t$). In our construction, the number of tapes in the \SyncAutRec{} instance is $k+1$, and each tape is defined according to the input sets $V_i$ and the structure of $G$, thus the new parameter remains bounded by a computable function of $k$ (in this case, $k+1$). 

To establish that the reduction is computable in logarithmic space, we note that each tape $T_i$ and $T^*$ is constructed in a highly regular fashion, depending only on local adjacency information in $G$ and the fixed structure of the vertex partitions. Specifically, each tape is constructed as a star where branches correspond to vertices in $V_i$, and the content of each cell is determined based on adjacency between vertices (for $\checkmark$ placements), or by constant rules tied to vertex or token indices (for character labels). Since the graph $G$ and the sets $V_i$ are given explicitly in the input, and adjacency queries (e.g., whether $v_iv_j \in E(G)$) can be answered by scanning the input using two pointers, the construction of the content of any given cell in the output can be computed using logarithmic space. Additionally, the overall encoding of the tapes does not require storing more than constant or logarithmic-size counters or indices at any point, as we can stream over the input multiple times to emit different portions of the output structure.

Before diving into the technical content of the proof let us explain the intuition.
\begin{itemize}
    \item At the beginning all the read heads are on middle cells. In order to move from one middle cell (say in $T_i$) to another, the read head has to go through the center. To do so, the letter $i$ has to be covered, and this can only be done by a middle cell in $T^*$. So as long as we did not put the reading head back on a middle cell, the reading head in $T^*$ is forced to stay on the $i$-th middle cell in $T^*$. This can be interpreted as  the token assigned to $V_i$ is currently moving.
    \item In particular the previous item ensures that all the reading heads should be on middle cells or on their pending paths in all but at most one tape.
    \item As in Theorem~\ref{thm:w2}, the paths and the synchronization ensure that whenever read heads move from a middle cell towards another, the corresponding vertices of $G$ form a dominating~set.
\end{itemize}

Note that, given a configuration $C$ where all the reading heads lie on middle cells, we can naturally associate a subset $V(C)$ of size $k$ containing exactly one vertex in each $V_i$. Such configurations are called \emph{middle} configurations. We now prove the soundness of the reduction and start with the simpler direction, namely:

\begin{lemma}
    If $(G,D_s,D_t)$ is a positive instance of then $(\Sigma,\{T_1,\ldots,T_k,T^*\},C_s,C_t)$ is a positive instance of \SyncAutRec{}.
\end{lemma}

\begin{proof}
Let $\D=D_1,\ldots,D_\ell$ be a token jumping reconfiguration sequence from $D_s=D_1$ to $D_t=D_\ell$ where every set $D_j$ contains exactly one vertex in each $V_i$. For every dominating set $D_j$ in the sequence, let $C_j$ be the middle configuration containing the first middle cell of $T^*$ and, for $i\in[1,k]$, the middle cell of $T_i$ corresponding to the vertex of $D_j\cap V_i$. Note that $D_j=V(C_j)$ for every $j\in[1,\ell]$, and in particular $C_s=C_1$ and $C_t=C_\ell$. To conclude, we provide a transformation from $C_j$ into $C_{j+1}$ for every $j<\ell$. We will first prove the following: 

\begin{claim}
    Let $C$ be a middle configuration such that $V(C)$ is a dominating set. We can reach from $C$:
    \begin{itemize}
        \item any middle configuration $C'$ with $V(C')=V(C)$ (so where the reading head on $T^*$ lies on another middle cell); and 
        \item any middle configuration $C'$ such that the reading head on $T^*$ lies on the $j$-th middle cell, and $V(C)\cap V_i=V(C')\cap V_i$ for every $i\neq j$.
    \end{itemize}
\end{claim}
\begin{claimproof}
The proof of this claim is inspired by Theorem~\ref{thm:w2}. We consider only the first item, the other being similar. We make the reading heads on $T_1,\ldots, T_k$ go back and forth the pending paths, while the reading head on $T^*$ goes through the center. Note that this is possible using the same process as in Theorem~\ref{thm:w2}.
\end{claimproof}

Now the lemma easily follows from the claim. Indeed, let $j$ be an integer. Assume that $D_{j+1}$ is obtained from $D_j$ by moving a token from $v_a\in V_a$ to $v'_a$. Since $D_j$ and $D_{j+1}$ are dominating sets, the claim ensures that we can reach from $C_j$ a configuration $C'$ where all the reading heads are at the same place except the one of $T^*$ which is on the $a$-th middle cell. Using the second item of the claim, we can reach $C''$ where the reading head on $T_a$ lies on the middle cell $v'_a$. Finally, we can use the claim one last time to get $C_{j+1}$ by moving back the reading head on $T^*$ to the first middle cell. This completes the proof.
\end{proof}

To prove the reduction is sound, it remains to consider the converse direction.
\begin{lemma}
    If $(\Sigma,\{T_1,\ldots,T_k,T^*\},C_s,C_t)$ is a positive instance of \SyncAutRec{} then so is $(G,D_s,D_t)$.
\end{lemma}

\begin{proof}
We first use a slight generalization of middle configurations (that are still called middle); we allow the reading heads to be at distance at most one from a middle cell. We say that a configuration $C$ is a \emph{central} configuration if some reading head lies on a center. Note that $C_s$ and $C_t$ are middle configurations.

Let $\R$ be a reconfiguration sequence from $C_s$ to $C_t$. We first claim that middle and central configurations must alternate in $\R$. Since reading heads must go through the center cell when moving from a middle cell to another, we easily get:

\begin{claim}
\label{cl:center}
    Let $C,C'$ be two middle configurations in $\R$ with $V(C)\neq V(C')$. Then there exists a center configuration between $C$ and $C'$ in $\R$. 
\end{claim}

Note that in a central configuration $C$, exactly one reading head lies on a center cell; otherwise, some letter among $\{1,\ldots,k\}$ is not dominated. We denote by $T_C$ the tape where the reading head is on a center cell in $C$. 
\begin{claim}
\label{cl:middle}
    Let $C,C'$ be two central configurations in $\R$ with $T_C\neq T_{C'}$. Then there exists a middle configuration between $C$ and $C'$ in $\R$. 
\end{claim}

\begin{claimproof}
    Without loss of generality, assume that $T_C=T_1$. By definition, the reading head $h_C$ is on a center cell in $C$. Moreover, in $C'$, it must be at distance $1$ from a middle cell. If it still lies on a middle path, then some letter is not dominated in $C'$ (either $1$ if $T_{C'}=T^*$ or $1$ or $j$ if $T_{C'}=T_j$). Therefore, at some point $h_C$ must have moved from a cell numbered $n+1$ to a middle cell. Due to synchronicity, all the other tapes at that point are on cells numbered with $n+1$ or $1$, hence are at distance at most $1$ from a middle cell. In particular, we get a middle configuration between~$C$~and~$C'$.
\end{claimproof}

Consider now the subsequence $\R'=C_1,\ldots,C_p$ of $\R$ restricted to the middle configurations, removing configurations until $V(C_i)\neq V(C_{i+1})$ for every $i<p$. We claim that $V(C_1),\ldots,V(C_p)$ are all dominating sets and this sequence is actually a valid TJ-reconfiguration sequence from $D_s$ to $D_t$. This is summarized in the two following claims.

\begin{claim}
    For $i\in[1,p]$, $V(C_i)$ is a dominating set of $G$.
\end{claim}

\begin{claimproof}
    The proof follows the ideas from Theorem~\ref{thm:w2}. First note that $V(C_1)=D_s$ is a dominating set. 

    Let $i<p$. By Claim~\ref{cl:center}, let $C$ be a central configuration between $C_i$ and $C_{i+1}$. Consider the numbers of the cells that the reading head of $T_C$ went through from $C$ to $C_{i+1}$. It first moves from $1$ to $2$, stays on $2$, then moves on $3$. Therefore, in $C$, all the other heads must be on leaves (with number $2$). 

    In particular, from $C$ to $C_{i+1}$, the heads moved to a middle cell from either a leaf or the center cell. Recall that the paths between these vertices were designed as in Theorem~\ref{thm:w2}, hence the same arguments show that $V(C_{i+1})$ is a dominating set.
\end{claimproof}

\begin{claim}
    For each $i<p$, the intersection of $V(C_i)$ and $V(C_{i+1})$ has size $k-1$.
\end{claim}

\begin{claimproof}
    Assume by contradiction that the intersection is smaller, that is there are two vertices $v,v'$ in $V(C_i)\setminus V(C_{i+1})$. Without loss of generality, we may assume that there is no intermediate middle configuration $C''$ between $C_i$ and $C_{i+1}$ with $V(C'')=V(C_i)$ or $V(C'')=V(C_{i+1})$ (otherwise, we replace $C_i$ or $C_{i+1}$ by $C''$). 
    
    Note that the reading heads on $v$ and $v'$ in $C_i$ has to move to some cell not associated with $\{v,v'\}$. They both have to go through a center cell, which yields two center configurations $C,C'$ between $C_i$ and $C_{i+1}$. Claim~\ref{cl:middle} then ensures that there is a middle configuration $C''$ between $C$ and $C'$, hence between $C_i$ and $C_{i+1}$. This is a contradiction.
\end{claimproof}
This completes the proof of the lemma.
\end{proof}

To complete the proof of the theorem, it remains to study the structure of the extension graph.

\begin{lemma}
    The $1$-structured treewidth (resp. $2$-structured pathwidth) of $(\Sigma,\{T_1,\ldots,T_k,T^*\})$ is at most $3$ (resp $5$).
\end{lemma}

\begin{proof}
    We provide a structured tree (resp. path) decomposition of $I'=(\Sigma\setminus\{\checkmark\},\{T_1,\ldots,T_k,T^*\})$ of width at most $2$ (resp. $4$) and then add $\checkmark$ to every bag.

    Similarly to the previous proof, observe that the center $c$ of $T^*$ is a cut vertex of the extended graph of $I'$, whose removal leaves $k$ connected components $C_1,\ldots, C_k$, such that each $C_i$ contains the alphabet vertex $i$, the tape vertices of $T_i$ and of the connected component of $T^*-c$ containing the $i$-th middle cell. 

    For $i\in[1,k]$, observe that the alphabet vertex $i$ is a cut-vertex of $C_i\cup\{c\}$, and each connected component obtained when removing $i$ contains only vertices from one tape, hence is a subdivided star of treewidth $1$ and pathwidth $2$. 
    
    Therefore, we can take a tree (resp. path) decomposition of each of these parts and add $i$ in all bags. We can then glue all of them together in a path-like way, adding $c$ to all bags, to get a $2$-structured path decomposition of width $4$. 
    
    For tree decomposition, we glue in a more careful way, by first creating a bag containing only $\{i\}$, and gluing all the decompositions identifying their bag $\{i\}$. This yields a $1$-structured tree decomposition $A_i$ of $C_i\cup\{c\}$ of width $2$. Note that since $i$ is not adjacent to $c$, we may assume that there is a bag containing only $c$. Now, we glue $A_1,\ldots,A_k$ identifying the bags containing only $\{c\}$, which provides the required decomposition. 
\end{proof}


\subsection{Synchronization lemmas (Proof of Theorem~\ref{thm:sync_to_async})}\label{sec:sync_to_async}

Theorem~\ref{thm:sync_to_async} relies on the following synchronization lemma.

\begin{lemma}
\label{lem:sync}
    Let $I=(\Sigma,\{T_1,\ldots,T_k\}, C_s,C_t)$ be an instance of \SyncAutRec. Then there is an \textsf{FPT}-reduction (and PL-reduction) to an irreducible instance $I'$ of \AutRec{} with $k+1$ tapes $T'_1,\ldots,T'_k,T$ such that:
    \begin{itemize}
        \item Each $T'_i$ is obtained from $T_i$ by adding to each cell at most $2$ letters from a set of $3$ fresh letters $a_i,b_i,c_i$.
        \item $T$ is a triangle.
        \item $I$ is a positive instance of \SyncAutRec{} if and only if $I'$ is a positive instance of \AutRec. 
        \item If $I$ has $\struct$-structured treewidth (resp. pathwidth) $\pwtw$, then $I'$ has $(\struct+1)$-structured treewidth (resp. pathwidth) at most $\pwtw+3\struct+3$.
        \item If $I$ is $\degen$-degenerate, then $I'$ is $(\degen+2)$-degenerate.
        \item If $I$ has a feedback vertex set of size $\fvs$, then $I'$ has a feedback vertex set of size at most $\fvs+3k+3$.
    \end{itemize}
\end{lemma}

\begin{proof}
Let $I=(\Sigma,\T,C_s,C_t)$ be an instance of \SyncAutRec{} with $\T=\{T_1,\ldots,T_k\}$. Let $\A=\{a_1,\ldots,a_k\}$, $\B=\{b_1,\ldots,b_k\}$ and $\C=\{c_1,\ldots,c_k\}$ be three sets of fresh letters.

For $i\in[1,k]$, we form $T'_i$ by adding the letters $a_i,b_i,c_i$ to the content of some cells in $T_i$. More precisely, letter $a_i$ (resp. $b_i,c_i$) is added on each cell whose number in $T_i$ is not $2$ (resp. $0,1$) modulo $3$. Now the remaining tape $T$ is a triangle whose cells contain $\A\cup\B$, $\B\cup \C$ and $\C\cup\A$ respectively. 

We show that the above transformation is a parameterized logspace (PL) reduction. For each $i \in [1,k]$, the construction of $T_i'$ involves adding the fresh symbols $a_i$, $b_i$, and $c_i$ to cells of $T_i$ according to their positions modulo $3$. Since computing a number modulo $3$ and assigning symbols based on this value requires only constant time and logarithmic space, and since each tape can be scanned sequentially, this modification is computable in logarithmic space with respect to the input size. The additional tape $T$, which consists of three cells arranged in a triangle containing fixed combinations of the disjoint symbol sets $\mathcal{A} \cup \mathcal{B}$, $\mathcal{B} \cup \mathcal{C}$, and $\mathcal{C} \cup \mathcal{A}$, is also logspace computable, as its structure and content are independent of the input size and can be generated using constant space. Therefore, the entire transformation is logspace computable and preserves the parameter up to a computable function, and hence constitutes a valid PL-reduction.

Intuitively, this construction behaves as follows: in order to cover all the letters of $\A \cup \B \cup C$, we will only be able to use one reading head in each $T_i$ plus one in $T$ (on a cell whose content will be $\A\cup \B$, $\B\cup \C$ or $\A\cup \C$). When the reading head $h$ of $T$ is on the vertex that misses the letters from $\B$, then all the other reading heads have to be on cells numbered with $1$ or $2$ mod $3$, so they cannot move further apart until we move $h$ (which is possible only when all the other reading heads lie on cells with the same number mod $3$). Let us formalize this argument:

Recall that $C_s$ and $C_t$ are numbered configurations. In particular, placing the read heads on $T'_1,\ldots,T'_k$ as in $C_s$ (resp.  $C_t$) allows to cover exactly $\A\cup \B$, $\B\cup \C$ or $\C\cup \A$. To obtain $C'_s$ (resp. $C'_t$), we place the reading head of $T$ on any cell covering the missing letters of $\A\cup\B\cup \C$. Let us denote by $\Sigma'=\Sigma \cup \A \cup \B \cup \C$ and $\T'$ the set of tapes $\T \cup \{T\}$, so that $I'=(\Sigma',\T',C'_s,C'_t)$ is an instance of \AutRec. Note that $I'$ is irreducible.

\begin{figure}
\centering

\begin{tikzpicture}[scale=0.8,thick]

\node at (3, 0) {\SyncAutRec{}};

\foreach \i in {1,2,3} {
    \node[circle, draw, minimum size=0.5cm] (L\i-0) at (0, -\i*1.75) {};
    \node[circle, draw, minimum size=0.5cm] (L\i-1) at (1.5, -\i*1.75) {};
    \node[circle, draw, minimum size=0.5cm] (L\i-2) at (3, -\i*1.75) {};
    \node[circle, draw, minimum size=0.5cm] (L\i-3) at (4.5, -\i*1.75) {};
    \node (L\i-4) at (6, -\i*1.75) {};
    
    \draw (L\i-0) -- (L\i-1) -- (L\i-2) -- (L\i-3);
    \draw[dashed] (L\i-3) -- (L\i-4);
}

\node[left=0.2cm of L1-0] {$T_1$};
\node[left=0.2cm of L2-0] {$T_2$};
\node[left=0.2cm of L3-0] {$T_k$};

\node at (12, 0) {\AutRec};

\foreach \i in {1,2,3} {
    \node[circle, draw, minimum size=0.5cm] (R\i-0) at (9, -\i*1.75) {};
    \node[circle, draw, minimum size=0.5cm] (R\i-1) at (10.5, -\i*1.75) {};
    \node[circle, draw, minimum size=0.5cm] (R\i-2) at (12, -\i*1.75) {};
    \node[circle, draw, minimum size=0.5cm] (R\i-3) at (13.5, -\i*1.75) {};
    \node(R\i-4) at (15, -\i*1.75) {};
    
    \draw (R\i-0) -- (R\i-1) -- (R\i-2) -- (R\i-3);
    \draw[dashed] (R\i-3) -- (R\i-4);
}

\node[left=0.2cm of R1-0] {$T'_1$};
\node[left=0.2cm of R2-0] {$T'_2$};
\node[left=0.2cm of R3-0] {$T'_k$};

\node[above=0.1cm of R1-0] {$a_1,b_1$};
\node[above=0.1cm of R1-1] {$b_1,c_1$};
\node[above=0.1cm of R1-2] {$c_1,a_1$};
\node[above=0.1cm of R1-3] {$a_1,b_1$};

\node[above=0.1cm of R2-0] {$a_2,b_2$};
\node[above=0.1cm of R2-1] {$b_2,c_2$};
\node[above=0.1cm of R2-2] {$c_2,a_2$};
\node[above=0.1cm of R2-3] {$a_2,b_2$};

\node[above=0.1cm of R3-0] {$a_r,b_r$};
\node[above=0.1cm of R3-1] {$b_r,c_r$};
\node[above=0.1cm of R3-2] {$c_r,a_r$};
\node[above=0.1cm of R3-3] {$a_r,b_r$};

\node[circle, draw, minimum size=0.5cm] (A) at (11, -7) {};
\node[circle, draw, minimum size=0.5cm] (B) at (13, -7) {};
\node[circle, draw, minimum size=0.5cm] (C) at (12, -8.7) {};

\draw (A) -- (B) -- (C) -- (A);

\node[left=0.1cm of A] {$\A \cup \B$};
\node[right=0.1cm of B] {$\B \cup \C$};
\node[right=0.1cm of C] {$\C \cup \A$};
\node[left=1.8cm of A] {$T$};

\end{tikzpicture}
    \caption{Illustration of the reduction between \SyncAutRec{} and \AutRec{}.}
    \label{fig:enter-label}
\end{figure}
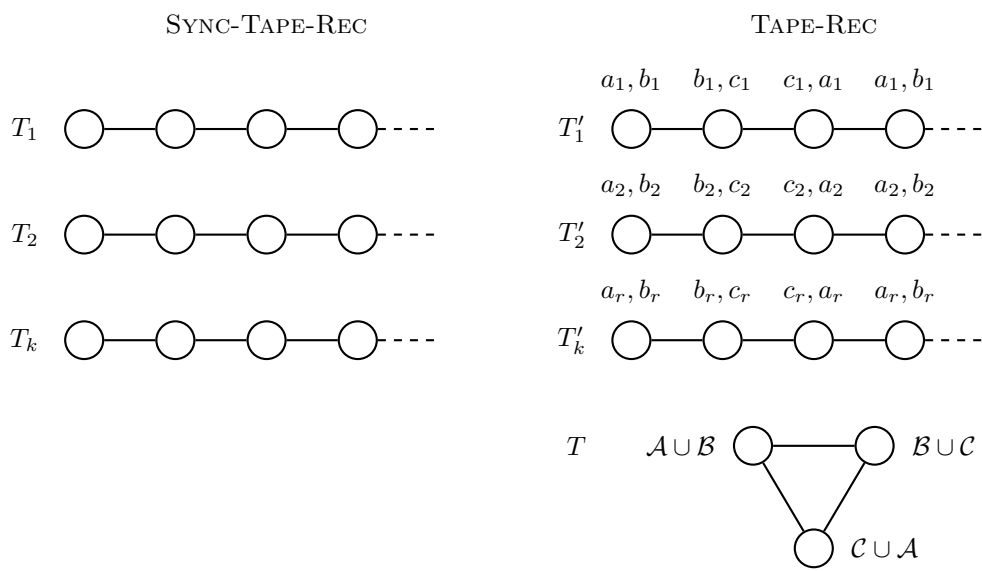

\begin{claim}
\label{cl:synctotape}
    If $I$ is a positive instance of \SyncAutRec{} then so is $I'$ for \AutRec{}.
\end{claim}
\begin{claimproof}
    Let $C_s=C_1,\ldots,C_p=C_t$ be a sequence of valid configurations for $I$ between the start and the end cells of the tapes. First observe that we can construct from each $C_i$ a valid configuration $C'_i$ of $I'$ by putting the reading heads of the tapes $T'_1,\ldots,T'_k$ on the same cells as in $T_1,\ldots, T_k$ for $C_i$. Since $C_i$ is synchronized, all the reading heads lie on cells numbered with $j$ or $j+1$ for some $j$. Then all the heads in $T'_1,\ldots, T'_k$ cover at least $\A$, $\B$ or $\C$. We then put the head on $T$ so that the remaining letters of $\A\cup\B\cup\C$ are covered. Note that there may be two choices when all cells under the heads in $C_i$ have the same number, and we then choose arbitrarily. Note that in particular, $C'_1=C'_s$ and $C'_p=C'_t$.

    Let $i<p$. We now exhibit a valid transformation from each $C'_i$ to $C'_{i+1}$. The only case to consider is when their reading heads in $T$ lie on different cells, say $\A\cup\B$ and $\A\cup\C$. In particular, in $C_i$ (resp. $C_{i+1}$), the reading head on $T'_j$ must cover $c_j$ (resp. $b_j$). By symmetry, assume that the reading head in $T_1$ moved from cell $v$ to $v'$ when going from $C_i$ to $C_{i+1}$. Since the other heads do not move from $C_i$ to $C_{i+1}$, the reading heads in $T'_j$ for $j>1$ must cover $\{b_j,c_j\}$, hence their cell number is $2$ mod $3$ in $C_i$ and $C_{i+1}$. In particular, they all have the same number, denoted by $c$. 
    We prove that either the reading head on $T$ can be moved to $\A\cup \C$ from $C'_i$ or the reading head on $T_1$ can be moved to $v'$  from $C'_i$.    
    Assume that the reading head on $T$ cannot be moved to $\A\cup \C$ from $C'_i$. This means that cell $v$ does not contain $b_1$. So $v$ is numbered with a $0$ mod $3$ value lying in $\{c,c+1,c-1\}$ (since $C_i$ is synchronized), hence with $c+1$. So $v'$ is numbered with either $c+2$ or $c$ and since all other heads are numbered $c$, $v'$ is also numbered $c$. Since $c$ equals $2$ mod $3$, it contains the label $c_1$ and so, $v$ can be moved to $v'$ from $C'_i$. 
\end{claimproof}

\begin{claim}
\label{cl:tapetosync}
    If $I'$ is a positive instance of \AutRec{} then so is $I$ for \SyncAutRec.
\end{claim}
\begin{claimproof}
    Let $C'_s=C'_1,\ldots,C'_p=C'_t$ be a sequence of valid configurations for $I'$. Since the tapes $T'_i$ were obtained by adding some letters to some cells, moving the reading heads in the same way on $T_1,\ldots,T_p$ yields a sequence of valid configurations $C_s=C_1,\ldots,C_p=C_t$ for $I$. We only have to prove that these configurations are synchronized.
    
    Note that $C_1=C_s$ is synchronized. Aiming towards a contradiction, let $j>1$ be the smallest index such that $C_j$ is not synchronized. Up to symmetry, this means that the reading head on $T_1$ is on a cell $u$ numbered $c$, while the reading head on $T_2$ is on a cell $v$ numbered $c+2$. Moreover, since $j$ is minimum, $C_{j-1}$ is synchronized, hence we can assume that the reading head in $T_2$ was on a cell $v'$ numbered with $c+1$ in $C_{j-1}$ and moved to obtain $C_j$. 

    Without loss of generality, assume that $c=0$ mod $3$. In particular, the cell $u$ in $T'_1$ contains $\{a_1,c_1\}$ and the cell $v'$ in $T'_2$ contains $\{a_2,b_2\}$, forcing the reading head of $T$ in $C'_{j-1}$ to lie on $\B\cup\C$. Similarly, $v$ contains $\{b_2,c_2\}$, forcing the head of $T$ in $C'_j$ to lie on $\A\cup\B$. But this is impossible since only one reading head could have moved from $C'_{j-1}$ to $C'_j$. 

    Therefore, all the $C_j$ are synchronized, hence $I$ is a positive instance of \SyncAutRec.
\end{claimproof}

To conclude, it remains to show that the structured width of the instances is somewhat preserved. Indeed, the other structural properties are easily satisfied; adding the $3k$ new alphabet vertices and the tape vertices from $T$ to a feedback vertex set of $I$ yields a feedback vertex set of $I'$. Moreover, following a $d$-degeneracy ordering in $I$ yields a $(d+2)$-degeneracy ordering in $I'$.

\begin{claim}
    If $I$ has $\struct$-structured treewidth (resp. pathwidth) $\pwtw$, then $I'$ has $(\struct+1)$-structured treewidth (resp. pathwidth) at most $\pwtw+3\struct+3$.
\end{claim}

\begin{claimproof}
    Consider an $\struct$-structured decomposition of $I$ of width $\pwtw$. We add the alphabet vertices $a_i,b_i,c_i$ to all the bags containing a tape vertex from $T_i$, and the three tape vertices of $T$ to all the bags. Observe that now the bag size increased by at most $3\struct+3$, and each bag contains tape vertices from at most $\struct+1$ tapes ($T$ plus at most $\struct$ of the $T_i$'s). 
\end{claimproof}

This concludes the proof of the lemma.
\end{proof}

When we have a more structured instance, for example an instance of \SyncPathAutRec, we can actually preserve the path structure with another synchronizing tape.

\begin{lemma}
\label{lem:syncpath}
     Let $I=(\Sigma,\{T_1,\ldots,T_k\}, C_s,C_t)$ be an instance of \SyncPathAutRec. Then there is an \textsf{FPT}-reduction (and PL-reduction) to an instance $I'$ of \PathAutRec{} with $k+1$ tapes $T'_1,\ldots,T'_k,T$ such that:
    \begin{itemize}
        \item Each $T'_i$ is obtained from $T_i$ by adding $3$ fresh letters $a_i,b_i,c_i$.
        \item $I$ is a positive instance of \SyncAutRec{} if and only if $I'$ is a positive instance of \AutRec. 
    \end{itemize}
\end{lemma}

\begin{proof}
    The proof is essentially similar in that case. The tapes $T'_i$ are constructed as before. The tape $T$ is now a path tape of length $q$, where $q$ is the maximum number of cells in $I$. The cell contents alternate as follows: $\C\cup \A, \A\cup\B,\B\cup \C,\C\cup \A,\ldots$. 
    
    Now the proof of Claim~\ref{cl:tapetosync} follows verbatim. For Claim~\ref{cl:synctotape}, the only change is that the configurations $C'_i$ for $I'$ are constructed by putting the reading head of $T$ on the $j$-th cell, where $j$ is the smallest number of a cell under a reading head in the corresponding configuration $C_i$ of $I$.  
\end{proof}


\section{Hardness results of tape reconfiguration on paths (Proofs of Section~\ref{sec:hardautrec})}

\subsection{Selection lemmas (Proof of Theorem~\ref{thm:Q-AUTREC})}

The reduction from \Q{} to \PathAutRec{} relies on the following technical lemma. 

\begin{lemma}
\label{lem:sel}
    Let $\T^1,\ldots,\T^k$ be $k$ tuples of path tapes. One can construct in \textsf{FPT} time $k+1$ path tapes $T_1,\ldots,T_k,T$ such that:
    \begin{itemize}
        \item $T$ is a path on $5$ vertices.
        \item Each $T_i$ has the same alphabet as the tapes in $\T^i$, plus three new letters.
        \item $\{\T^1,\ldots,\T^k\}$ is a positive instance of \Q{} if and only if $\{T_1,\ldots,T_k,T\}$ is a positive instance of \PathAutRec.
    \end{itemize}
\end{lemma}

\begin{proof}
    Let $(\Sigma, \{\T^1,\ldots,\T^k\})$ be an instance of \Q{}. We define $T_1,\ldots,T_k,T$ as follows: 
    For $1\leq i \leq k$, in each tape of $\T^i$, we add a fresh letter $a_i$ to all the cells, a fresh letter $s_i$ to the start cell and a fresh letter $e_i$ to the end cell. Then, $T_i$ is the concatenation of the $\T_j^i$, adding for each $j$ an empty cell adjacent to the end cell of $\T_j^i$ and the start cell of $\T_{j+1}^i$. 
    
    The tape $T$ is a path on $5$ cells whose contents are $\Sigma \cup \A \cup \S \cup \E$, $\Sigma \cup \A \cup \E$, $ \S \cup \E$, $\Sigma \cup \A \cup \S$ and $\Sigma \cup \A \cup \S \cup \E$,  where $\A=\{a_1,\ldots,a_k\}$, $\S=\{s_1,\ldots,s_k\}$ and $\E=\{e_1,\ldots,e_k\}$, see Figure~\ref{fig:selector}. Let $\Sigma'=\Sigma\cup \A\cup\S\cup \E$.
    
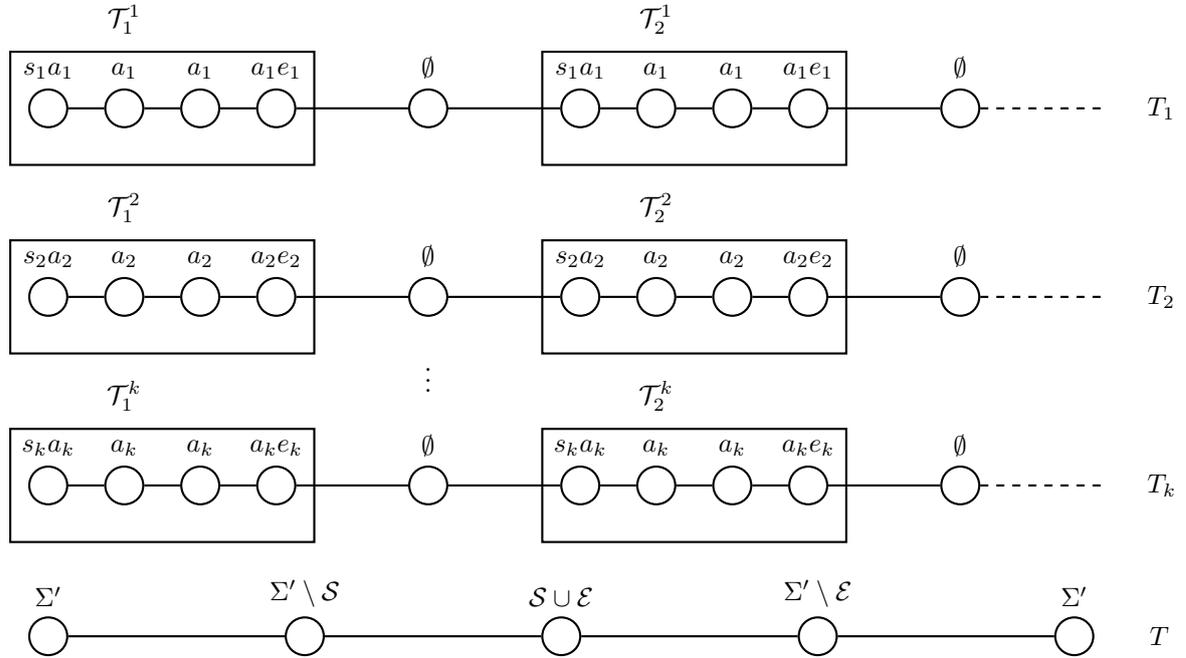
\begin{figure}[!ht]
\centering
\begin{tikzpicture}[thick]

\tikzstyle{state}=[circle, draw, minimum size=0.5cm, inner sep=0]
\tikzstyle{box}=[rectangle, draw, minimum height=1.5cm, minimum width=4cm]

\node[box] (T1_1) at (3,0) {};
\node[state, label=above:$s_1 a_1$] (T1_1_1) at (1.5,0) {};
\node[state, label=above:$a_1$] (T1_1_2) at (2.5,0) {};
\node[state, label=above:$a_1$] (T1_1_3) at (3.5,0) {};
\node[state, label=above:$a_1 e_1$] (T1_1_4) at (4.5,0) {};
\node[state, label=above:$\emptyset$] (T1_1_end) at (6.5,0) {};
\node[box] (T1_2) at (10,0) {};
\node[state, label=above:$s_1 a_1$] (T1_2_1) at (8.5,0) {};
\node[state, label=above:$a_1$] (T1_2_2) at (9.5,0) {};
\node[state, label=above:$a_1$] (T1_2_3) at (10.5,0) {};
\node[state, label=above:$a_1 e_1$] (T1_2_4) at (11.5,0) {};
\node[state, label=above:$\emptyset$] (T1_2_end) at (13.5,0) {};
\node (T1_t_1) at (15.5,0) {};

\node[above=.6cm of T1_1_2] {$\T_1^1$};
\node[above=.6cm of T1_2_2] {$\T_2^1$};
\node[right=0.2cm of T1_t_1] {$T_1$};

\draw (T1_1_1) -- (T1_1_2);
\draw (T1_1_2) -- (T1_1_3);
\draw (T1_1_3) -- (T1_1_4);
\draw (T1_1_4) -- (T1_1_end);
\draw (T1_1_end) -- (T1_2_1);
\draw (T1_2_1) -- (T1_2_2);
\draw (T1_2_2) -- (T1_2_3);
\draw (T1_2_3) -- (T1_2_4);
\draw (T1_2_4) -- (T1_2_end);
\draw[dashed] (T1_2_end) -- (T1_t_1);

\tikzset{yshift=.5cm}
\node[box] (T2_1) at (3,-3) {};
\node[state, label=above:$s_2 a_2$] (T2_1_1) at (1.5,-3) {};
\node[state, label=above:$a_2$] (T2_1_2) at (2.5,-3) {};
\node[state, label=above:$a_2$] (T2_1_3) at (3.5,-3) {};
\node[state, label=above:$a_2 e_2$] (T2_1_4) at (4.5,-3) {};
\node[state, label=above:$\emptyset$] (T2_1_end) at (6.5,-3) {};
\node[box] (T2_2) at (10,-3) {};
\node[state, label=above:$s_2 a_2$] (T2_2_1) at (8.5,-3) {};
\node[state, label=above:$a_2$] (T2_2_2) at (9.5,-3) {};
\node[state, label=above:$a_2$] (T2_2_3) at (10.5,-3) {};
\node[state, label=above:$a_2 e_2$] (T2_2_4) at (11.5,-3) {};
\node[state, label=above:$\emptyset$] (T2_2_end) at (13.5,-3) {};

\node (T2_t_1) at (15.5,-3) {};

\node[above=.6cm of T2_1_2] {$\T_1^2$};
\node[above=.6cm of T2_2_2] {$\T_2^2$};
\node[right=0.2cm of T2_t_1] {$T_2$};

\draw (T2_1_1) -- (T2_1_2);
\draw (T2_1_2) -- (T2_1_3);
\draw (T2_1_3) -- (T2_1_4);
\draw (T2_1_4) -- (T2_1_end);
\draw (T2_1_end) -- (T2_2_1);
\draw (T2_2_1) -- (T2_2_2);
\draw (T2_2_2) -- (T2_2_3);
\draw (T2_2_3) -- (T2_2_4);
\draw (T2_2_4) -- (T2_2_end);
\draw[dashed] (T2_2_end) -- (T2_t_1);

\node at (6.5,-4) {$\vdots$};
\tikzset{yshift=1.5cm}
\node[box] (Tk_1) at (3,-7) {};
\node[state, label=above:$s_k a_k$] (Tk_1_1) at (1.5,-7) {};
\node[state, label=above:$a_k$] (Tk_1_2) at (2.5,-7) {};
\node[state, label=above:$a_k$] (Tk_1_3) at (3.5,-7) {};
\node[state, label=above:$a_k e_k$] (Tk_1_4) at (4.5,-7) {};
\node[state, label=above:$\emptyset$] (Tk_1_end) at (6.5,-7) {};
\node[box] (Tk_2) at (10,-7) {};
\node[state, label=above:$s_k a_k$] (Tk_2_1) at (8.5,-7) {};
\node[state, label=above:$a_k$] (Tk_2_2) at (9.5,-7) {};
\node[state, label=above:$a_k$] (Tk_2_3) at (10.5,-7) {};
\node[state, label=above:$a_k e_k$] (Tk_2_4) at (11.5,-7) {};
\node[state, label=above:$\emptyset$] (Tk_2_end) at (13.5,-7) {};
\node (Tk_t_1) at (15.5,-7) {};

\node[above=.6cm of Tk_1_2] {$\T^k_1$};
\node[above=.6cm of Tk_2_2] {$\T^k_2$};
\node[right=0.2cm of Tk_t_1] {$T_k$};

\draw (Tk_1_1) -- (Tk_1_2);
\draw (Tk_1_2) -- (Tk_1_3);
\draw (Tk_1_3) -- (Tk_1_4);
\draw (Tk_1_4) -- (Tk_1_end);
\draw (Tk_1_end) -- (Tk_2_1);
\draw (Tk_2_1) -- (Tk_2_2);
\draw (Tk_2_2) -- (Tk_2_3);
\draw (Tk_2_3) -- (Tk_2_4);
\draw (Tk_2_4) -- (Tk_2_end);
\draw[dashed] (Tk_2_end) -- (Tk_t_1);

\tikzset{yshift=1cm}
\node[state, label=above:$\Sigma'$] (S1) at (1.5,-10) {};
\node[state, label=above:$\Sigma'\setminus\S$] (S2) at (4.875,-10) {};
\node[state, label=above: $\S \cup \E$] (S3) at (8.25,-10) {};
\node[state, label=above:$\Sigma'\setminus\E$] (S4) at (11.625,-10) {};
\node[state, label=above:$\Sigma'$] (S5) at (15,-10) {};
\node[right=.6cm of S5] {$T$};

\draw (S1) -- (S2);
\draw (S2) -- (S3);
\draw (S3) -- (S4);
\draw (S4) -- (S5);

\end{tikzpicture}
    \caption{The instance of \PathAutRec{} reduced from the instance of \Q.}
    \label{fig:selector}
\end{figure}

\begin{claim}
\label{cl:seltopath}
    If $(\Sigma,\{\T^1,\ldots,\T^k\})$ is a positive instance of \Q{} then  $(\Sigma',\{T_1,\ldots,T_k,T\})$ is a positive instance of \PathAutRec{}.
\end{claim}

\begin{claimproof}
If $(\Sigma,\{\T^1,\ldots,\T^k\})$ is a positive instance of \Q, let $\T_{p_1}^1,\ldots,\T^k_{p_k}$ be tapes such that $\{\T_{p_1}^1,\ldots,\T^k_{p_k}\}$ is a positive instance of \PathAutRec{}. We describe a sequence of transformations for $\{T_1,\ldots,T_k,T\}$. Initially, the head on $T$ covers the whole alphabet $\Sigma'$ so we can move freely the other heads. In each tape $T_i$, we put the reading head on the first cell of $\T_{p_i}^i$. Since $ \{\T^1_{p_1},\ldots,\T^k_{p_k}\}$ is a positive instance of \AutRec{}, $\Sigma$ is now covered by the heads in $T_1,\ldots, T_k$, and so is $\A \cup \S$ by construction. We move the head of $T$ to the third cell. Then $\S \cup \E$ is covered so we can move all the reading heads on each $T_i$ to the end cell of $\T^i_{p_i}$, according to the transformation for \PathAutRec; at each step $\S \cup \E$ is covered by the head in $T$ and $\Sigma\cup \A$ is covered by the heads in $T_1,\ldots,T_k$. After this, these heads cover $\Sigma \cup \E$ so we can move the head in $T$ to the end cell, and then move freely the other heads to their respective end cell.
\end{claimproof}

\begin{claim}
\label{cl:pathtosel}
    If $(\Sigma',\{T_1,\ldots,T_k,T\})$ is a positive instance of \PathAutRec{} then $(\Sigma,\{\T^1,\ldots,\T^k\})$ is a positive instance of \Q. 
\end{claim}

\begin{claimproof}
    Consider a sequence of valid configurations for the positive instance $(\Sigma,\{T_1,\ldots,T_k,T\}$) of \AutRec{}. Since the head $h$ on $T$ must go from the start to end cell, this sequence must go through four configurations $C_s, C'_s, C'_e,C_e$ such that $h$ moves from the second to third vertex of $T$ from $C_s$ to $C'_s$, stays on the third vertex in all configurations between $C'_s$ and $C'_e$, and finally moves to the fourth vertex from $C'_e$ to $C_e$. 
    
    Between $C'_s$ and $C'_e$, $h$ covers no letter in $\A$, hence for $i\in [1,k]$, the reading head in $T_i$ stays on a part corresponding to $\T_{j_i}^i$ for some index $j_i$. Moreover, $h$ covers no letter in $\S$ in $C_s$, hence the head in $T_i$ must lie on the start cell of $\T_{j_i}^i$ in $C_s$ (and $C'_s$). Similarly, it also lies on the end cell of $\T_{j_i}^i$ in $C_e$ and $C'_e$. Therefore, the moves of the reading heads in $T_1,\ldots,T_k$ between $C'_s$ and $C'_e$ provide a sequence of valid configurations for $\{\T_{j_1}^1,\ldots,\T_{j_k}^k\}$, so $(\Sigma,\{\T^1,\ldots,\T^k\})$ is a positive instance of \Q. 
\end{claimproof}

This completes the proof of the lemma. 
\end{proof}

\subsection{The case of path tape reconfiguration (Proof of Theorem~\ref{thm:wstarhard})}\label{sec:proofw*}

\thmWstarHard*

We already know from Theorems~\ref{thm:w2} and~\ref{thm:Q-AUTREC} that \Q{} is \textsf{W[2]}-hard. To prove Theorem~\ref{thm:wstarhard}, we introduce the problems \textsc{Or}-\Q{} and \textsc{And}-\Q. The former is defined in a straightforward way:

\medskip
\noindent
\textsc{Or}-\Q{} \\
\textbf{Input:} A set of $p$ instances $\{I_1,\ldots,I_p\}$ of $\Q$, each of them being $k$ tuples of tapes over the same alphabet $\Sigma$.\\
\textbf{Parameter:} $k+|\Sigma|$\\
\textbf{Output:} Yes if $I_j$ is a positive instance of $\Q$ for at least one value of $j\in[1,p]$.\medskip

One could think of defining \textsc{And}-\Q{} similarly, asking for all the $I_j$ to be positive instances, but this is not sufficient for our purposes. We actually need the instances to have a common solution, hence we define \textsc{And}-\Q{} as follows:

\medskip
\noindent
\textsc{And}-\Q{} \\
\textbf{Input:} A set of $p$ instances $\{I_1,\ldots,I_p\}$ of $\Q$, each of them being $k$ tuples of tapes over the same alphabet $\Sigma$.\\
\textbf{Parameter:} $k+|\Sigma|$\\
\textbf{Output:} Yes if there is a common solution $i_1,\ldots,i_k$ to all the instances $I_1,\ldots,I_p$.\medskip

We claim that these two problems can be \textsf{FPT}-reduced to \Q. Before proving the claim, let us discuss why it implies Theorem~\ref{thm:wstarhard}. We construct, by induction on $h \geqslant 2$, a reduction $\Phi_h$ from the \textsc{Weighted-$h$-Normalized-Sat} problem to the \Q\  problem. In other words, for every $h$-normalized formula $\varphi$ and every $k$, there is an equivalent instance $\Phi_h(\varphi,k)$ of \Q.
For the base case, observe that we already know $\Phi_2$ since \Q{} is \textsf{W[2]}-hard.
We now explain how to construct $\Phi_{h+1}$ from $\Phi_h$.

Let $\varphi$ be an $(h+1)$-normalized boolean formula and $k\in\mathbb{N}$. By definition, $\varphi$ can be written as $\bigwedge_{i=1}^a \bigvee_{j=1}^b\psi_{i,j}$, for some $a$ and $b$, where each $\psi_{i,j}$ is an $h$-normalized formula. In particular, we can construct $a \cdot b$ instances $\Phi_h(\psi_{i,j},k)$ of \Q. Now observe that for every $i\in[1,a]$, $\bigvee_{j=1}^b\psi_{i,j}$ is satisfied by a weight-$k$ truth assignment if and only if $\{\Phi_h(\psi_{i,j},k)\mid j\in[1,b]\}$ is a positive instance of \textsc{Or}-\Q. Since this problem can be reduced to \Q, we can get an equivalent instance $I_i$ of \Q, $i\in[1,a]$. 

Now, $\varphi$ is a positive instance of \textsc{Weighted-$(h+1)$-Normalized-Sat}  if and only if $\{I_1,\ldots,I_a\}$ is a positive instance of \textsc{And}-\Q. Hence, using the reduction from \textsc{And}-\Q{} to \Q, we can construct an equivalent instance $\Phi_{h+1}(\varphi,k)$ of \Q, which concludes the induction. 

It remains to prove the claim, which we prove in two lemmas.

\begin{lemma}
\label{lem:redand}
    There is an \textsf{FPT}-reduction from \textsc{And}-\Q{} to \Q. 
\end{lemma}

\begin{proof}
Let $I_1,\ldots,I_p$ be $p$ instances of \Q, each of them being $k$ tuples of tapes over the same alphabet $\Sigma$.
We define an instance $I$ of \Q, that will be numbered, consisting of $k$ tuples as follows: the $i$-th component of the $j$-th tuple $I$ is obtained by gluing together the $i$-th component of the $j$-th tuple of all instances similarly to the construction from Lemma~\ref{lem:sel}. 
This time, we duplicate the start and end cell of each tape and add new cell of content $\Sigma$, connected to the copy of the end cell of a tape and to the copy of the start cell of the next tape. 

Then, for every $j$, we number the cells of the $j$-th instance with $4j-2$, its duplicated start (resp. end) cells with $4j-3$ (resp. $4j-1$), and the added cell with content $\Sigma$ with $4j$. This gives a numbered instance $I$ of \Q, and we can apply the same technique used in the synchronization Lemma~\ref{lem:syncpath} to get an unsynchronized instance $I'$ of \Q. This adds a $1$-tuple (so $I'$ has size $k+1$), containing only the tape created by the lemma. Note that here the fresh letters introduced are shared among all tapes of each tuple of $I$. 

\begin{claim}
    If there exists a common solution $i_1,\ldots,i_k$ to all the instances $I_1,\ldots,I_p$ then $I'$ is a positive instance for $\Q$. 
\end{claim}

\begin{claimproof}
    Assume $(i_1,\ldots,i_k)$ is a common solution to $I_1,\ldots,I_p$ and let us prove $(i_1,\ldots,i_k,1)$ is also a solution for $I'$. Let $T_1,\ldots,T_k,T$ be the corresponding tapes, each $T_j$ being the concatenation of the $i_j$-th tapes of the $j$-th tuples in $I_1,\ldots,I_p$, and $T$ being the synchronizing tape as per  Lemma~\ref{lem:syncpath}. 
    
    We now give a synchronized transformation for $T_1,\ldots,T_k$. Assume all the reading heads are all on a duplicated start cells (which is the case on the start configuration). Then we can move all the reading heads on the original start cells since they have the same content. Now, all the reading heads are at the start of a positive instance (whose cells all have the same number) formed by tapes in some $I_i$ so we can move all the reading heads to the corresponding end cells. Again, we move all reading heads to the duplicate end cells, then to the added cells (since they have content $\Sigma$), then to the next duplicated start cells.
\end{claimproof}

\begin{claim}
    If $I'$ is a positive instance for $\Q$ then there exists  common solution $i_1,\ldots,i_k$ to all the instances $I_1,\ldots,I_p$.
\end{claim}

\begin{claimproof}
    Let $(i_1,\ldots,i_k,1)$ be a solution for $I'$. Again, let $T_1,\ldots,T_k$ be the concatenated tapes, and observe that by Lemma~\ref{lem:syncpath}, they have a synchronized transformation. We show that $i_1,\ldots,i_k$ is also a solution for $I_j$ for any $j$. 
    
    Since the transformation is synchronized, at some point, all the reading head lie on cells numbered $4j-3$ (resp. $4j-1$), namely the duplicate start (resp. end) cells of the tapes from $I_j$. In between, the intermediate 
    cells (numbered with $4j-2$) have the same content as in $I_j$ so the sequence of steps in $I'$ provides a valid sequence for $I_j$.
\end{claimproof}
This concludes the proof of the lemma.
\end{proof}

\begin{lemma}
\label{lem:redor}
    There is an \textsf{FPT}-reduction from \textsc{Or}-\Q{} to \Q. 
\end{lemma}

\begin{proof}
Since \Q{} and \PathAutRec{} are equivalent, we reduce \textsc{Or}-\PathAutRec{} to \PathAutRec{}. Let $I_1,\ldots,I_p$ be instances of \PathAutRec{} of the same size $k$ and over the same alphabet $\Sigma$. Then, we construct $k$ tuples, the $j$-th one containing the $j$-th tape of each instance, and apply the following claim. 

\begin{claim}
    Let $\T^1,\ldots,\T^k$ be $k$ tuples of path $\Sigma$-tapes of the same size. One can construct in \textsf{FPT} time $k+2$ path $\Sigma'$-tapes $T_1,\ldots,T_{k+2}$ such that:
    \begin{itemize}
        \item $\{T_1,\ldots,T_{k+2}\}$ is a positive instance of \PathAutRec{} if and only if there exists $j$ such that $\{\T^1_j,\ldots, \T^k_j\}$ is a positive instance.
        \item $|\Sigma'|=|\Sigma|+6k$
    \end{itemize}
\end{claim}

\begin{claimproof}  
    In this proof, we merge ideas from Lemmas~\ref{lem:syncpath} and~\ref{lem:sel}. More precisely, we glue the tapes in each tuple together as in Lemma~\ref{lem:sel}, but we want to progress in a synchronized way on them, so we add a synchronizing tape as in Lemma~\ref{lem:syncpath}.

    Let $T_1,\ldots,T_k,T$ be the tapes given by Lemma~\ref{lem:sel}, except that we duplicate each start and end cell (this is necessary to ensure that we can lift a transformation of $\{\T^1_j,\ldots,\T^k_j\}$ to a synchronized transformation in $\{T_1,\ldots,T_k,T\}$). In particular, we create $3k$ fresh letters. We number these tapes as follows.  Consider the $j$-th tape of some $\T^i$. Each of its cells in $T_i$ gets number $4j-2$. The copy of its start (resp. end) cell gets number $4j-3$ (resp. $4j-1$), and the cell with empty content adjacent to the end cell gets number $4j$. 

    Now, following Claims~\ref{cl:seltopath} and~\ref{cl:pathtosel}, we easily see that $\{\T^1_j,\ldots,\T^k_j\}$ is a positive instance of \PathAutRec{} for some $j$ if and only if there is a transformation in $\{T_1,\ldots,T_k,T\}$ where the heads on $\{T_1,\ldots,T_k\}$ stay synchronized.

    To ensure this is the case, we use the construction from Lemma~\ref{lem:syncpath}, adding three letters on the cells of each $T_i$, together with a new tape $T^*$. Now we get that $\{\T_j^1,\ldots,\T_j^k\}$ is a positive instance of \PathAutRec{} for some $j$ if and only it $\{T_1,\ldots, T_k,T,T^*\}$ is a positive instance of \PathAutRec, as needed.
\end{claimproof}

This concludes the proof of the lemma since, by construction, $I_j$ is a positive instance if and only if $\{\T^1_j,\ldots,\T^k_j\}$ is a positive instance of \PathAutRec.
\end{proof}


\section{\textsf{FPT} algorithms}
The goal of this section is to prove Theorems~\ref{thm:FPT-planar} and~\ref{thm:FPT-minor}. Let us restate them for completeness:

\thmFPTplanar*

\thmFPTminor*

These two statements share similar proofs, inspired from the work of Lokshtanov et al.~\cite{DBLP:journals/jcss/LokshtanovMPRS18} who proved that TJ-\textsc{Dominating Set Reconfiguration} is \textsf{FPT} on nowhere dense graphs and of the proof of Ito et al.~\cite{DBLP:conf/isaac/ItoKO14} who proved that \textsc{Independent Set Reconfiguration} under token jumping is \textsf{FPT} on $K_{3,d}$-free graphs. 

Let $\mathcal{G}_d$ be either the class of $K_{3,d}$-free graphs or the class of $K_{4,d}$-minor-free graphs (we give  additional lemmas depending on the classes later). Note that since in the token sliding model, all the tokens can only move in their own connected component, we can assume that $G$ is connected. Indeed, if $G$ is not, we simply have to solve each connected component independently (or declare a no-instance if the number of tokens per component does not match in the source and target configurations).

A \emph{$k$-domination core} of a graph $G$ (or simply a domination core when $k$ and $G$ are clear from context) is a subset of vertices $X$ such that a subset $D$ of size at most $k$ is a dominating set of $G$ if and only if $D$ dominates $X$. In other words, in order to dominate the graph, one only has to guarantee that $X$ is dominated. 
From~\cite{eiben2019lossy}, all the graphs in $\mathcal{G}_d$ having a dominating set of size $k$ admit a $k$-domination core of size at most $(2d+1)k^{d+1}$ that can be found in polynomial time; starting with $X = V(G)$, and as long as $|X| > (2d+1)k^{d+1}$ one can find (in polynomial time) a vertex $x \in X$ such that $X \setminus \{x\}$ is still a $k$-domination core. We can therefore assume in the rest of the section that each graph in $\mathcal{G}_d$ comes with a domination core. 

Let $G$ be a graph of $\mathcal{G}_d$ and $I=(G,k,D_s,D_t)$ be an instance of \TSDSR.  Note that super-sets of domination cores are still domination cores, so we assume that all upcoming domination cores contain $D_s\cup D_t$ (we need this condition for convenience to ensure that vertices outside of $X$ do not belong to the source and target dominating sets). Fix such a domination core $X$ of $G$. By definition, $I$ is a positive instance of \TSDSR{} if and only if it is a positive instance of the following problem:
\medskip

\noindent
\textsc{Domination Core Reconfiguration} (\dscr)\\
\textbf{Input:} A graph $G$, two dominating sets $D_s$ and $D_t$ of $G$ of size $k$, and subset $X$ of vertices.\\
\textbf{Parameter:} $k$.\\
\textbf{Output:} Does there exist a reconfiguration sequence from $D_s$ to $D_t$ under the token sliding rule that dominates $X$ all along the sequence?
\medskip

Restricting the condition to dominating only the domination core allows us to basically forget about dominating the vertices in $V \setminus X$. Such vertices can then only be useful for connectivity reasons, in order to slide tokens. Let $(G,X,D_s,D_t)$ be an instance of \dscr{} with $G\in\mathcal{G}_d$ and $X$ being a domination core of $G$. In the rest of the section, we shall modify the graph $G$, so $X$ might no longer be a domination core. However, we will still require to dominate $X$ and only $X$.

Lokshtanov et al.~\cite{DBLP:journals/jcss/LokshtanovMPRS18} partitioned the vertices of $G$ into classes according to their neighborhood in the domination core $X$ and proved that one simply has to keep one vertex per class. In the case of token sliding, it is not clear if we are allowed to reduce to one vertex per class. We successfully adapt the domination core technique to the token sliding model, but only on smaller classes of graphs (a necessary restriction given our hardness results), namely $K_{3,d}$ and $K_{4,d}$-minor free graphs.  


For every $Y \subseteq X$, let us denote by $V_Y$ the subset of vertices $v$ of $V \setminus X$ such that $N(v) \cap X = Y$. The vertices of $V_Y$ form the \emph{$Y$-class}. A $Y$-class has \emph{type $p$} if $|Y|=p$. By abuse of notation, we say that a class of type $p$ is a $p$-class, and moreover a $Y$-class is included in a $Y'$-class if $Y \subseteq Y'$. Our goal is to obtain a kernel for \dscr, by reducing the classes via  decreasing their type. 

First, observe that large classes of large type create large complete bipartite subgraphs, and so do vertices with many neighbors in the same class of large type. 
In particular, classes of large type in a graph in $\mathcal{G}_d$ have bounded size, as summarized by the following result.

\begin{lemma}\label{lem:basicKtt}
Let $G$ be a $K_{q,d}$-free graph. Then, every class of type at least $q$ has size less than $d$. Moreover, for every $x \notin X$ and every class $C$ of type at least $d-1$ we have $|N(x) \cap C| < q$.
\end{lemma}

It may happen that classes of smaller types are large. Our goal in this section is to construct equivalent instances where these classes have bounded size. 
The easiest reduction consists in handling \emph{twins}, which are two vertices $x,y \notin X$ such that $N(x) \setminus \{y \} \subseteq N(y)$ (we note that our definition is not the standard definition). 

\begin{lemma}\label{lem:twin}
If $G$ contains two vertices $x,y \notin X$ such that $N(x) \setminus \{y \} \subseteq N(y)$ then $G$ and $G-x$ are equivalent instances of \dscr. 
In particular, if $G$ contains two twins $x,y\notin X$ then one of them can be removed.
\end{lemma}

\begin{proof}
   Every time a token moves away from $x$ or to $x$ in a reconfiguration sequence in $G$, we can use $y$ instead. Since every neighbor of $x$ is also a neighbor of $y$, the resulting sequence is still a valid reconfiguration sequence in $G'$. 
\end{proof}

In what follows, we assume that $G$ and all upcoming graphs are connected and twin-free. To further reduce the size of $G$, we first need to understand the structure of the classes and of the edges between the classes. The former is quite simple for minor-free graphs; we can assume that they are independent sets, as shown by the following result.

\begin{lemma}\label{lem:basicminorKtt}
Let $G$ be a graph and $G'$ be obtained by contracting a connected component $K$ inside a class $C$ of $G$. Then $G$ and $G'$ are equivalent instances of \dscr.
\end{lemma}

\begin{proof}
If there is a transformation in $G$ for \dscr, we can indeed perform the same in $G'$ by simply omitting a move if it corresponds to moving from a vertex of $K$ to another vertex of $K$. Since all the vertices of $K$ have the same neighborhood in $X$, each intermediate set indeed dominates $X$.

Conversely, assume that we have a transformation in $G'$. We claim that we can adapt it into a transformation for $G$. If we do not move a token from or to the contracted vertex, we perform the same move in $G$. If a token moves to the merged vertex, we simply move in $G$ the corresponding token to one of its neighbors in $K$. Finally, if a token leaves the merged vertex to a vertex $v$ in $G'$, we move a token on $K$ in $G$ to a vertex of $N(v) \cap K$ via a path in $K$, which is possible since $K$ is connected, and then move the token to $v$. Again, since all the vertices of $K$ have the same neighborhood in $C$, all the intermediate sets dominate $X$, which concludes the proof.
\end{proof}

Note that contraction may not preserve membership in $G\in\mathcal{G}_d$ (for example if the initial graph is only $K_{3,d}$-free). However, we only need the $K_{3,d}$-freeness hypothesis to bound the size of classes of type at least $3$, and never use it anymore afterwards.

\subsection{Edges between classes of small type}

Let us start with $0$-classes. 

\begin{lemma}\label{lem:univ0}
    Let $G'$ be the graph obtained from $G$ by adding a vertex $y$ complete to $V \setminus X$ and anticomplete to $X$. $(G',X,D_s,D_t)$ is a yes-instance of \dscr{} if and only if $(G,X,D_s,D_t)$ is a yes-instance of \dscr{}.
\end{lemma}
\begin{proof}
If there was a transformation within $G$, the transformation still exists in $G'$ since we simply add a vertex not in $X$. Assume now that there is a dominating set $D$ in the transformation using the universal vertex $y$. Since $y$ is anticomplete to $X$, then $D \setminus y$ is dominating $X$. Let $v \not\in X$ be the vertex on which the token on $y$ will be slid in the transformation and let $w \not\in X$ be the vertex where the token came from. Since $G$ is connected, we can replace the sequence of moves $w \rightarrow y \rightarrow v$ by a sequence of slides from $w$ to $v$ in $G$.
\end{proof}

Now, if we apply Lemma~\ref{lem:univ0} and then~\ref{lem:twin} to the $0$-class, we remark that the following holds:

\begin{lemma}
\label{lem:class0_new}
 We can assume that the $0$-class has size at most $1$.
\end{lemma}


We now investigate the structure of the edges between the classes. We roughly show that we can assume most of the classes of small type to be linked by at most one edge. This is a consequence of the following.

\begin{lemma}
\label{lem:edgeremoval}
    Let $D$ and $D'$ be two dominating sets that differ by a token slide on the edge $uv$. If $N(u)\cap N(v)\cap X$ contains at most one vertex not in $D$ and $G-uv$ is connected, then there is a sequence of slides  transforming $D$ to $D'$ without using the edge $uv$. 
\end{lemma}

\begin{proof}
    Let $S=D\setminus \{u\}=D'\setminus\{v\}$ be the set of vertices occupied by the tokens that are common to $D$ and $D'$. Note that $S$ dominates all of $X$ except maybe $N(u)\cap N(v)\cap (X\setminus S)$. By hypothesis, this set is either empty (that is, $S$ is already dominating $X$) or contains at most one vertex $z\notin S$. In the former case, $S$ is dominating $X$ so we can freely move the token on $u$ to $v$ following some path in $G-uv$ (which is connected by hypothesis). In the latter case, we can just move the token on the path $uzv$ since all of $X$ is dominated by $S\cup\{z\}$. 
\end{proof}

Using this lemma, we can easily remove edges between classes of type at most $2$.

\begin{lemma}\label{lem:edges2classes}
Let $uv$ be an edge between two classes of type at most $2$. If $G-uv$ is connected, then $G$ and $G-uv$ are equivalent instances of \dscr.
\end{lemma}

\begin{proof}
Note that the removal of edges cannot create new slides, hence any transformation in $G-uv$ is still a transformation in $G$. Conversely, take a transformation between $D_s$ and $D_t$ in $G$ and let us prove that it still exists in $G'$. 

If we do not use the edge $uv$ during the transformation, the conclusion follows, so we can assume that at some step we slide a token from $u$ to $v$. Observe that $N(u)\cap X$ and $N(v)\cap X$ are distinct sets of size at most $2$, hence $N(u)\cap N(v)\cap X$ has size at most $1$. We may thus apply Lemma~\ref{lem:edgeremoval} to provide another transformation that does not slide a token on $uv$, which concludes the proof.
\end{proof}

Since we have added a $0$-vertex adjacent to all the vertices of $V \setminus X$, we know that, if we keep these edges, the graph is connected. So Lemma~\ref{lem:edges2classes} ensures that we can remove almost all the edges pairs of $2$-classes, pairs of $1$-classes and between a $1$-class and a $2$-class. 

Note that in order to preserve the connectivity of the graph, we simply have to keep at most one edge between any two classes of non-$0$ type (since any vertex of $X$ incident to a non-$0$-type class connects all the vertices of the class together). So we can assume that there is at most one edge between every class of size $1$ or $2$. By the same argument, we can assume that the vertex of the $0$-class has at most one neighbor in each class. 






\subsection{Proof of Theorem~\ref{thm:FPT-planar}}

For $K_{3,d}$-free graphs, we already know that classes of type at least $3$ have bounded size by Lemma~\ref{lem:basicKtt}. The previous arguments show that there are basically no edges between the classes of type at most $2$ and that the $0$-class has also bounded size. We now show that every large enough remaining class must contain a pair of twins which can thus be reduced by Lemma~\ref{lem:twin}. This is summarized in the following result.

\begin{lemma}\label{lem:bound2classes}
Let $G$ be an instance of \dscr{} where all classes of type $0$ or at least $3$ have size at most $p>1$. If some class of type $1$ or $2$ has size more than $2^{p\cdot 2^{|X|}}$, then the class contains a vertex $u$ such that $G$ and $G-u$ are equivalent instances of \dscr. 
\end{lemma}

\begin{proof}
Let $C$ be a class of type $1$ or $2$. We show that the vertices of $G$ cannot have many pairwise distinct neighborhoods in $C$. More precisely, we show that $\{N(x)\cap C \mid x\in V(G)\}$ has size at most $p\cdot 2^{|X|}$. In particular, if $C$ has more than $2^{p\cdot 2^{|X|}}$ vertices, then two of them must be twins and $G$ can be reduced by Lemma~\ref{lem:twin}.

Since classes of type $0$ or at least $3$ have size at most $p$, each such class creates at most $p$ neighborhoods in $C$. By Lemma~\ref{lem:edges2classes}, we can assume that there is at most one edge between $C$ and each class of type $1$ or $2$, hence each class creates at most $2$ neighborhoods in $C$. Therefore, in total, each of the $2^{|X|}$ classes creates at most $2$ or $p$ neighborhoods, as needed.
\end{proof}

We are now ready to conclude the proof of Theorem~\ref{thm:FPT-planar}.

\begin{proof}[Proof of Theorem~\ref{thm:FPT-planar}]

Let $G$ be a $K_{3,d}$-free graph having a dominating set of size $k$. By~\cite{eiben2019lossy}, one can find in polynomial time a domination core $X$ of $G$ of size at most $(2d+1)k^{d+1}$. By Lemma~\ref{lem:basicKtt}, all the classes of type at least $3$ have size less than $d$. Moreover, using Lemma~\ref{lem:class0_new}, we can assume that the $0$-class has size $1$. 
Similarly, applying Lemma~\ref{lem:bound2classes}, the classes of size $1$ or $2$ can be reduced to at most $2^{2^{2|X|+1}}$ vertices. 
Now $G$ has size bounded by a function of $k$ and $d$, which concludes the proof (recall that the number of classes is a function of $|X|$). 
\end{proof}

\subsection{Reducing 3-classes (Proof of Theorem~\ref{thm:FPT-minor})}

Observe that Lemmas~\ref{lem:class0_new} and~\ref{lem:bound2classes} are valid for all graphs in $\mathcal{G}_d$. Together, they show that it is enough to bound the size of classes of type at least $3$ to obtain the desired bound. For Theorem~\ref{thm:FPT-planar}, this is an easy consequence of $K_{3,d}$-freeness (see Lemma~\ref{lem:basicKtt}). However, for Theorem~\ref{thm:FPT-minor}, the same claims only bound the size of classes of type at least $4$ and we have no reason to assume that $3$-classes have bounded size. The goal of this section is to deal with those classes.

Let $C,C'$ be two classes. We say that the pair $(C,C')$ is \emph{fat} if there exists a matching larger than $k \cdot d$ between $C$ and $C'$. We start with a few generic results about fat pairs in $K_{q,d}$-minor-free graphs.

\begin{lemma}\label{lem:fat_incl}
Let $C$ be a $(q-1)$-class in a $K_{q,d}$-minor-free graph. If $(C,C')$ is fat then $N(C')\cap X\subsetneq N(C)\cap X$.
\end{lemma}
\begin{proof}
Assume by contradiction that $N(C') \cap X$ contains a vertex $y$ such that $y \notin N(C)$. Contract all edges between $y$ and $C'$. In the resulting graph, $C$ contains at least $d$ vertices, all complete to $N(C)\cap X$ and to the merged vertex. Therefore, $G$ has a $K_{q,d}$-minor, a contradiction.
\end{proof}

\begin{lemma}\label{lem:iffat_incl}
Let $C$ be a $(q-1)$-class in a $K_{q,d}$-minor-free graph. If $(C,C')$ is fat, then every dominating set of size at most $k$ contains a vertex of $N(C') \cap X$. In particular, $C'$ is not the $0$-class.
\end{lemma}

\begin{proof}
Assume by contradiction that there is a dominating set $D$ of size $k$ disjoint from $N(C') \cap X$. Let $M$ be a matching between $C$ and $C'$ of size $kd+1$. By the pigeonhole principle, there exists $u\in D$ that dominates at least $d+1$ endpoints of $M$ in $C'$. If we contract $u$ with these endpoints, the resulting vertex has degree at least $d+1$ in $C$ if $u\notin C$ and at least $d$ if $u\in C$. This is impossible by Lemma~\ref{lem:basicKtt}. 
\end{proof}

Using Lemma~\ref{lem:iffat_incl} with $q=4$, we can reduce $3$-classes for $K_{4,d}$-minor-free graphs.

\begin{lemma}
\label{lem:nofatedge}
Let $C$ be a $3$-class of a $K_{4,d}$-minor-free graph $G$. Let $G'$ be the graph obtained by removing all the edges between $C$ and all classes $C'$ such that $(C,C')$ is fat. Then $G$ and $G'$ are equivalent instances of \dscr.
\end{lemma}

\begin{proof}
Note that removing edges in a graph cannot create new slides. In particular, any reconfiguration sequence in $G'$ directly yields one in $G$. Consider now the converse direction, and take a transformation in $G$ from $D_s$ to $D_t$ for \dscr. By Lemma~\ref{lem:iffat_incl}, all the dominating sets in the sequence must contain a vertex of $N(C') \cap X$. In particular, $C'$ has type at least $1$ hence $G'$ is still connected. Moreover, by Lemma~\ref{lem:fat_incl}, $N(C')\cap X \subsetneq N(C)\cap X$, so $C$ is dominated at each step.

Let us extend this sequence to a transformation in $G'$. 
The only problematic case is when the transformation moves a token between $a \in C$ and $b\in C'$. 
Note that the other tokens dominate all of $X$, except maybe some vertices in $N(a)\cap N(b)\cap X=N(C')\cap X$.
However, since $N(C')\cap X\subsetneq N(C)\cap X$, we have $|N(C')\cap X|\le 2$. 
Moreover, by Lemma~\ref{lem:iffat_incl}, $N(C')\cap X$ contains a token, so at most one vertex $z$ of $N(C')\cap X$ may not be dominated.
In that case, we replace the token slide on $ab$ in $G$ by a slide on the path $azb$ in $G'$; at the intermediate step, $X$ is indeed dominated.
If there is no such $z$, then the tokens not on $\{a,b\}$ already dominate $X$ so we can freely move the token on any path from $a$ to $b$.
\end{proof}

We can now bound the size of $3$-classes in $K_{4,d}$-minor-free graphs.

\begin{corollary}\label{coro:3classes}
Let $G$ be a $K_{4,d}$-minor-free graph. There is an equivalent instance $G'$ where all $3$-classes have size at most $2^{2^{|X|}\cdot (kd+2^{kd})}$.
\end{corollary}

\begin{proof}
By Lemma~\ref{lem:nofatedge}, for every $3$-class $C$, we can safely remove all the edges between $C$ and the classes $C'$ such that $(C,C')$ is fat.

Let $C$ be a $3$-class. Similarly to Lemma~\ref{lem:class0_new}, we show that the vertices of $G$ cannot have many pairwise distinct neighborhoods in $C$ otherwise we find twins in $C$.

Consider a class $C'$ of type at most $3$. Let $M$ be a maximum matching between $C$ and $C'$, which has size at most $kd$. Since $M$ is maximal, the unmatched vertices in $C'$ are only adjacent in $C$ to matched vertices. In particular, they create at most $2^{kd}$ neighborhoods. The other $kd$ matched vertices create at most $kd$ neighborhoods. 

Moreover, recall that by Lemma~\ref{lem:basicKtt}, each class of type at least $4$ contains less than $d$ vertices. Therefore, each of the $2^{|X|}$ classes creates at most ($d$ or) $kd+2^{kd}$ neighborhoods in $C$, as needed.
\end{proof}

We may now conclude the proof of Theorem~\ref{thm:FPT-minor}, similarly to Theorem~\ref{thm:FPT-planar}.

\begin{proof}[Proof of Theorem~\ref{thm:FPT-minor}]
Let $G$ be a $K_{4,d}$-minor-free graph having a dominating set of size $k$. By~\cite{eiben2019lossy}, one can find in polynomial time a domination core $X$ of $G$ of size at most $(2d+1)k^{d+1}$. By Lemma~\ref{lem:basicKtt}, all the classes of type at least $4$ have size less than $d$. Moreover, using Lemma~\ref{lem:class0_new}, we can assume that the $0$-class has size $1$. 
By Corollary~\ref{coro:3classes}, each $3$-class can be reduced to size $q=2^{2^{|X|}\cdot (kd+2^{kd})}$.

We can now apply Lemma~\ref{lem:bound2classes}, to reduce the classes of size $1$ or $2$ to $2^q$ vertices. Now $G$ has size bounded by a function of $k$ and $d$, which concludes the proof.
\end{proof}



\section{Tape reconfiguration for bounded size alphabet is in \textsf{P}}\label{sec-count-tapes-reduce}

In this section, we give a small additional result showing that when the size of the alphabet is constant then the \AutRec{} problem is in \textsf{P}. Recall that for an instance $I$ of \AutRec{}, we are given an alphabet $\Sigma$, a set $\mathcal{T} = \{T_1,\ldots,T_p\}$ of $\Sigma$-tapes, and an initial and final configuration. The goal is to check if it is possible to transform the initial configuration into the final one via a sequence of elementary transformations. 
Let $G = (V,E)$ be the extended graph of $I$, where $V(G) = \bigcup_i V(T_i)\cup\Sigma$ and $xy \in E(G)$ whenever there exists $i$ such that either $xy \in E(T_i)$ or $x$ is a cell of some $T_i$ which contains the letter $y$. Let $n = |V(G)|$ and $m = |E(G)|$.

The key argument consists in constructing equivalent instances with fewer tapes by showing that every instance $I$ of \AutRec{} with more than $2|\Sigma|$ tapes can be modified, in time $2^{\mathcal{O}(|\Sigma|)} \cdot (n + m)^{\mathcal{O}(1)}$, into an equivalent instance $I'$ with at least one fewer tape. We say that $I$ can be \emph{tape reduced}. 
The correctness of this reduction is based on the following lemma.

\begin{lemma}
\label{lem:match}
    Given a set of $|\Sigma|+1$ tapes, one can extract in time $2^{\mathcal{O}(|\Sigma|)} \cdot (n + m)^{\mathcal{O}(1)}$ a subset of $q>0$ tapes such that:
    \begin{enumerate}[label=(\alph*)]
        \item At least $1$ and at most $q - 1$ letters appear in these tapes.
        \item Each such letter can be matched with a cell containing it such that no two selected cells lie on the same tape.
    \end{enumerate}
\end{lemma}

\begin{proof}
    Given a set of tapes, we define its alphabet as the set of letters that appear on at least one cell of these tapes. Consider a non-empty inclusion-wise minimal subset $L$ of tapes whose alphabet has size less than $|L|$. Note that such a subset must exist since the whole set of tapes has size $|\Sigma|+1$ and its alphabet has size at most $|\Sigma|$. Moreover, the alphabet of $L$ cannot be empty, otherwise there is a tape whose contents are all empty, and we can safely delete this tape from the instance. 

    Assume now that the second item is not satisfied. By Hall's theorem, there exists a set $X$ of letters such that less than $|X|$ tapes in $L$ use a letter from $X$. But removing these tapes from $L$ yields a set $L'$ of tapes whose alphabet has size less than $|L|-|X|\leqslant |L'|$, a contradiction with the minimality of $L$.
\end{proof}

Using Lemma~\ref{lem:match}, one can easily derive a worse bound of $3|\Sigma|$ as follows. 

\begin{lemma}
\label{lem:3sigma}
    Every instance of \AutRec{} with more than $3|\Sigma|$ tapes can be tape reduced in time $2^{\mathcal{O}(|\Sigma|)} \cdot (n + m)^{\mathcal{O}(1)}$.
\end{lemma}

\begin{proof}
Let $I$ be an instance of \AutRec{} with more than $3|\Sigma|$ tapes. 
Observe that $\Sigma$ must be covered by all the start positions.    
We may thus select a subset $L_s$ of at most $|\Sigma|$ tapes whose starting cells cover $\Sigma$.
We do the same for the end cells and get a subset $L_e$.
Now, assuming that we have at least $3|\Sigma| + 1$ tapes, we may apply Lemma~\ref{lem:match} to get in time $2^{\mathcal{O}(|\Sigma|)} \cdot (n + m)^{\mathcal{O}(1)}$ a subset $L$ of $0 < q \leq |\Sigma|$ tapes that is disjoint from $L_s \cup L_e$.

We now claim that removing $L$ and all the letters of its alphabet from the labels of the remaining vertices yields an instance $I'$ equivalent to $I$. 
Since we only removed tapes and letters, one can easily extract a reconfiguration sequence for $I'$ from any reconfiguration sequence for $I$. Conversely, assume that there is a reconfiguration sequence for $I'$ and let us describe a sequence for $I$.
We start by moving all the heads on tapes of $L$ towards the cells given by part (b) of Lemma~\ref{lem:match}.
This is possible since the heads on $L_s$ already cover all of $\Sigma$. Now all the letters in the alphabet of $L$ are covered by the heads in $L$, and we may use the reconfiguration sequence for $I'$ to move the heads not in $L$ towards their destination. Finally, since the end cells on $L_e$ cover all of $\Sigma$, we can freely move the heads in $L$ towards their destination. 
\end{proof}

\begin{lemma}
\label{lem:2sigma}
    Every instance of \AutRec{} with more than $2|\Sigma|$ tapes can be tape reduced in time $2^{\mathcal{O}(|\Sigma|)} \cdot (n + m)^{\mathcal{O}(1)}$.
\end{lemma}

\begin{proof}
We now refine the argument to get the $2|\Sigma|$ bound. Without loss of generality, by Lemma~\ref{lem:3sigma}, we assume that the number of tapes is at least $2|\Sigma| + 1$ and at most $3|\Sigma|$. 

We proceed like in the proof of Lemma~\ref{lem:3sigma}, except that we do not extract the set $L_s$.
Instead, we apply Lemma~\ref{lem:match} to all tapes but $L_e$, remove the obtained set $L$ of tapes and their alphabet. Moreover, denoting by $x_1,\ldots,x_q$ the cells given by part (b) of Lemma~\ref{lem:match}, we assume that the sum of distances, over all $i \in [q]$ from cell $x_i$ to the left cell $s_i$ of its tape, is minimized. 
Now the only thing remaining to show is how to move the heads from $\{s_1,\ldots,s_q\}$ to $\{x_1,\ldots,x_q\}$ since afterwards, the arguments of the weaker $3|\Sigma|$ bound carry out verbatim. 

Assume that the alphabet of $L$ is $\{\alpha_1,\ldots,\alpha_q\}$, where each $\alpha_i$ appears in the content of $x_i$. 
Consider the following auxiliary oriented graph; it has one vertex for each $i\in[1,q]$, and an arc $(i,j)$ if $\alpha_i$ appears in tape $j$ on a cell strictly closer to $s_j$ than $x_j$. 
By our choice of $x_1,\ldots,x_q$, this graph must have no directed cycle (such a cycle would yield another choice of $x_i$'s that further minimizes the sum of distances over all $i \in [q]$ from cell $x_i$ to the left cell $s_i$ of its tape), thus up to renumbering, we may assume that all the arcs $(i,j)$ satisfy $i<j$.
In particular, this means that on each tape $j$, the contents of the cells closer to $s_j$ than $x_j$ contain only letters $\alpha_i$ with $i<j$. We may then move the heads from $s_i$ to $x_i$ by considering increasing values of $i$. Indeed, each time we consider tape $i$, the letters appearing on the cells between $s_i$ and $x_i$ are already covered by the heads on $x_1,\ldots,x_{i-1}$. 
\end{proof}

\begin{theorem}
    \AutRec{} is solvable in time $2^{\mathcal{O}(|\Sigma|)} \cdot (n + m)^{\mathcal{O}(|\Sigma|^2)}$. Hence, if the alphabet has constant size then \AutRec{} is polynomial-time solvable.
\end{theorem}

\begin{proof}
By Lemma~\ref{lem:2sigma}, we can, in time $2^{\mathcal{O}(|\Sigma|)} \cdot (n + m)^{\mathcal{O}(1)}$, tape reduce any instance of \AutRec{} so that the number of tapes is at most $2|\Sigma|$.

Having bounded the number of tapes by at most $2|\Sigma|$, we can bound the total number of possible valid configurations by $(n + m)^{\mathcal{O}(|\Sigma|)}$ and the number of edges in the configuration graph by $(n + m)^{\mathcal{O}(|\Sigma|^2)}$. A simple brute-force approach now suffices to solve the problem in time $2^{\mathcal{O}(|\Sigma|)} \cdot (n + m)^{\mathcal{O}(|\Sigma|^2)}$, as needed. 
\end{proof}


\end{document}